\documentclass[numbib]{imaiaiai}

\usepackage{xcolor}
\usepackage{url}

\usepackage{enumitem}
\usepackage{color}
\usepackage{mathtools}
\usepackage{bbm}
\usepackage{caption}
\usepackage{subcaption}

\definecolor{linkblue}{rgb}{0,0,0.7}
\definecolor{citeblue}{rgb}{0,0.3,0.5}
\usepackage[colorlinks = true, linkcolor = linkblue, urlcolor=blue, citecolor = citeblue]{hyperref}

\usepackage{amsmath,amssymb,amsthm}
\usepackage{graphicx}

\def\lspd{$(\mathrm{LS}_{\tau}^{*})$}
\def\qppd{$(\mathrm{QP}_{\lambda}^{*})$}
\def\bppd{$(\mathrm{BP}_{\sigma}^{*})$}
\def\ls{$(\mathrm{LS}_{\tau, K})$}
\def\qp{$(\mathrm{QP}_{\lambda, K})$}
\def\bp{$(\mathrm{BP}_{\sigma, K})$}

\DeclareMathOperator*{\argmin}{\arg\min}
\DeclareMathOperator{\1}{\mathbbm{1}}
\DeclareMathOperator{\sgn}{\mathrm{sgn}}
\DeclareMathOperator{\prox}{\mathrm{prox}}

\newcommand{\iid}{\ensuremath{\overset{\text{iid}}{\sim}}}
\newcommand{\E}{\mathbb E}

\newcommand{\reals}{\ensuremath{\mathbb{R}}}
\newcommand{\nats}{\ensuremath{\mathbb{N}}}

\newcommand{\sph}{\ensuremath{\mathbb{S}}}
\newcommand{\dee}[2]{\ensuremath{\frac{\mathrm{d}#1}{\mathrm{d}#2}}}

\newcommand{\ip}[1]{\ensuremath{\langle #1\rangle}}

\renewcommand{\d}{\ensuremath{\,\mathrm{d}}}

\newcommand{\ie}{\emph{i.\@e.\@,~}}
\newcommand{\eg}{\emph{e.\@g.\@,~}}

\setlist[enumerate, 1]{align=left, label=\arabic*.}
\setlist[enumerate, 2]{align=left, label=(\alph*)}
\setlist[itemize,1]{align=left, topsep=-2pt}

\date{\normalsize\today}

\AtBeginDocument{
  \label{CorrectFirstPageLabel}
  
}

\begin{document}

\title{\Large Sensitivity of $\ell_{1}$ minimization to parameter choice}

\shorttitle{Sensitivity of $\ell_{1}$ minimization to parameter choice}
\shortauthorlist{Berk, A., Plan, Y., Yilmaz, \"O}
  
\author{%
  {\sc Aaron Berk}$^{*}$\\[2pt]
  Dept. Mathematics, University of British Columbia\\
  $ˆ*$\email{Corresponding author: aberk$@$math.ubc.ca}\\[6pt]
  {\sc Yaniv Plan}\\[2pt]
  Dept. Mathematics, University of British Columbia\\
  {yaniv$@$math.ubc.ca}\\[6pt]
  {\sc and}\\[6pt]
  {\sc \"Ozg\"ur Yilmaz} \\[2pt]
  Dept. Mathematics, University of British Columbia\\
  {oyilmaz$@$math.ubc.ca}}

\maketitle

\begin{abstract}
  {The use of generalized \textsc{Lasso} is a common technique for recovery of
    structured high-dimensional signals. Each generalized \textsc{Lasso}
    program has a governing parameter whose optimal value depends on properties
    of the data. At this optimal value, compressed sensing theory explains why
    \textsc{Lasso} programs recover structured high-dimensional signals with
    minimax order-optimal error. Unfortunately in practice, the optimal choice
    is generally unknown and must be estimated. Thus, we investigate stability
    of each \textsc{Lasso} program with respect to its governing parameter. Our
    goal is to aid the practitioner in answering the following question:
    \emph{given real data, which \textsc{Lasso} program should be used?} We
    take a step towards answering this by analyzing the case where the
    measurement matrix is identity (the so-called proximal denoising setup) and
    we use $\ell_{1}$ regularization. For each \textsc{Lasso} program, we
    specify settings in which that program is provably unstable with respect to
    its governing parameter. We support our analysis with detailed numerical
    simulations. For example, there are settings where a 0.1\% underestimate of
    a \textsc{Lasso} parameter can increase the error significantly; and a 50\%
    underestimate can cause the error to increase by a factor of $10^{9}$.}
  {Parameter instability, Sparse proximal denoising, \textsc{Lasso}, Compressed
    sensing, Convex optimization}
\end{abstract}

\section{Introduction}
\label{sec:introduction}

A fundamental problem of signal processing centers the development and analysis
of efficacious methods for structured signal recovery that are widely applicable
in practice. Frequently in applications, the signal is assumed to be structured
according to some data model and measured by a particular acquisition
method. For example, in image deblurring one might assume the objects of
interest lie in the dual of a Besov space \cite{meyer2001oscillating, garnett2007image}, while in
MRI applications, one might assume the images are sparse in a wavelet domain,
and measured by subsampling their Fourier coefficients
\cite{lustig2007sparse}. There is extensive literature concerned with those
applications in which the goal is to recover the ground-truth signal from
acquired measurements by a prescribed convex program that exploits the signal
structure. For example, compressed sensing (CS) has demonstrated that a
scale-invariant structure such as sparsity can be captured by convex
optimization.

The above paradigm can be put in the following mathematical language. Assume
that $K \subseteq \reals^{N}$ is a nonempty closed and convex set. Denote the
gauge of $K$ by $\|x\|_{K} := \inf \{\lambda > 0 : x \in \lambda K \}$ and
observe that $\|\cdot\|_{K}$ may be a norm for certain choices of $K$. Assume
that a signal $x_{0} \in \reals^{N}$ is ``structured'' in the sense that
$\|x_{0}\|_{K}$ is relatively small. Suppose $A \in \reals^{m \times N}$
defines the linear measurement process and define the measurements
$y = Ax_{0} + \eta z$ where $z \in \reals^{m}$ is a possibly stochastic noise
vector with noise level $\eta > 0$. Here, $1 \leq m, N < \infty$ are integers
and we do not yet make an assumption on the relative size of $m$ and $N$. For
$\tau, \sigma, \lambda > 0$, we define the following three \emph{generalized
\textsc{Lasso}} programs, which are convex, where the goal is to best
approximate the original signal $x_{0}$.
\begin{align}
  \label{eq:ls-tau-K}
  \hat x (\tau; y, A, K) %
  &:= \argmin \Big\{ \|y - A x\|_{2} : x \in \tau K \Big\}
    \tag{$\mathrm{LS}_{\tau, K}$}
  \\
  \label{eq:qp-lambda-K}
  x^{\sharp}(\lambda; y, A, K) %
  &:= \argmin \Big\{ \frac12 \|y - Ax\|_{2}^{2} + \lambda \|x\|_{K} : %
    x \in \reals^{N} \Big\}%
    \tag{$\mathrm{QP}_{\lambda, K}$}
  \\
  \label{eq:bp-sigma-K}
  \tilde x(\sigma; y, A, K) %
  &:= \argmin \Big\{ \|x\|_{K} : \|y - Ax \|_{2} \leq \sigma \Big\}
    \tag{$\mathrm{BP}_{\sigma, K}$}
\end{align}
For brevity of notation, when it is clear from context, we omit explicit
dependence of $\hat x, \tilde x, x^{\sharp}$ on $y, A$ and $K$. We include below
several examples of this general set-up:
\begin{enumerate}

\item To obtain total variation (TV) denoising for [continuous-valued discrete]
  images, define for $x \in \reals^{N\times N}$,
  \begin{align*}
    \|x\|_{\mathrm{BV}} %
    := \|x\|_{1} + \sum_{\alpha \in [N]^{2}} \sum_{\beta \in \nu(\alpha)}
    |x_{\alpha} - x_{\beta}|,
  \end{align*}
  where $[N] = \{1, 2, \ldots, N\}$ and
  $\nu : [N]^{2} \to \mathcal{P}([N]^{2})$ is the neighbour map that determines
  which ``pixels'' $x_{\beta}$ of the image are the neighbours of the pixel
  $x_{\alpha}$. If $\alpha = (i,j)$ and $2 \leq i,j \leq N-1$ then one
  typically has $\nu(i,j) = \{ (i-1, j), (i, j-1), (i+1,j), (i, j+1)\}$ with a
  variety of choices for the remaining indices. So defined,
  $x^{\sharp}(\lambda; y, I, K)$ is a well-known denoising model for
  two-dimensional images when $A = I$ is the identity matrix and
  $K := \{ \|x\|_{\mathrm{BV}} \leq 1 \}$ \cite{rudin1992nonlinear}. Instead
  defining
  $\|x\|_{\mathrm{BV}} := \|x\|_{1} + \sum_{i = 1}^{N-1} |x_{i+1} - x_{i}|$ for
  $x \in \reals^{N}$, one obtains an equivalent denoising method for
  one-dimensional signals. With minor modification of $x^{\sharp}(\lambda)$ to
  allow for $A$ to act as a bounded linear operator on
  $x \in \reals^{N\times N}$ (\eg convolution with a Gaussian kernel), one may
  extend the model for image deblurring \cite{daubechies2005variational}.

\item Say that $x \in \reals^{N}$ is $s$-sparse if
  $x \in \Sigma_{s}^{N} := \{ x \in \reals^{N} : \|x\|_{0} \leq s\}$ where
  $\|x\|_{0} = \#\{j : x_{j} \neq 0 \}$. Define $K := B_{1}^{N}$, suppose
  $x_{0} \in \reals^{N}$ is $s$-sparse for some $s \geq 1$ and suppose that
  $A \in \reals^{m\times N}$ is a Gaussian random matrix with
  $A_{ij} \iid \mathcal{N}(0, m^{-1/2})$. Then we obtain three common variants
  of the \textsc{Lasso} that solve the ``vanilla'' CS problem:
  the constrained \textsc{Lasso} yielding $\hat x(\tau; y, A, K)$, basis pursuit
  denoise yielding $\tilde x(\sigma; y, A, K)$, and the unconstrained
  \textsc{Lasso} yielding $x^{\sharp}(\lambda; y, A, K)$.

\item When $A = I$ is the identity matrix, $(\mathrm{LS}_{\tau, K})$ yields the
  orthogonal projection onto $\tau K$, which we denote by
  $\mathrm{P}_{\tau K}(y) := \hat x (\tau; y, I, K)$. Similarly,
  $(\mathrm{QP}_{\lambda, K})$ yields the proximal operator for the gauge
  induced by $K$, which we denote by 
  $\prox_{\lambda^{-1}K}(y) := x^{\sharp}(\lambda; y, I, K)$. Proximal operators
  are the workhorses of proximal algorithms. Projected gradient descent methods
  rely on $\mathrm{P}_{\tau K}(y)$, while $\prox_{\lambda^{-1}K}(y)$ is central
  to proximal gradient descent methods. 

\item For example, suppose $y = \Phi x_{0} + \eta z$ where $x_{0}$ is
  $s$-sparse, $\Phi\in \reals^{m \times N}$ is a Gaussian random matrix with
  $m \ll N$ and $\eta z$ is scaled normal random noise. A well-known way of
  solving for $\hat x(\tau; y, \Phi, B_{1}^{N})$ where $B_{1}^{N}$ is the unit
  $\ell_{1}$ ball, is to compute the following projected gradient descent
  scheme:
  \begin{align*}
    x^{t + 1} %
    := \mathrm{P}_{\tau B_{1}^{N}} ( x^{t} - \mu^{t}\nabla \|\Phi x - y\|_{2}^{2}).
  \end{align*}

\item Assume that $x' \in \reals^{N}$ is $s$-sparse and let
  $x_{0} = \Psi^{-1} x'$ where $\Psi$ is the orthonormal DFT matrix. Given
  $y = x_{0} + \eta z$, the vector $\hat x(\tau; y, \Psi^{-1}, B_{1}^{N})$ gives
  an analogue of running so-called constrained proximal denoising in Fourier
  space.

\item Consider a matrix $x \in \reals^{N\times N}$, let $\|x \|_{*}$ denote its
  nuclear norm and define
  $K := \{ x \in \reals^{N\times N} : \|x\|_{*} \leq 1 \}$. Then
  $\tilde x(\sigma)$ gives the standard optimization program for recovering a
  low-rank matrix $x_{0}\in \reals^{N \times N}$ from measurements
  $Ax := \ip{A_{i}, x} = \sum_{\alpha \in [N]^{2}} A_{i, \alpha} x_{\alpha}$.

\end{enumerate}
In both the second and final examples, the signal $x_{0}$ does not (necessarily)
belong to the structure set $K$. Instead, $K$ serves as a kind of structural
proxy. To clarify, $K = B_{1}^{N}$ in the second example, which is a structural
proxy for sparse vectors in the sense that if $x \in \reals^{N}$ is $s$-sparse
then $\|x\|_{1} / \|x\|_{2}$ is relatively small compared to non-sparse vectors. A similar statement holds for low-rank matrices and the nuclear norm, as in the final example.

Because of the myriad applications of this class of programs to real-world
problems, it is imperative to fully characterize the performance and stability
of these algorithms. For example, the error rates of $\hat x(\tau)$ are
well-known when $\tau$ is equal to the optimal parameter choice, $A$ is a
subgaussian random matrix and $K$ is a symmetric, closed, convex set containing
the origin \cite{foucart2013mathematical, liaw2017simple,
  oymak2013squared}. However, the error of the estimator $\hat x(\tau)$ is not
fully characterized in this setting for values of $\tau$ that are not the
optimal choice. Similarly, there lacks a full comparison of the error behaviour
between the three estimators $\hat x(\tau), \tilde x(\sigma)$ and
$x^{\sharp}(\lambda)$ as a function of their governing parameters. It is an
open question if there are settings in which one estimator is always preferable
to another.


Perhaps the most common example of where these programs are used is CS. CS is a
provably stable and robust technique for simultaneous data acquisition and
dimension reduction \cite{foucart2013mathematical}. Take the linear measurement
model $y = Ax_{0}$, where $x_{0} \in \reals^{N}$ is $s$-sparse. The now
classical CS result \cite{candes2006robust, candes2006stable, candes2006near,
  davenport2011introduction, donoho2006compressed, foucart2013mathematical}
shows if $A$ is suitably random and has $m \geq C s\log (N/s)$ rows, then one
may efficiently recover $x_{0}$ from $(y, A)$. Numerical implementations of CS
are commonly tied to one of three convex $\ell_{1}$ programs: constrained
\textsc{Lasso}, unconstrained \textsc{Lasso}, and quadratically constrained
basis pursuit \cite{van2008probing}. The advent of suitable fast and scalable
algorithms has made the associated family of convex $\ell_{1}$ minimization
problems extremely useful in practice \cite{friedlander2014gauge,
  friedman2010regularization, park2007l1, van2008probing}.

Proximal Denoising (PD) is a simplification of its more general CS counterpart,
in which the measurement matrix is identity. PD uses convex optimization as a
means to recover a structured signal corrupted by additive noise. We define
three convex programs for PD: constrained proximal denoising, basis pursuit
proximal denoising, and unconstrained proximal denoising. To bear greatest
relevance to CS, we assume that $x_{0}$ is $s$-sparse, having no more than $s$
non-zero entries, and that $y = x_{0} + \eta z$, where $z \iid \mathcal{N}(0,1)$
and $\eta > 0$.  For $\tau, \sigma, \lambda > 0$, respectively,
\begin{align}
  \label{eq:lspd}
  \hat x(\tau) %
  &:= \argmin_{x\in \reals^{N}} \big\{ \|y - x\|_{2}^{2} : \|x\|_{1} \leq \tau \big\} %
    \tag{$\mathrm{LS}_{\tau}^{*}$}
  \\
  \label{eq:bppd}
  \tilde x(\sigma) %
  &:= \argmin_{x\in \reals^{N}} \big\{ \|x\|_{1} : \|y - x\|_{2}^{2} \leq \sigma^{2} \big\} %
    \tag{$\mathrm{BP}_{\sigma}^{*}$}
  \\
  \label{eq:qppd}
  x^{\sharp}(\lambda) %
  &:= \argmin_{x\in \reals^{N}} \big\{ \frac12 \|y - x\|_{2}^{2} + \lambda \|x\|_{1} \big\}. %
    \tag{$\mathrm{QP}_{\lambda}^{*}$}
\end{align}
These are clear simplifications of $(\mathrm{LS}_{\tau, K})$,
$(\mathrm{QP}_{\lambda, K})$ and $(\mathrm{BP}_{\sigma, K})$ introduced above,
in which $K = B_{1}^{N}$ is the $\ell_{1}$ ball and where we use ${}^{*}$ to
denote that the measurement matrix $A \in \reals^{N \times N}$ is identity.

Following the dicussion above, minimax order-optimal recovery results for CS
and PD programs rely on the ability to make a specific choice of the program's
governing parameter (\ie ``using an oracle'')
\cite{foucart2013mathematical}. However, the optimal choice of the governing
parameter for these programs is generally unknown in practice. Consequently, it
is desirable that the error of the solution exhibit stability with respect to
variation of the parameter about its optimal setting. If the optimal choice of
parameter yields order-optimal recovery error, then one may hope that a
``nearly'' optimal choice of parameter admits ``nearly'' order-optimal recovery
error, too, in the sense that the discrepancy in error is no greater than a
multiplicative constant that depends smoothly on the discrepancy in parameter
choice. For example, if $R(\alpha)$ is the mean-squared error of a convex
program with parameter $\alpha > 0$, and $\alpha^{*} > 0$ is the value yielding
minimal error, then one may hope for smooth dependence on $\alpha$, such as
\begin{align*}
  R(\alpha) %
  \lesssim A(\alpha) R(\alpha^{*}),
\end{align*}
where $A: \reals \to \reals^{+}$ is a nonnegative smooth function with
$A(\alpha^{*}) = 1$. For example, the risk for {\qppd} satisfies this
expression with $A(\lambda) = (\lambda / \lambda^{*})^{2}$ when
$\lambda \geq \lambda^{*}$.

Unfortunately, such a hope cannot be guaranteed in general. We prove the
existence of regimes in which PD programs exhibit \emph{parameter instability}
--- small changes in parameter values can lead to blow-up in risk. Moreover,
since the three versions of PD are equivalent in a sense (\emph{cf.}
Proposition \ref{prop:pd-equivalence}), one might think it does not matter
which to choose in practice. However, in this paper we demonstrate regimes in
which one program exhibits parameter instability, while the other two do
not. For example, in the very sparse regime, our theory and simulations suggest
not to use {\bppd}, while in the low-noise regime, they suggest not to use
{\lspd}. At the same time, we identify situations where PD programs perform
well in theory and simulations alike.

We explore the connection between PD and CS numerically, observing that our
theoretical results for PD are mirrored in the CS setup. This holds in both
completely synthetic experiments, and for a more realistic example using the
Shepp-Logan phantom. Thus, the theoretical results in this paper can help
practitioners decide which program to use in CS problems with real data.

\section{Summary of results to follow}
\label{sec:results-summary}

This section contains three sibling results that simplify the main results in
the next sections by considering asymptotic versions of them. By
``risk'', we mean the noise-normalized expected squared error (nnse) of an
estimator. The risks for the estimators $\hat x(\tau), x^{\sharp}(\lambda)$ and
$\tilde x(\sigma)$ are, respectively:
\begin{align*}
  \hat R(\tau ; x_{0}, N, \eta) %
  &:= \eta^{-2} \E \|\hat x (\tau) - x_{0} \|_{2}^{2},
  \\[.2cm] 
  R^{\sharp}(\lambda; x_{0}, N, \eta) %
  &:= \eta^{-2} \E \|x^{\sharp}(\eta\lambda) - x_{0} \|_{2}^{2},
  \\[.2cm] 
  \tilde R(\sigma; x_{0}, N, \eta) %
  &:=  \eta^{-2} \E \|\tilde x(\sigma) - x_{0}\|_{2}^{2}.
\end{align*}
Denote $\Sigma_{s}^{N} := \{ x\in \reals^{N} : \|x\|_{0} \leq s\}$ where
$\|x\|_{0}$ gives the number of non-zero entries of $x$; it is not a
norm. Denote by $R^{*}(s, N)$ the following optimally tuned worst-case risk for
{\lspd}:
\begin{align*}
  R^{*}(s,N) %
  := \sup_{x_{0} \in \Sigma_{s}^{N}} %
  \hat R(\|x_{0}\|_{1} ; x_{0}, N, \eta) %
  = \max_{\substack{x_{0} \in \Sigma_{s}^{N}\\\|x_{0}\|_{1}=1}}%
  \lim_{\eta\to 0} \hat R(1; x_{0}, N, \eta).
\end{align*}
A proof of the second equality appears in Proposition \ref{prop:risk-equivalence}. We
use $R^{*}(s,N)$ as a benchmark, noting it is order-optimal in
Proposition \ref{prop:lspd-optimal-risk}.

In section \ref{sec:param-inst-lspd}, we show that {\lspd} exhibits an asymptotic
phase transition in the low-noise regime. There is exactly one value $\tau^{*}$
of the governing parameter yielding minimax order-optimal error, with any
choice $\tau \neq \tau^{*}$ yielding markedly worse behaviour. The intuition for
this result is that {\lspd} is extremely sensitive to the value of $\tau$ in the
low-noise regime, making empirical use of {\lspd} woefully unstable in this
regime.
\begin{theorem}
  \label{thm:simplified-lspd}
  \begin{align*}
    \lim_{N\to\infty} %
    \max_{\substack{x_{0} \in \Sigma_{s}^{N}\\\|x_{0}\|_{1} = 1}} %
    \lim_{\eta \to 0} %
    \frac{\hat R(\tau; x_{0}, N, \eta)}{R^{*}(s, N)} %
    & =
      \begin{cases}
        \infty & \tau < \tau^{*}\\
        1 & \tau = \tau^{*} = 1 \\
        \infty & \tau > \tau^{*}
      \end{cases}
  \end{align*}
\end{theorem}

Next, in section \ref{sec:param-inst-qppd}, we show that {\qppd} exhibits an
asymptotic phase transition. The worst-case risk over $x_{0} \in \Sigma_{s}^{N}$
is minimized for parameter choice $\lambda^{*} = O(\sqrt{\log(N/s)})$
\cite{oymak2016sharp}.  While $\lambda^{*}$ has no closed form expression, it
satisfies $\lambda^{*}/\sqrt{2\log(N)} \xrightarrow{N\to\infty} 1$ for $s$ fixed
(Proposition \ref{prop:asymptotic-equivalence}).  Thus, we consider the normalized
parameter $\mu = \lambda/\sqrt{2 \log(N)}$.  The risk
$R^{\sharp}(\lambda; x_{0}, N, \eta)$ is minimax order-optimal when $\mu >1$ and
suboptimal for $\mu <1$.
\begin{theorem}
  Let $\lambda(\mu, N) := \mu \sqrt{2 \log N}$ for $\mu > 0$. Then, 
  \label{thm:simplified-qppd}
  \begin{align*}
    \lim_{N\to\infty} \sup_{x_{0} \in \Sigma_{s}^{N}} %
    \frac{R^{\sharp}(\lambda(\mu, N); x_{0}, N, \eta)}{R^{*}(s, N)} %
    & =
      \begin{cases}
        O(\mu^{2}) & \mu \geq 1\\
        \infty & \mu < 1
      \end{cases} 
  \end{align*}
\end{theorem}

Lastly, we show in section \ref{sec:param-inst-bppd} that {\bppd} is poorly behaved
for all $\sigma > 0$ when $x_{0}$ is very sparse. Namely,
$\tilde R(\sigma; x_{0}, N, \eta)$ is asymptotically suboptimal for \emph{any}
$\sigma > 0$ when $s/N$ is sufficiently small.
\begin{theorem}
  \label{thm:simplified-bppd}
  \begin{align*}
    {\adjustlimits\lim_{N\to\infty} \sup_{x_{0} \in \Sigma_{s}^{N}} \inf_{\sigma > 0}}
    \frac{\tilde R(\sigma ; x_{0}, N, \eta)}{R^{*}(s, N)} %
    &= \infty
  \end{align*}
\end{theorem}

All numerical results are discussed in section \ref{sec:numerical-results}, and
proofs of most theoretical results are deferred to section \ref{sec:proofs}. Next,
we add two clarifications. First, the three PD programs are equivalent in a
sense.
\begin{proposition}
  \label{prop:pd-equivalence}
  Let $0 \neq x_{0} \in \reals^{N}$ and $\lambda > 0$. Where
  $x^{\sharp}(\lambda)$ solves {\qppd}, define
  $\tau := \|x^{\sharp}(\lambda)\|_{1}$ and
  $\sigma := \|y - x^{\sharp}(\lambda)\|_{2}$. Then $x^{\sharp}(\lambda)$ solves
  {\lspd} and {\bppd}.
\end{proposition}
However, $\tau$ and $\sigma$ have \emph{stochastic} dependence on $z$, and this
mapping may not be smooth. Thus, parameter stability of one program is not
implied by that of another. Second, $R^{*}(s, N)$ has the desirable property
that it is computable up to multiplicative constants. The proof follows by
\cite{oymak2016sharp} and standard bounds in \cite{foucart2013mathematical}. We
don't claim novelty for this result, and defer its full proof to
section \ref{sec:proof-lspd-optimal-risk}.
\begin{proposition}
  \label{prop:lspd-optimal-risk}
  Let $s \geq 1, N \geq 2$ be integers, let $\eta > 0$ and suppose
  $y = x_{0} + \eta z$ for $z \in \reals^{N}$ with
  $z_{i} \iid \mathcal{N}(0,1)$. Let
  $M^{*}(s, N) := \inf_{x_{*}} \sup_{x_{0} \in \Sigma_{s}^{N}} \eta^{-2} \|x_{*}
  - x_{0}\|_{2}^{2}$ be the minimax risk over arbitrary estimators
  $x_{*} = x_{*}(y)$. There is $c, C_{1}, C_{2} > 0$ such that for
  $N \geq N_{0} = N_{0}(s)$, with $N_{0} \geq 2$ sufficiently large,
  \begin{align*}
    cs\log(N/s) \leq M^{*}(s, N) \leq \adjustlimits \inf_{\lambda > 0} \sup_{x_{0} \in \Sigma_{s}^{N}} %
    R^{\sharp}(\lambda; x_{0}, N, \eta) %
    \leq C_{1} R^{*}(s, N) %
    \leq C_{2} s\log(N/s).
  \end{align*}
\end{proposition}
Thus, in the simplified theorems above, we could have normalized by any of the
above expressions instead of $R^{*}(s,N)$, because all three expressions are
asymptotically equivalent up to constants. In contrast, a consequence of
Proposition \ref{prop:lspd-optimal-risk} using Theorem \ref{thm:simplified-bppd} is that
\begin{align*}
  \adjustlimits \inf_{\sigma > 0} \sup_{x_{0} \in \Sigma_{s}^{N}} %
  \tilde R(\sigma; x_{0}, N, \eta) %
  \geq \adjustlimits \sup_{x_{0} \in \Sigma_{s}^{N}} \inf_{\sigma > 0} %
  \tilde R(\sigma; x_{0}, N, \eta)
  \gg R^{*}(s, N).
\end{align*}
In particular, removing the parameters' noise dependence destroys the
equivalence attained in Proposition \ref{prop:pd-equivalence}.

\subsection{Related work}
\label{ssec:related-work}

PD is a simple model that elucidates crucial properties of models in general
\cite{elad2012sparse}. As a central model for denoising, it lays the groundwork
for CS, deconvolution and inpainting problems \cite{elad2010role}. A fundamental
signal recovery phase transition in CS is predicted by geometric properties of
PD \cite{amelunxen2014living}, because the minimax risk for PD is equal to the
statistical dimension of the signal class \cite{oymak2016sharp}. This quantity
is a generalized version of $R^{*}(s, N)$ introduced above.

Robustness of PD to inexact information is discussed briefly in
\cite{oymak2016sharp}, wherein sensitivity to constraint set perturbation is
quantified, including an expression for right-sided stability of unconstrained
PD. Essentially, PD programs are proximal operators, a powerful tool
in convex and non-convex optimization \cite{bolte2010alternating,
  combettes2011proximal}. For a thorough treatment of proximal operators and
proximal point algorithms, we refer the reader to \cite{bertsekas2003convex,
  eckstein1992douglas, rockafellar1976monotone}. Thus is PD interesting in its
own right, as argued in \cite{oymak2016sharp}.

Equivalence of the above programs is illuminated from several perspectives
\cite{bertsekas2003convex, van2008probing, oymak2016sharp}. PD risk is
considered with more general convex constraints \cite{chatterjee2014new}. A
connection has been made between the risk of Unconstrained \textsc{Lasso} and
$R^{\sharp}(\lambda; x_{0}, N, \eta)$ \cite{bayati2012lasso,
  bayati2011dynamics}. In addition, there are near-optimal error bounds for
worst-case noise demonstrating that equality-constrained basis pursuit
($\sigma = 0$) performs well under the noisy CS model ($\eta \neq 0$)
\cite{wojtaszczyk2010stability}. It should be noted that these results do not
contradict those of this work, as random noise can be expected to perform better
than worst-case noise in general. Recently, a bound on the unconstrained
\textsc{Lasso} MSE has been proven, which is uniform in $\lambda$ and uniform in
$x_{0} \in B_{p}^{N}$ \cite[Thm 3.2]{miolane2018distribution}. Note that this
also does not run contrary to the left-sided parameter instability result
mentioned above as the uniformity in $\lambda$ is over a pre-specified interval
chosen independently of the optimal parameter choice $\lambda^{*}$, and the assumption on signal structure is different.



\subsection{Notation}
\label{ssec:notation}

We use the standard notation for the Euclidean $p$ norm, $\|\cdot\|_{p}$, for
values $p \geq 1$, and occasionally make use of the overloaded notation
$\|x\|_{0} := \# \{ i \in [N] : x_{i} \neq 0\}$ to denote the number of nonzero
entries of a vector $x$. Let $N \in \nats$ be an integer representing
dimension. Let $x_{0} \in \Sigma_{s}^{N} \subseteq \reals^{N}$ be an $s$-sparse
signal with support set $T \subseteq [N] := \{1, 2, \ldots, N\}$, where
$s \ll N$ and
$\Sigma_{s}^{N} := \{ x \in \reals^{N} : 0 \leq \|x\|_{0} \leq s \}$ denotes the
set of $s$-sparse vectors. We use $x$ or $x'$ to denote an arbitrary $s$-sparse
signal, whereas $x_{0}$ denotes the signal for a given problem. Let
$z \in \reals^{N}$ be a normal random vector with covariance matrix equal to the
identity, $z_{i} \iid \mathcal{N}(0,1)$. Denote by $\eta \in (0,1)$ the standard
deviation and suppose $y = x_{0} + \eta z$. Moreover, let
$Z \sim \mathcal{N}(0,1)$ denote a standard normal random variable. Denote
$B_{p}^{N} := \{ x\in\reals^{N} : \|x\|_{p} \leq 1\}$ the standard $\ell_{p}$
ball and for a set $\mathcal{C}\subseteq \reals^{N}$, denote by
$\gamma \mathcal{C} := \{ \gamma x : x \in \mathcal{C}\}$ the scaling of
$\mathcal{C}$ by $\gamma$. All additional notation shall be introduced in
context.






\section{Main theoretical tools}
\label{sec:theoretical-tools}

In this section, we synthesize several known results from convex analysis and
probability theory, some with proof sketches to provide intuition. We outline
notation to refer to common objects from convex analysis. We introduce two
well-known tools for characterizing the \emph{effective dimension} of a set, and
state a result that connects these tools with PD estimators
\cite{oymak2016sharp}. We state a projection lemma that introduces a notion of
ordering for projection operators. To our knowledge, this final result in
\ref{sssec:projection-lemma} novel. In \ref{sssec:refin-bounds-gauss} we
state two recent results giving refined bounds on the Gaussian mean width of
convex polytopes intersected with Euclidean balls \cite{bellec2017localized}.

\subsection{Tools from convex analysis}
\label{ssec:tools-from-convex}

Let $f: \reals^N \to \reals$ be a convex function and let $x \in
\reals^N$. Denote by $\partial f(x)$ the subdifferential of $f$ at the point
$x$,
\begin{align*}
  \partial f (x) := \{ v \in \reals^N : \forall y, f(y) \geq f(x) + \ip{v, y-x} \}
\end{align*}
Note that $\partial f(x)$ is a nonempty, convex and compact set. Given $A
\subseteq \reals^N$ and $\lambda > 0$, denote
\begin{align*}
  \lambda A &:= \{ \lambda a : a \in A\},
  & 
  \mathrm{cone}(A) &:= \{ \lambda x : x \in A, \lambda \geq 0\}.
\end{align*}
For a nonempty set $\mathcal{C}$ and $x\in \reals^N$, denote the distance of $x$
to $\mathcal{C}$ by
$\mathrm{dist}(x, \mathcal{C}) := \inf_{w \in \mathcal{C}} \|x - w\|_2$. If
$\mathcal{C}$ is also closed and convex, then there exists a unique point in
$\mathcal{C}$ attaining the minimum, denoted
\begin{align*}
  \mathrm{P}_{\mathcal{C}}(x) := \argmin_{w \in \mathcal{C}} \|x - w\|_2.
\end{align*}
Denote by $C^{\circ} := \{ v \mid \forall x \in \mathcal{C}, \ip{v,x} \leq 0 \}$
the \emph{polar cone} of $\mathcal{C}$; and define the statistical dimension
\cite{amelunxen2014living} of $\mathcal{C}$ by
\begin{align*}
  \mathbf{D}(\mathcal{C}) := \E [\mathrm{dist}(g, \mathcal{C})^2], \quad g \sim
  \mathcal{N}(0, I_N)
\end{align*}
The descent set of a non-empty convex set $\mathcal{C}$ at a point
$x \in \reals^{N}$ is given by
$F_{\mathcal{C}}(x) := \{ h : x + h \in
\mathcal{C}\}$. The tangent cone is given by
$T_{\mathcal{C}}(x) := \mathrm{cl}(\mathrm{cone}(F_{\mathcal{C}}(x)))$ where
$\mathrm{cl}$ denotes the closure operation; it is the smallest closed cone
containing the set of feasible directions. With these tools, we recall the
result of \cite{oymak2016sharp} in the PD context, giving a precise
characterization of the risk for {\lspd}.

\begin{theorem}[{\cite[Theorem 2.1]{oymak2016sharp}}]
  \label{thm:hassibi-2-1}
  Let $\mathcal{C}$ be a non-empty closed and convex set, let $x \in C$ be an
  arbitrary vector and assume that $z \sim \mathcal{N}(0, I_N)$. Then
  \begin{align}\label{eq:oymakhassibi}
    \sup_{\eta > 0} \frac{1}{\eta^2} \E \|\mathrm{P}_\mathcal{C}(x + \eta z)
    - x\|_2^2 = \mathbf{D}(T_{\mathcal{C}}(x)^\circ).
  \end{align}
\end{theorem}

In that work, the authors note
$\mathbf{D}(T_{\mathcal{C}}(x)^\circ) \approx w^{2}\big(T_{\mathcal{C}}(x) \cap
B_{2}^{N} \big)$, where $w(\cdot)$ denotes the Gaussian mean
width. Specifically, Gaussian mean width gives a near-optimal characterization
of the risk for {\lspd}. Thus, $w^{2}(\cdot)$ represents an \emph{effective
  dimension} of a structured convex set \cite{plan2013robust, plan2014dimension,
  plan2016generalized}.

\begin{definition}[Gaussian mean width]
  The Gaussian mean width (GMW) of a set $K \subseteq \reals^N$ is given by
  \begin{align*}
    w(K) := \E \sup_{x \in K} \ip{x,g}, \quad g \sim \mathcal{N}(0, I_N).
  \end{align*}
\end{definition}



Next, we include one set of conditions under which $\tilde x(\sigma)$ lies in
the descent cone of the structure set, yielding a useful norm inequality. This
proposition is a simplification of classical results found in
\cite{foucart2013mathematical}.
\begin{proposition}[Descent cone condition]
  \label{prop:descent-cone}
  For $s \geq 0$, let $x_{0} \in \Sigma_{s}^{N}$. Suppose $y = x_{0} + \eta z$
  where $\eta > 0$ and $z \in \reals^{N}$ with $z_{i} \iid \mathcal{N}(0,1)$
  lies on the event $\mathcal{E} := \{ \|z\|_{2}^{2} \leq N - 2\sqrt N \}$. If
  $\tilde x (\sigma)$ solves {\bppd} with $\sigma \geq \eta \sqrt N$, then
  $\|\tilde x\|_1 \leq \|x_{0}\|_1$ and
  $\|\tilde x -x_{0} \|_1 \leq 2 \sqrt{s} \|\tilde x - x_{0}\|_2$.
\end{proposition}

\begin{proof}[Proof of Proposition \ref{prop:descent-cone}] %
  Since $\sigma^2 \geq N$ and $\|z\|_2^2 \leq N - 2\sqrt N \leq \sqrt N$, it
  follows by $\tilde x(\sigma)$ being the minimizer and $x_{0}$ being in the
  feasible set that $\|\tilde x(\sigma)\|_1 < \|x_{0}\|_1$ on
  $\mathcal{E}$. Hence, $\tilde x - x_{0} \in T_{B_{1}^{N}}(x_{0})$, the
  $\ell_{1}$ tangent cone of $x_{0}$. By
  Lemma \ref{lem:descent-cone-equivalence}, one obtains the desired identity,
  \begin{align*}
    \|\tilde x - x_{0} \|_{1} %
    &= \|h\|_{1} = \|h_{T}\|_{1} + \|h_{T^{C}}\|_{1} %
      \leq \ip{\mathrm{sgn}(\tilde x_{T} - x_{0}), h_{T}} - \ip{\sgn x_{0}, h}
    \\
    &\leq \|\mathrm{sgn}(\tilde x_{T} - x_{0}) - \mathrm{sgn}(x_{0})\|_{2}
      \|h\|_{2} %
      \leq 2\sqrt{s}\|h\|_{2}.
  \end{align*}

\end{proof}

\begin{lemma}[Equivalent $\ell_{1}$ descent cone characterization]
  \label{lem:descent-cone-equivalence}
  Let $x \in \Sigma_{s}^{N}$ with non-empty support set $T \subseteq [N]$ and
  define $\mathcal{C} := \|x\|_{1} B_{1}^{N}$. Let
  $T_{\mathcal{C}}(x) = \mathrm{cone}(F_{\mathcal{C}}(x))$ be the tangent cone
  of the scaled $\ell_{1}$ ball about the point $x$ and define the set
  $K(x) := \{ h \in \reals^{N} : \|h_{T^{C}}\|_{1} \leq -\ip{\mathrm{sgn}(x),
    h}\}$. Then $T_{\mathcal{C}}(x) = K(x)$.
\end{lemma}

\subsubsection{Projection lemma}
\label{sssec:projection-lemma}

We introduce a result that to our knowledge is novel: the projection
lemma. Given $z \in \reals^{N}$, this lemma orders the one-parameter family of
projections $z_{t} := \mathrm{P}_{t K}(z)$ as a function of $t > 0$ when $K$ is
a closed and convex set with $0 \in K$. Namely, as depicted in
Figure \ref{fig:projection-lemma},
$\|\mathrm{P}_{t K}(z)\|_{2} \leq \|\mathrm{P}_{u K}(z)\|_{2}$ for
$0 < t \leq u < \infty$.

This lemma has immediate consequences for the ability of proximal algorithms to
recover the $0$ vector from corrupted measurements. Note that the set $K$ need
be neither symmetric nor origin-centered, but it must be convex, in general; we
have included a pictorial counterexample in
Figure \ref{fig:projection-lemma-nonconvex} to depict why.

\begin{figure}[b!]
  \hfill
  \begin{subfigure}{.4\textwidth}
    \centering
    \includegraphics[width=.8\linewidth]{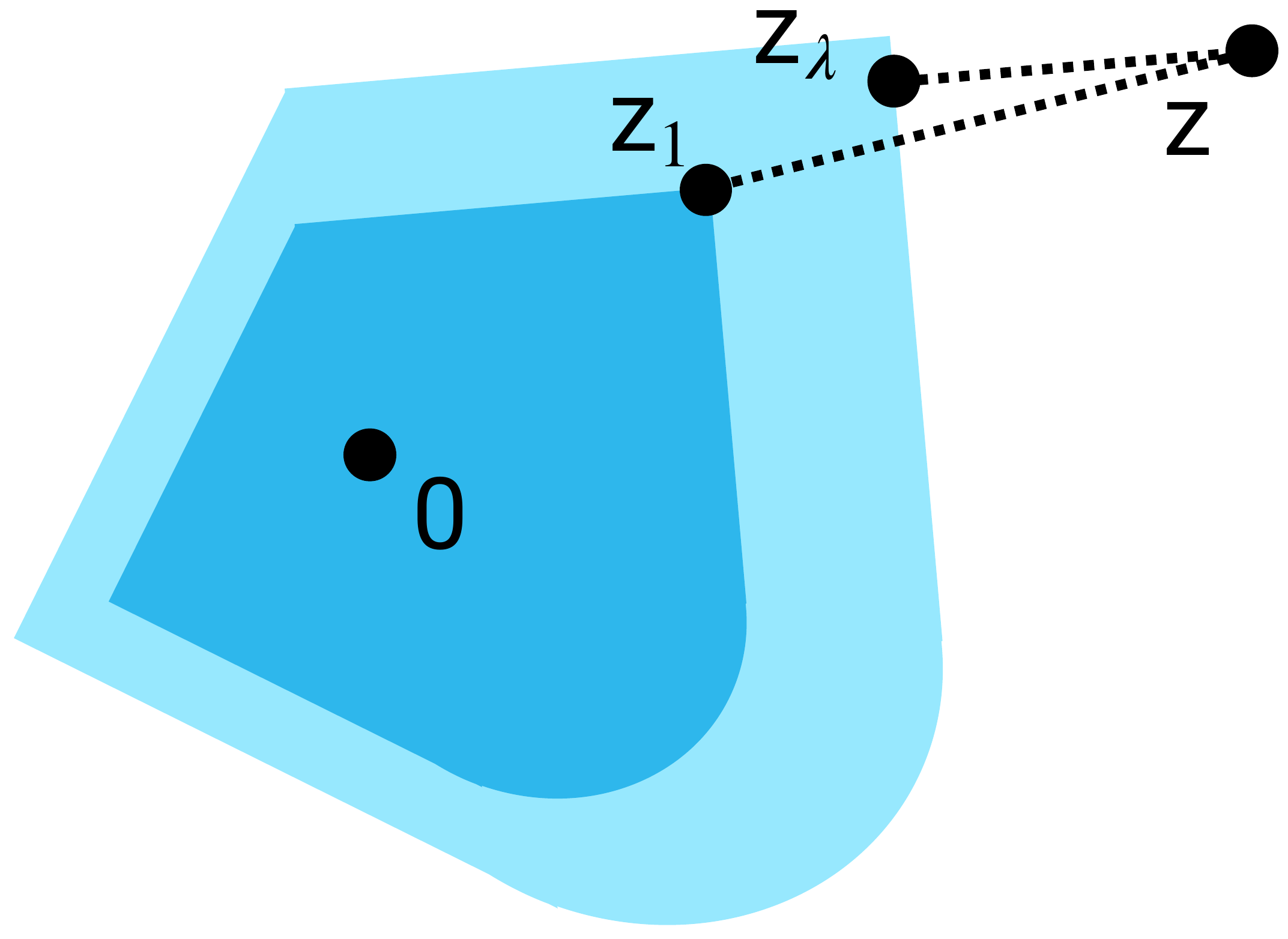}
    \subcaption{\label{fig:projection-lemma}}
  \end{subfigure}
  \hfill
  \begin{subfigure}{.4\textwidth}
    \includegraphics[width=.8\linewidth]{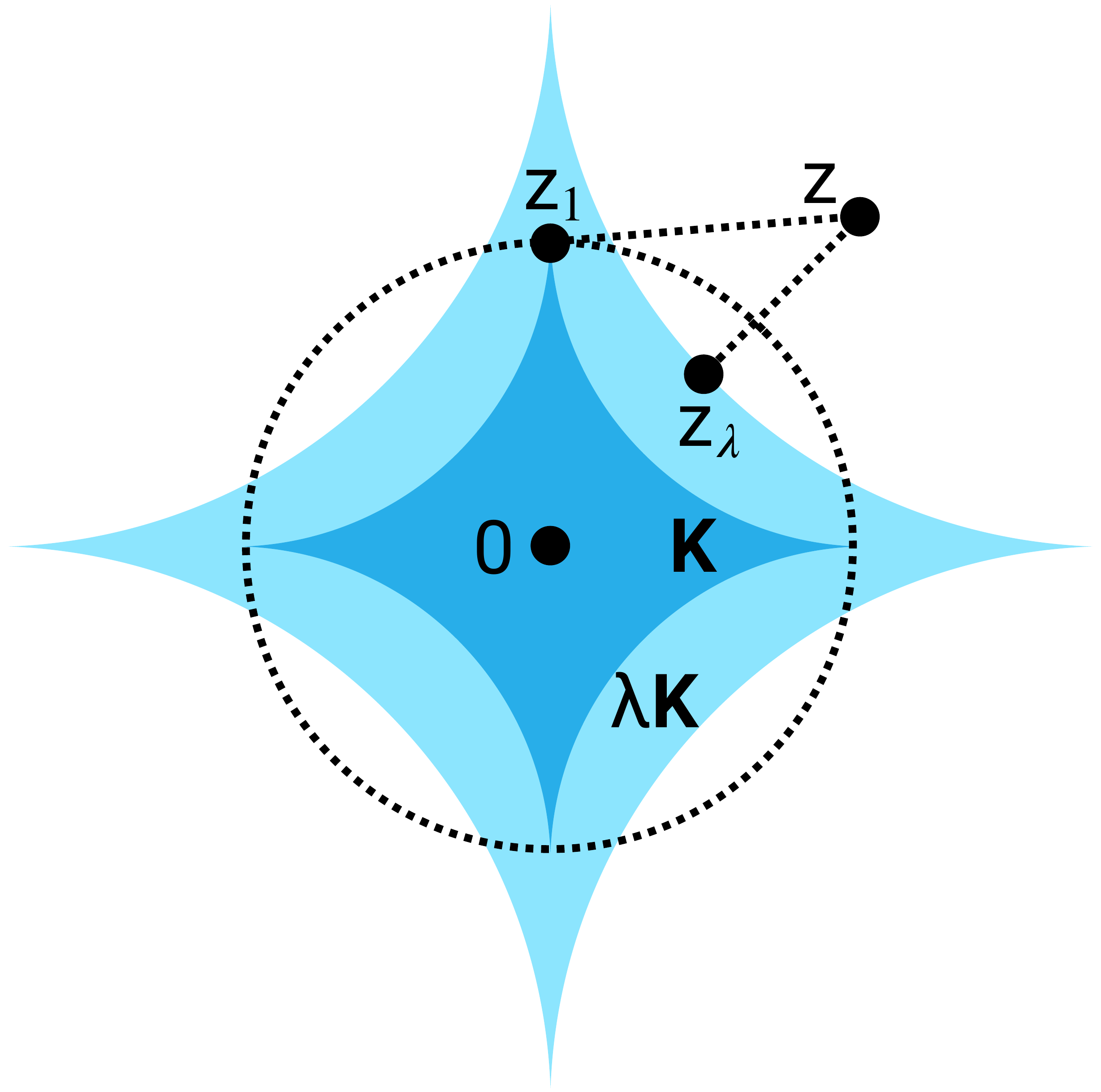}
    \subcaption{\label{fig:projection-lemma-nonconvex}}
  \end{subfigure}
  \caption{(\protect\subref{fig:projection-lemma}) A visualization of the
    lemma. Projecting $z$ onto the outer and inner set gives $z_{\lambda}$ and
    $z_{1}$, respectively; evidently, $\|z_{1}\|_{2} \leq
    \|z_{\lambda}\|_{2}$. (\protect\subref{fig:projection-lemma-nonconvex}) A
    counterexample using scaled $\ell_{p}$ balls for some $0 < p < 1$,
    suggesting why $K$ must be convex in general. Here, $z$ is projected inwards
    onto $\lambda K$, but towards a distal vertex when projected onto $K$. }
  \hfill
\end{figure}

\begin{lemma}[Projection lemma]
  \label{lem:projection-lemma}
  Let $K \subseteq \reals^{n}$ be a non-empty closed and convex set with
  $0 \in K$, and fix $\lambda \geq 1$. For $z \in \reals^{n}$,
  \begin{align*}
    \|\mathrm{P}_{K}(z)\|_{2} \leq \|\mathrm{P}_{\lambda K}(z)\|_{2}. 
  \end{align*}
\end{lemma}

The following is an alternative version of Lemma \ref{lem:projection-lemma} which
quickly follows.

\begin{corollary}
  \label{coro:projection-lemma}
  Let $K \subseteq \reals^{n}$ be a non-empty closed and convex set with
  $0 \in K$ and let $\|\cdot\|_{K}$ be the gauge of $K$. Given
  $y \in \reals^{n}$ define
  \begin{align*}
    x_{\alpha} := \argmin \{ \|x\|_{K} : \|x - y\|_{2} \leq \alpha \} 
  \end{align*}
  Then $\|x_{\alpha}\|_{2}$ is decreasing in $\alpha$. 
\end{corollary}

\begin{remark}
  The proof of Lemma \ref{lem:projection-lemma} examines the derivative of the
  function $f(t) := \|u_{t}\|_{2}^{2}$, where
  $u_{t} := t \mathrm{P}_{\lambda K}(z) + (1-t) \mathrm{P}_{K}(z)$, and yields a
  growth rate of this derivative at $t = 0$:
  \begin{align*}
    \frac12 \left.\dee{}{t} f(t)\right|_{t=0} %
    = \ip{z_{1}, z_{\lambda} - z_{1}} %
    \geq \frac{\|z_{\lambda} - z_{1}\|_{2}^{2}}{\lambda - 1}. 
  \end{align*}
\end{remark}

%
%
%
%
%
%

\subsection{Tools from probability theory}
\label{ssec:tools-from-prob}

For a full treatment of the topics herein, we refer the reader to
\cite{foucart2013mathematical, vershynin2018high, van2014probability,
  adler2009random}. We start by defining subgaussian random variables and
stating Hoeffding's inequality, which characterizes how they concentrate in high
dimensions.

\begin{definition}[$\psi_{2}$-norm]
  The subgaussian norm of a random variable $X$ is
  \begin{align*}
    \|X\|_{\psi_{2}} := \sup_{p\geq 1} p^{-1/2}\big( \E |X|^{p}\big)^{1/p}
  \end{align*}
  A random variable $X$ is subgaussian iff $\|X\|_{\psi_{2}} < \infty$. 
\end{definition}

\begin{theorem}[General Hoeffding's inequality {\cite[Theorem 2.6.3]{vershynin2018high}}]
  \label{thm:hoeffding-subgaussian}
  Let $X_{i}$, $i = 1, \ldots n$, be mean-zero subgaussian random variables and let $a \in \reals^{n}$. For $t > 0$,
  \begin{align*}
    \mathbb{P}\big( \big| \sum_{i=1}^{n} a_{i} X_{i} \big| \geq t \big) \leq e \cdot \exp\Big( \frac{-t^{2}}{C \sum_{i=1}^{n}a_{i}^{2} \|X_{i}\|_{\psi_{2}}^{2}}\Big)
  \end{align*}
\end{theorem}

One may define subexponential random variables in a way similar to subgaussian
random variables. They, too, admit a concentration inequality.

\begin{definition}[$\psi_{1}$ norm]
  The subexponential norm of a random variable is
  \begin{align*}
    \|X\|_{\psi_{1}} := \sup_{p \geq 1} p^{-1} \big( \E|X|^{p} \big)^{1/p}. 
  \end{align*}
  A random variable $X$ is subexponential iff $\|X\|_{\psi_{1}} < \infty$.
\end{definition}

\begin{theorem}[Bernstein's inequality {\cite[Theorem 2.8.1]{vershynin2018high}}]
  \label{thm:bernstein}
  Let $X_{1}, \ldots, X_{n}$ be independent mean-zero
  subexponential random variables. Then for all
  $\{a_{1}, \ldots, a_{n}\} \in \reals^{n}$,
  \begin{align*}
    \mathbb{P}\big( | \sum_{i=1}^{n}a_{i} X_{i} | \geq t \big) %
    \leq 2 \exp\Big( - C \min \Big\{ \frac{t^{2}}{k^{2}\|a\|_{2}^{2}},
    \frac{t}{k\|a\|_{\infty}}\Big\}\Big), %
    \qquad t \geq 0, k := \max_{i} \|X_{i}\|_{\psi_{1}}
  \end{align*}
\end{theorem}

Finally, we introduce a result of Borell, Tsirelson, Ibragimov, and Sudakov
about Gaussian processes, which states that the supremum of a Gaussian process
defined over a topological space $T$ behaves nearly like a normal random
variable. For a proof of this result, we refer the reader to
\cite{adler2009random}.

\begin{theorem}[Borell-TIS inequality \cite{borell1975brunn, cirel1976norms}]
  \label{thm:borell-tis}
  Let $T$ be a topological space and let $\{f_{t}\}_{t \in T}$ be a centred (\ie mean-zero) Gaussian process on $T$ with
  \begin{align*}
    \|f\|_{T} %
    &:= \sup_{t \in T} |f_{t}| %
    &
      \sigma_{T}^{2} %
    &:= \sup_{t\in T} \E \big[ |f_{t}|^{2} \big]
  \end{align*}
  such that $\|f\|_{T}$ is almost surely finite. Then $\E \|f\|_{T}$ and
  $\sigma_{T}$ are both finite and for each $u > 0$,
  \begin{align*}
    \mathbb{P}\big( \|f\|_{T} > \E \|f\|_{T} + u\big) %
    \leq \exp\big( - \frac{u^{2}}{2 \sigma_{T}^{2}} \big).
  \end{align*}

\end{theorem}

\subsubsection{Refined bounds on Gaussian mean width}
\label{sssec:refin-bounds-gauss}

Two recent results yield improved upper- and lower-bounds on the GMW of convex
polytopes intersected with Euclidean balls \cite{bellec2017localized}. Each is
integral to demonstrating {\bppd} parameter instability. The first describes how
local effective dimension of a convex hull scales with neighbourhood size.

\begin{proposition}[{\cite[Prop 1]{bellec2017localized}}]
  \label{prop:bellec1}
  Let $m \geq 1$ and $N \geq 2$. Let $T$ be the convex hull of $2N$ points in $\reals^{m}$ and assume $T\subseteq B_{2}^{m}$. Then for $\gamma \in (0, 1)$,
  \begin{align*}
    w (T \cap \gamma B_{2}^{m}) %
    \leq \min \big\{ %
    4 \sqrt{ \max \big\{1, \log(8e N \gamma^{2}) \big\} }, %
    \gamma \sqrt{ \min\{m, 2N\}} \big\}
  \end{align*}

\end{proposition}

The second result shows that Proposition \ref{prop:bellec1} is tight up to
multiplicative constants.

\begin{proposition}[{\cite[Prop 2]{bellec2017localized}}]
  \label{prop:bellec2}
  Let $m \geq 1$ and $N \geq 2$. Let $\gamma \in (0, 1]$ and assume for
  simplicity that $s = 1/ \gamma^{2}$ is a positive integer such that
  $s \leq N /5$. Let $T$ be the convex hull of the $2N$ points
  $\{ \pm M_{1}, \ldots, \pm M_{N}\} \subseteq \sph^{m-1}$. Assume that for some
  real number $\kappa \in (0, 1)$ we have
  \begin{align*}
    \kappa \| \theta \|_{2} \leq \| M \theta \|_{2} \qquad %
    \text{for all $\theta \in \reals^{N}$ such that $\|\theta\|_{0} \leq 2s$},
  \end{align*}
  Then
  \begin{align*}
    w (T \cap \gamma B_{2}^{m}) %
    \geq ( \sqrt 2 / 4) \kappa \sqrt{ \log ( N \gamma^{2} / 5)}. 
  \end{align*}
\end{proposition}

\section{{\lspd} parameter instability}
\label{sec:param-inst-lspd}

We describe a parameter instability regime for {\lspd}, revealing a regime in
which there is exactly one choice of parameter $\tau^{*} > 0$ such that
$\hat R(\tau^{*}; x_{0}, N, \eta)$ is minimax order-optimal. Specifically,
Theorem \ref{thm:constr-pd} shows that $\hat R(\tau; x_{0}, N, \eta)$ exhibits an
asymptotic singularity in the limiting low-noise regime (by low-noise regime, we
mean hereafter the regime in which $\eta \to 0$).

In \ref{ssec:lspd-numerics} we complement this asymptotic result with
numerical simulations that contrast how the three risks behave in a simplified
experimental context. The numerics support that Theorem \ref{thm:constr-pd} provides
accurate intuition to guide how {\lspd} can be expected to perform in practice
when the noise level is small relative to the magnitude of the signal's entries.

The analogue of the classical CS result is included in our result as the special
case $\tau = \tau^{*} = \|x_{0}\|_{1}$ (\emph{cf}.\@
Proposition \ref{prop:lspd-optimal-risk}). The cases for $\tau \neq \tau^{*}$ may seem
surprising initially, but can be understood with the following key intuition:
the approximation error is controlled by the effective dimension of the
constraint set.

First, one should generally not expect good recovery when the signal lies
outside the constraint set. When $\tau < \tau^{*}$, $y$ lies outside of the
constraint set with high probability in the limiting low-noise
regime. Accordingly, there is a positive distance between the true signal and
the recovered signal which may be lower-bounded by a dimension-independent
constant. Hence, the risk is determined by the reciprocal of the noise variance,
growing unboundedly as $\eta \to 0$.

On the other hand, when $\tau > \tau^{*}$, $y$ lies within the constraint set
with high probability in the limiting low-noise regime. Thus, the problem is
essentially unconstrained in this setting, so the effective dimension of the
constraint set for the problem should be considered equal to that of the ambient
dimension. In particular, one should expect that the error be proportional to
$N$.

\begin{theorem}[{\lspd} parameter instability]
  \label{thm:constr-pd}
  Let $s \geq 1, \eta > 0$ and let
  $x_{0} \in \Sigma_{s}^{N}\setminus\Sigma_{s-1}^{N}$. Given $\tau > 0$,
  \begin{align*}
    \lim_{\eta \to 0} \hat R(\tau; x_{0}, N, \eta) =
    \begin{cases}
      \infty & \tau < \|x_{0}\|_{1}\\
      R^{*}(s, N)
      & \tau = \|x_{0} \|_{1}\\
      N & \tau > \|x_{0}\|_{1}
    \end{cases}
  \end{align*}

\end{theorem}

In summary, the surprising part of this result is that there is a sharp phase
transition between two unstable regimes, with the optimal regime lying on the
boundary of the two phases. We argue this suggests that there is only one
reasonable choice for $\tau$ in the low-noise regime. Observe, that
Theorem \ref{thm:constr-pd} connects with Theorem \ref{thm:simplified-lspd} by taking
the limit of the problem as $N \to \infty$ after first restricting to signals of
a finite norm (arbitrarily, $1$) so that the essence of the result is preserved.

\section{{\qppd} parameter instability}
\label{sec:param-inst-qppd}

We show that $R^{\sharp}(\lambda; x_{0}, N, \eta)$ is smooth in the low-noise
regime. This result becomes evident from the closed-form expression for $R^{\sharp}(\lambda; s, N)$ that emerges for this special case. At first, this smoothness result seems to stand in contrast to the
``cusp-like'' behaviour that we observe analytically and numerically for
$\lim_{\eta \to 0} \hat R(\tau; x_{0}, N, \eta)$ (\emph{cf}.\@
Figure \ref{fig:lspd-numerics}). However, $R^{\sharp}(\lambda; s, N)$ 
possesses unfavourable dependence on $N$ that is elucidated in
Theorem \ref{thm:qppd-instability}.

Briefly, if the governing parameter $\lambda$ is too small, then the risk grows
unboundedly as a power law of $N$ in high dimensions. This rate of growth
implies that the risk is minimax suboptimal for such $\lambda$. To our
knowledge, this result is novel. In contrast, for all suitably large $\lambda$,
$R^{\sharp}(\lambda; s, N)$ admits the desirable property suggested in
section \ref{sec:introduction}:
$R^{\sharp}(\lambda; s, N) \lesssim (\lambda / \lambda^{*})^{2} R^{*}(s,
N)$. The result, stated in Theorem \ref{thm:qppd-rhs-stability}, essentially follows
from known \textsc{Lasso} bounds for RIP matrices: $R(\lambda) \leq \lambda^{2} s$. Thus,
in the low-noise regime, $R^{\sharp}(\lambda; x_{0}, N, \eta)$ exhibits a phase
transition between order-optimal and suboptimal regimes.

The numerics of section \ref{sec:qppd-numerics} suggest a viable constant for the
growth rate of the risk when $\lambda$ is too small, and support
Theorem \ref{thm:qppd-rhs-stability} in the case where $\lambda$ is sufficiently
large. These numerics also clarify the role that the dimension-dependent growth
rate serves in the stability of {\qppd} about $\lambda^{*}$.

\subsection{Smoothness of the risk}
\label{ssec:smoothness-qppd-mse}

The {\qppd} estimator for a problem with noise level $\eta > 0$ and with
parameter $\lambda > 0$ is given by soft-thresholding by $\eta\lambda$. In
particular, $x^{\sharp}(\eta\lambda)$ is a smooth function with respect to the
problem parameters, hence so is $R^{\sharp}(\lambda; x_{0}, N, \eta)$ (being a
composition of smooth functions). However, the closed form expression for
$R^{\sharp}(\lambda; x_{0}, N, \eta)$ is unavailable, because the expectations
involved are untractable in general. When the noise-level vanishes this is no
longer true and we may compute an exact expression in terms of $\lambda, s$ and
$N$ for the risk. Specifically, we note that the smoothness result below is not
special to the case where $\eta \to 0$, but is notable because of the closed
form expression for the risk that is obtained.

Moreover, the result is notable, because the closed form expression is
equivalent (in some precisely definable sense) to
$R^{\sharp}(\lambda; x_{0}, N, \eta)$ when $\eta > 0$ and the magnitudes of the
entries of $x_{0}$ are all large (\ie ``the signal is well-separated from the
noise''). We make this connection after the main results discussed below. In
turn, this connects Theorem \ref{thm:qppd-instability} and
Theorem \ref{thm:qppd-rhs-stability} to Theorem \ref{thm:simplified-qppd}, where the
analytic expression is used to derive the so-called left-sided parameter
instability and right-sided parameter stability results.

\begin{proposition}[{$R^{\sharp}(\lambda; x_{0}, N, \eta)$ smoothness}]
  \label{prop:qppd-smooth}
  Let $s \geq 0, N \geq 1, x_{0} \in \Sigma_{s}^{N}$ and $\eta > 0$. For
  $\lambda > 0$,
  \begin{align}
    \label{eq:pd-lambda}
    \lim_{\eta\to 0 } R^{\sharp}(\lambda; x_{0}, N, \eta) %
    = s (1 + \lambda^{2}) + 2(N-s)\big[(1+\lambda^{2})\Phi(-\lambda) %
    - \lambda\phi(\lambda)\big]
  \end{align}

\end{proposition}

\begin{remark}
  Here and beyond, we denote the limiting low-noise risk by
  $R^{\sharp}(\lambda; s, N) := \lim_{\eta \to 0} R^{\sharp}(\lambda; x_{0}, N,
  \eta)$; and define the function $G(\lambda) := (1+\lambda^{2})\Phi(-\lambda) %
  - \lambda\phi(\lambda)$ for notational brevity, where $\phi$ and $\Phi$ denote
  the standard normal pdf and cdf, respectively.
\end{remark}

An equivalence in behaviour is seen between the low-noise regime $\eta \to 0$
and the large-entry regime $|x_{0,j}| \to \infty$ for
$j \in \mathrm{supp}(x_{0})$ with $\eta > 0$. For both programs, the noise level
is ``effectively'' zero by comparison to the size of the entries of
$x_{0}$. This type of scale invariance allows us to re-state the previous result
as a $\max$ formulation.

\begin{corollary}[$\max$-formulation]
  \label{coro:qppd-max-formulation}
  Let $s \geq 0, N\geq 1, x_{0} \in \Sigma_{s}^{N}$ and $\eta >
  0$. For $\lambda > 0$,
  \begin{align*}
    \sup_{x_{0} \in \Sigma_{s}^{N}} R^{\sharp}(\lambda; x_{0}, N, \eta)%
    = R^{\sharp}(\lambda; s, N)
  \end{align*}
\end{corollary}

\subsection{Left-sided parameter instability}
\label{ssec:left-sided-parm-inst}

We reveal an asymptotic regime in which $R^{\sharp}(\lambda; s, N)$ is minimax
suboptimal for all $\lambda$ sufficiently small. The result follows from showing
the risk derivative is large for all $\lambda < \bar \lambda$ when $s$ is
sufficiently small relative to $N$. Here, $\bar \lambda := \sqrt{2\log N}$ is an
Ansatz estimate of $\lambda^{*}$ used to make the proof proceed
cleanly. Finally, we show in what sense $\bar \lambda$ is asymptotically
equivalent to $\lambda^{*}$ in Proposition \ref{prop:asymptotic-equivalence}.

The proof for the bound on the risk derivative follows by calculus and a
standard estimate of $\Phi(-\lambda)$ in terms of $\phi(\lambda)$. Its scaling
with respect to the ambient dimension destroys the optimal behaviour of
$R^{\sharp}(\lambda; x_{0}, N)$ for all $\lambda < \bar \lambda$. The proof of
this result, stated in Theorem \ref{thm:qppd-instability}, follows immediately from
Lemma \ref{lem:qppd-instability} by the fundamental theorem of calculus.

\begin{lemma}[risk derivative instability]
  \label{lem:qppd-instability}
  Fix $s \geq 1$. For any $\varepsilon \in (0, 1)$, there exists $ C > 0$ and an
  integer $N_{0} = N_{0}(s) \geq s$ so that for all $N \geq N_{0}$
  \begin{align*}
    - \left.\dee{}{u} \right|_{u = 1- \varepsilon}\hspace{-16pt} R^{\sharp}(u\bar \lambda ; s, N) 
    \geq C N^{\varepsilon}
  \end{align*}
  where $\bar\lambda = \sqrt{2 \log (N)}$ is an estimate of the optimal
  parameter choice for {\qppd}.
  
\end{lemma}
\begin{theorem}[{\qppd} parameter instability]
  \label{thm:qppd-instability}
  Under the conditions of the previous lemma, for $\varepsilon \in (0,1)$ there
  exists a constant $C > 0$ and integer $N_{0} \geq 1$ such that for all
  $N \geq N_{0}$,
  \begin{align*}
    R^{\sharp}((1-\varepsilon)\bar \lambda; s, N) %
    \geq C \frac{N^{\varepsilon}}{\log N}. 
  \end{align*}

\end{theorem}

Though these results may initially seem surprising, we claim they are sensible
when viewed in comparison to unregularized proximal denoising (\ie
$\lambda = 0$). In this case, sparsity of the signal $x_{0}$ is unused and so
one expects error be proporitional to the ambient dimension, as in
section \ref{sec:param-inst-lspd}. In the low-noise regime, the sensitivity of the
program to $\lambda$ is apparently amplified, and for $\lambda > 0$ one may
still expect {\qppd} to behave similarly to unregularized proximal denoising,
begetting risk that behaves like a power law of $N$.

\begin{proposition}[Asymptotic equivalence]
  \label{prop:asymptotic-equivalence}
  Let $N \in \nats$ with $N \geq 2$, $s \in [N]$ and
  $\bar \lambda = \sqrt{2\log N}$. For given problem data, suppose
  $x^{\sharp}(\lambda)$ solves {\qppd}, and let $\lambda^{*}$ be the optimal
  parameter choice for $R^{\sharp}(\lambda; s, N)$. Then
  \begin{align*}
    \lim_{N\to\infty}\frac{\bar\lambda}{\lambda^{*}} =  1
  \end{align*}
\end{proposition}
\begin{remark}
  The value $\bar \lambda$ estimates the optimal parameter choice for {\qppd} in
  the following sense
  \cite{oymak2016sharp}. 
  \begin{align*}
    \lambda^{*} = O(\sqrt{\log (N/s)}) \approx \sqrt{2\log N} =: \bar \lambda
  \end{align*}
\end{remark}

\subsection{Right-sided parameter stability}
\label{ssec:right-sided-param}

In the low-noise regime, $R^{\sharp}$ may still be order-optimal if $\lambda$ is
chosen large enough. Specifically, if $\lambda = L\lambda^{*}$ for some $L > 1$,
then $R^{\sharp}(\lambda; x_{0}, N)$ is still minimax order-optimal. We claim no
novelty for the result of this section, but use it as a contrast to elucidate
the previous theorem. Whereas for $\lambda < \bar \lambda$ we are penalized for
under-regularizing in the low-noise regime in high dimensions, the theorem below
implies that we are not penalized for over-regularizing.

\begin{theorem}
  \label{thm:qppd-rhs-stability}
  {\qppd} is parameter stable in the sense that for any $\lambda > 0$ satisfying
  $L = \lambda / \lambda^{*} > 1$, there is $N_{0} = N_{0}(s, \lambda) \geq 2$ so
  that for all $N \geq N_{0}$,
  \begin{align*}
    \frac{R^{\sharp}(\lambda; s, N)}{R^{*}(s,N)} %
    \leq CL^{2}.
  \end{align*}
  
\end{theorem}

Observe that the theorem still holds in the event that $\lambda^{*}$ is replaced
by $\bar \lambda$. Thus, one may obtain the exact point of the phase transition,
$\bar\lambda$, observed in Theorem \ref{thm:simplified-qppd}. In fact, with this
note, Theorem \ref{thm:simplified-qppd} follows as a direct consequence of the
results of this section by letting $N \to \infty$.

\section{{\bppd} parameter instability}
\label{sec:param-inst-bppd}

The program {\bppd} is maximin suboptimal for very sparse vectors $x_{0}$. We
show that $\tilde R(\sigma; x_{0}, N, \eta)$ scales as a power law of $N$ for
all $\sigma > 0$. This rate is significantly worse than $R^{*}(s,N)$. When
$x_{0}$ is very sparse and {\bppd} is underconstrained, then
$\sigma \geq \eta N$ and \ref{ssec:underconstrained-bppd} proves that
$\tilde R(\sigma; x_{0}, N, \eta) = \Omega(\sqrt N)$. When {\bppd} is
overconstrained, then $\sigma \leq \eta \sqrt N$ and
\ref{ssec:overconstrained-bppd}, proves that
$\tilde R(\sigma; x_{0}, N, \eta) = \Omega(N^{q})$ for some $q > 0$ when $x_{0}$
is very sparse.

Intuitively, {\bppd} kills not only the noise, but also eliminates too much of
the signal content when underconstrained and $s$ is small compared to
$N$. Because the signal is very sparse, destroying the signal content is
disastrous to the risk. When overconstrained, the remaining noise overwhelms the
risk, because the off-support has size approximately equal to the ambient
dimension.

The above two steps are combined in Theorem \ref{thm:bppd-minimax} as a minimax
formulation over all $\sigma > 0$ and $x_{0} \in \Sigma_{s}^{N}$. In
Theorem \ref{thm:bppd-maximin}, this result is strengthened to a maximin statement
over $x_{0} \in \Sigma_{s}^{N}$ and all $\sigma > 0$.

Although these results may seem to run contrary to the apparent efficacy of the
CS analogue of {\bppd} in empirical settings, we assure the reader that they are
consistent. The type of parameter instability described in this section occurs
at very large dimensions, in the setting where $s \geq 1$ is fixed. Thus,
although these results bode poorly for the ability of {\bppd} to recover even
the $0$ vector (arguably a desirable property of a denoising program), many
structured high-dimensional signals observed in practice are not so sparse [in a
basis] as to belong to the present regime. Nevertheless, this result serves as a
caveat for the limits of a popular $\ell_{1}$ convex program.

\subsection{Underconstrained {\bppd}}
\label{ssec:underconstrained-bppd}

The proof of this result uses standard methods from CS and may be found in
\ref{ssec:proofs-bppd-results}.

\begin{lemma}
  \label{lem:bppd-uc}
  Let $s \geq 1$ and let $x_{0} \in \Sigma_{s}^{N}\setminus \Sigma_{s-1}^{N}$
  be an exactly $s$-sparse signal with $|x_{j}| \gtrsim N$ for all
  $j\in \mathrm{supp}(x_{0})$. If $\sigma > \eta \sqrt N$, then there exists a
  constant $C > 0$ and integer $N_{0} = N_{0}(s) \geq 2$ such that if
  $N \geq N_{0}$ then
  \begin{align*}
    \tilde R(\sigma; x_{0}, N, \eta) \geq C \sqrt N. 
  \end{align*}

\end{lemma}

\subsection{Overconstrained {\bppd}}
\label{ssec:overconstrained-bppd}

The proof that $\tilde R(\sigma; x_{0}, N, \lambda)$ scales as a power law of
$N$ when $\sigma \leq \eta \sqrt N$ proceeds by an involved argument, hinging on
two major steps. The first step is to find an event whose probability is
lower-bounded by a universal constant, on which {\bppd} fails to recover the $0$
vector when $\sigma = \eta \sqrt N$. Then, Lemma \ref{lem:projection-lemma}
extends this result to all $\sigma \leq \eta \sqrt N$. At this point, one may
obtain the minimax result of Theorem \ref{thm:bppd-minimax}, as well as a partial
maximin result for all $x_{0} \in \Sigma_{s}^{N}$ on the restriction to
$\sigma \leq \eta \sqrt N$. Then, to strengthen these claims to a maximin result
over all $\sigma > 0$, we prove a lemma that leverages elementary properties
from convex analysis to show how the error of an estimator may be controlled by
that of a lower dimensional estimator from the same class.

In this section, we state key results for building intuition and defer technical
results and proofs to \ref{ssec:proofs-bppd-results}.

\begin{theorem}[Overconstrained Maximin]
  \label{thm:bppd-oc-maximin}
  There exist universal constants $C > 0, q \in (0, \frac12)$ and
  $N_{0} \geq 2$ an integer such that for all $N \geq N_{0}, s \geq 0$ and
  $\eta > 0$,
  \begin{align*}
    \sup_{x_{0} \in \Sigma_{s}^{N}} \inf_{\sigma \leq \eta \sqrt N} \tilde R(\sigma; x_{0}, N, \eta) %
    \geq C N^{q}.
  \end{align*}

\end{theorem}

By scaling, it is sufficient to prove this result in the case where $\eta =
1$. The discussion below thus assumes $y = x_{0} + z$, while results are stated
in full generality. The main result relies on proving
\begin{align*}
  \inf_{\sigma \leq \sqrt N} \tilde R(\sigma; x_{0}, N, 1) \geq C N^{q}
\end{align*}
when $x_{0} \equiv 0$, trivially implying the equation before it. Thus, the
problem now becomes that of recovering the $0$ vector from standard normally
distributed noise:
\begin{align*}
  \tilde x (\sigma) = \argmin \{ \|x\|_{1} : \|x - z\|_{2}^{2} \leq \sigma^{2}\}.
\end{align*}
Here and below, we denote the feasible set in {\bppd} by
$F(z; \sigma) = B_{2}^{N}(z; \sigma)$ and use the notation $F := F(z; \sqrt
N)$. For $\lambda > 0$ and $0 < \alpha_{2} \leq \alpha_{1} < \infty$, define
$K_{i} = \lambda B_{1}^{N} \cap \alpha_{i} B_{2}^{N}$ to be the intersection of
the $\ell_{1}$-ball scaled by $\lambda$ with the $\ell_{2}$-ball scaled by
$\alpha_{i}$ for $i = 1, 2$.

With $\sigma = \sqrt N$, we prove a geometric lemma. A pictorial representation
of this lemma appears in Figure \ref{fig:bppd-oc-3a}, in which we have represented
$\lambda B_{1}^{N}$ using Milman's 2D representation of high-dimensional
$\ell_{1}$ balls to facilitate the intuition for how they behave in the present
context. The key to the proof of Theorem \ref{thm:bppd-oc-maximin} is the geometric
lemma below, Lemma \ref{lem:bppd-oc-3a}. It proves there exists an $\ell_{1}$
ball of radius $\lambda$ that intersects the feasible set, hence a solution
$ \tilde x(\sigma)$ must satisfy $\|\tilde x (\sigma) \|_{1} \leq \lambda$.
Further, it shows that any vector in the ball $\lambda B_{1}^{N}$ which has
small Euclidean norm does not intersect the feasible set.  Thus, the solution
must have large Euclidean norm.

Finally, this geometric lemma verifies that the previous three conditions occur
on an event occurring with at least probability $k_{3} > 0$. As an immediate
consequence, this lemma yields a lower risk bound,
Corollary \ref{coro:bppd-oc-3a}. The integers
$N_{0}^{\eqref{prop:bppd-oc-1b}}, N_{0}^{\eqref{prop:bppd-oc-2b}}$ are defined
in the technical results of
\ref{sssec:supporting-propositions-for-bppd-oc-3a}.

\begin{lemma}[Geometric lemma]
  \label{lem:bppd-oc-3a}
  Let $K_{1}, K_{2}, F$ be defined as above. Let
  $N_{0}^{\eqref{lem:bppd-oc-3a}} := \max \{ N_{0}^{\eqref{prop:bppd-oc-1b}},
  N_{0}^{\eqref{prop:bppd-oc-2b}}\}$ be a universal constant and suppose
  $N > N_{0}^{\eqref{lem:bppd-oc-3a}}$. There are universal constants
  $k_{3} = k_{3}(N_{0}^{\eqref{lem:bppd-oc-3a}}) > 0$, $C_{3}, q > 0$, and an
  event $\mathcal{E}$ such that
  
  \begin{minipage}{.225\textwidth}
    \begin{enumerate}
    \item $K_{1} \cap F \neq \emptyset$
    \end{enumerate}
  \end{minipage}
  \begin{minipage}{.2255\textwidth}
    \begin{enumerate}
      \setcounter{enumi}{1}
    \item $K_{2} \cap F = \emptyset$
    \end{enumerate}
  \end{minipage}
  \begin{minipage}{.225\textwidth}
    \begin{enumerate}
      \setcounter{enumi}{2}
    \item $\alpha_{2} > C_{3} N^{q}$
    \end{enumerate}
  \end{minipage}
  \begin{minipage}{.225\textwidth}
    \begin{enumerate}
      \setcounter{enumi}{3}
    \item $\mathbb{P}(\mathcal{E}) > k_{3}$. 
    \end{enumerate}
  \end{minipage}
\end{lemma}

\begin{figure}[b!]
  \begin{minipage}[h]{.45\linewidth}
    \centering
    \includegraphics[width=\textwidth]{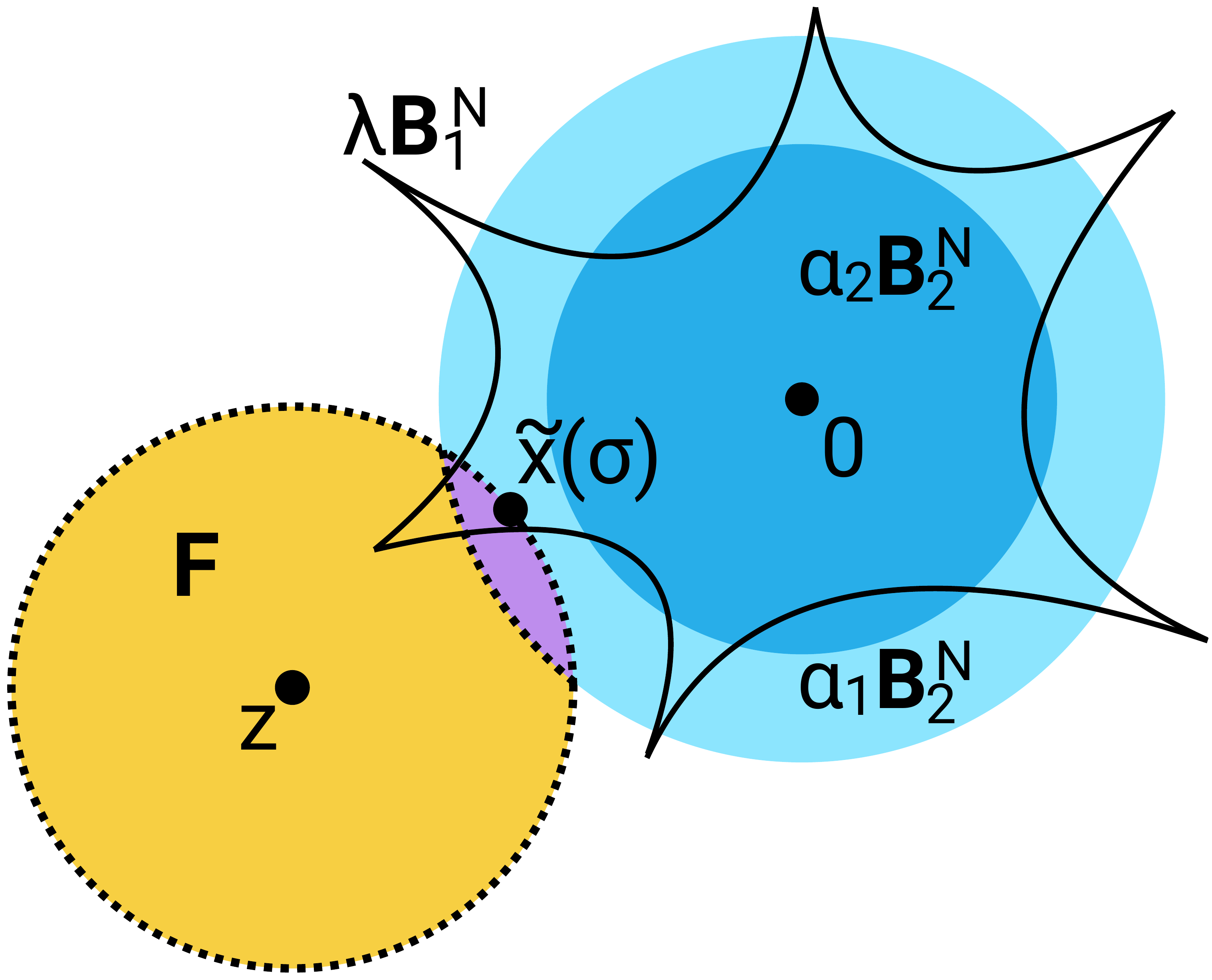}
  \end{minipage}
  \hfill
  \begin{minipage}[h]{.45\linewidth}
    \caption{A visualization of the lemma. We use Milman's 2D representation of
      high-dimensional $\ell_{1}$ balls to facilitate the intuition. In this
      setting, $\tilde x(\sigma)$ must lie inside $\lambda B_{1}^{N}$. On the
      event $\mathcal{E}$ described by the lemma, one simultaneously finds
      $K_{1} \cap F \neq \emptyset$ and $K_{2} \cap F =
      \emptyset$. \label{fig:bppd-oc-3a} }
  \end{minipage}
\end{figure}

\begin{corollary}
  \label{coro:bppd-oc-3a}
  Fix $\eta > 0$. There are universal constants $C, q > 0$ such that for all
  $N \geq N_{0}^{\eqref{lem:bppd-oc-3a}}$,
  \begin{align*}
    \tilde R(\eta \sqrt N; 0, N, \eta) \geq C N^{q}.
  \end{align*}
\end{corollary}

Next we extend Corollary \ref{coro:bppd-oc-3a} from the case where $\sigma = \sqrt N $
to any positive $\sigma \leq \sqrt N$. The proof of this result
follows near immediately from the projection lemma in
Lemma \ref{lem:projection-lemma}. Thus, one finds $\tilde x(\sigma) $ has
Euclidean norm at least as large as $\tilde x(\sqrt N)$ when $\tilde x(\sigma)$
is an estimator of the $0$ vector.

\begin{lemma}
  \label{lem:bppd-oc-sigma-lt-sqrtN}
  Let $0 < \sigma_{1} < \sigma_{0} = \sqrt N$ and $x_{0} \equiv 0$. Define
  $\tilde x(\sigma_{0}), \tilde x(\sigma_{1})$ as in {\bppd} for
  $\sigma = \sigma_{0}, \sigma_{1}$, respectively. Then
  $\|\tilde x(\sigma_{1})\|_{2}^{2} \geq \|\tilde
  x(\sigma_{0})\|_{2}^{2}$. Moreover, for $N \geq 2$,
  \begin{align*}
    \E \|\tilde x(\sigma_{1})\|_{2}^{2} %
    \geq \E \|\tilde x(\sigma_{0})\|_{2}^{2}.
  \end{align*}

\end{lemma}


\subsection{Minimax results}
\label{ssec:minimax-results}

We now have the tools to state a minimax instability result for
{\bppd}. Informally, the best worst-case risk scales as a power law of $N$ in
the very sparse regime. In particular, for $s$ fixed and $N$ sufficiently large,
there is no choice of $\sigma > 0$ yielding order-optimal risk for its
corresponding worst-case signal.

\begin{theorem}[Minimax Suboptimality]
  \label{thm:bppd-minimax}
  There are universal constants $C > 0, q \in (0, \frac12], N_{0} \geq 2$ such
  that for all $N \geq N_{0}, \eta \geq 0$ and $s \geq 1$,
  \begin{align*}
    \adjustlimits \inf_{\sigma > 0} \sup_{x \in \Sigma_{s}^{N}} \tilde R(\sigma; x, N, \eta) \geq C N^{q}
  \end{align*}


\end{theorem}

\subsection{Maximin results}
\label{ssec:maximin-results}

The final result of this section establishes maximin parameter instability for
all $x_{0} \in \Sigma_{s}^{N}$ and $\sigma > 0$. To do this, we must show there
exists a choice of signal $x_{0} \in \reals^{N}$ admitting no choice of
$\sigma > 0$ bestowing order optimal recovery error. To this end, we will
demonstrate that the previous overconstrained instability results extend to
$s$-sparse signals with $s \geq 1$. This will be enough to yield a choice of
$x_{0}$ whose recovery is suboptimal over the whole parameter range.

\begin{lemma}[Overconstrained {\bppd}, $s\geq 1$]
  \label{lem:bppd-oc-sge1}
  Let $x_{0} \in \Sigma_{s}^{N}$ with
  $\mathrm{supp}(x_{0}) \subseteq T \subseteq [N]$, let $y = x_{0} + \xi$ for
  some $\xi \in \reals^{N}$, let $x_{1} := (x_{0})_{T^{C}} \in \Sigma_{0}^{N-s}$
  and fix $\sigma > 0$. Let $\tilde x = \tilde x(\sigma) \in \reals^{N}$ be the
  solution of {\bppd} where $x_{0}$ is the ground truth, and let
  $\tilde x' = \tilde x'(\sigma) \in \reals^{N-s}$ be the solution of {\bppd}
  where $x_{1}$ is the ground truth. Then
  \begin{align*}
    \|\tilde x_{T^{C}}\|_{2} \geq \|\tilde x'\|_{2}. 
  \end{align*}
\end{lemma}

An immediate consequence of this result is the following inequality between the
Euclidean norms of the error vectors.

\begin{corollary}
  \label{coro:bppd-oc-sge1}
  Let $h := \tilde x - x_{0}$ and $h' := \tilde x' - x_{1}$, where
  $x_{0}, x_{1}, \tilde x, \tilde x'$ are defined as above. Then,
  \begin{align*}
    \|h\|_{2} \geq \|h'\|_{2}. 
  \end{align*}
\end{corollary}

\begin{remark}
  The above corollary is not yet sufficient to imply the desired maximin result
  below. As per the lemma, if $N-s \geq N_{0}^{\eqref{lem:bppd-oc-3a}}$ then
  $\tilde x'$ is parameter unstable for $\sigma \leq \sqrt{N-s}$ and so
  $\tilde x$ is, too. The fix for this slight mismatch is trivial, but
  technical. The result can be extended to the range $\sigma \leq \sqrt N$ by
  adjusting the constants in the proof of Lemma \ref{lem:bppd-oc-3a} and its
  constituents, leveraging the fact that $(N-s) / N \to 1$ as $N \to \infty$ and
  re-selecting $N_{0}^{\eqref{lem:bppd-oc-3a}}$ if necessary. We omit the
  details of this technical exercise.
\end{remark}

We proceed under the assumption that the constants have been tuned to allow for
$\tilde x'$ parameter instability to imply $\tilde x$ parameter instability for
all $\sigma \leq \sqrt N$. Thus equipped, we state the following maximin
parameter instability result for {\bppd}. The proof of this result proceeds by
finding a signal $x_{0} \in \Sigma_{s}^{N}$ such that
$\tilde R(\sigma; x_{0}, N, \eta)$ is suboptimal for all $\sigma > 0$. Since
Lemma \ref{lem:bppd-uc} applies only to signals $x_{0}$ with at least one non-zero
entry, one shows there exists such a signal which simultaneously admits poor
risk for $\sigma \leq \eta \sqrt N$ and $\sigma \geq \eta \sqrt N$. For example,
it is enough to take $x_{0} := N e_{1}$ where $e_{1} \in \reals^{N}$ is the
first standard basis vector.

\begin{theorem}[{\bppd} maximin suboptimality]
  \label{thm:bppd-maximin}
  There are universal constants $C > 0, q \in (0, \frac12]$ and $N_{0} \geq 1$
  such that for all $N \geq N_{0}$
  \begin{align*}
    \sup_{x_{0} \in \Sigma_{s}^{N}} \inf_{\sigma > 0} %
    \tilde R(\sigma; x_{0}, N, \eta) %
    \geq C N^{q}. 
  \end{align*}
\end{theorem}

\begin{remark}
  The current result is given in a maximin framework. This framework is stronger
  than the minimax one in which these types of results are typically framed. In
  essence, the maximin framework assumes that the minimizer has knowledge about
  the ground truth signal $x_{0}$; even still it is not possible to choose
  $\sigma$ to achieve order-optimal risk.
\end{remark}

\section{Numerical Results}
\label{sec:numerical-results}

Let $\mathfrak{P} \in \{ \text{{\lspd}, {\qppd}, {\bppd}} \}$ be a PD program
with solution $x^{*}(\varrho)$ where $\varrho \in \{ \tau, \lambda, \sigma\}$ is
the associated parameter. Given a signal $x_{0}$ and noise $\eta z$, denote by
$\mathcal L(\varrho; x_{0}, N, \eta z)$ the loss associated to $\mathfrak P$ and
define $\varrho^{*} = \varrho(x_{0}, \eta) > 0$ to be the value of $\varrho$
yielding best risk (\ie where $\E_{z} \mathcal L(\varrho; x_{0}, N, \eta z)$ is
minimal). We say the normalized parameter $\rho$ for the problem $\mathfrak P$
is given by $\rho := \varrho / \varrho^{*}$ and note that $\rho = 1$ is a
population estimate of the argmin of
$\mathcal L(\varrho; x_{0}, N, \eta \hat z)$; by the law of large numbers, this
risk estimates well an average of such losses over many realizations $\hat
z$. Finally, define the auxiliary function
$L (\rho; x_{0}, N, \eta \hat z) := \mathcal L(\rho \varrho^{*}; x_{0}, N, \eta
\hat z)$.

The plots in Figures Figure \ref{fig:lspd-numerics-a}, Figure \ref{fig:lspd-numerics-b},
Figure \ref{fig:qppd-instability-b}, Figure \ref{fig:bppd-numerics-a} and
Figure \ref{fig:parameter-stability} visualize the average loss,
\begin{align}
  \label{eq:avg-loss}
  \bar L(\rho_{i}; x_{0}, N, \eta, k) %
  := \frac{1}{k} \sum_{j = 1}^{k} L(\rho_{i}; x_{0}, N, \eta \hat z_{ij})
\end{align}
for each program, evaluated on a grid $\{\rho_{i}\}_{i=1}^{n}$ of size $n$ and
plotted on a $\log$-$\log$ scale, where
$L(\rho; x_{0}, N, \eta\hat z) = \eta^{-2} \|x^{*}(\varrho) -
x_{0}\|_{2}^{2}$. Here, each of the $nk$ realizations of the noise is
distributed according to $\hat z_{ij} \sim \mathcal{N}(0,1)$, the noise level
given by $\eta$ and the signal by $x_{0}$ where $x_{0} = N \sum_{i=1}^{s} e_{i}$
with $e_{i}$ being the $i$th standard basis vector. The grid
$\{\rho_{i}\}_{i=1}^{n}$ was logarithmically spaced and centered about
$\rho_{(n+1)/2} = 1$ with $n$ always odd. The solutions to each PD problem were
obtained using standard available methods in Python: \texttt{sklearn}'s
\texttt{minimize\_scalar} function from the \texttt{optimize} module was used
for solving {\lspd} and {\bppd} \cite{scikit-learn}, while the solution to
{\qppd} was obtained \emph{via} soft-thresholding. Finally, the optimal values
$\tau^{*}, \lambda^{*}$ and $\sigma^{*}$ were either determined analytically
(\eg $\tau^{*} = \|x_{0} \|_{1}$), or estimated on a dense grid about an
approximately optimal value for that parameter. Initial guesses for $\sigma^{*}$
and $\lambda^{*}$ were $\eta \sqrt{N}$ and $\sqrt{2 \log (N/s)}$ respectively.

\subsection{{\lspd} numerical simulations}
\label{ssec:lspd-numerics}

\begin{figure}[t]
  \begin{subfigure}{.5\linewidth}
\includegraphics[width=\linewidth]{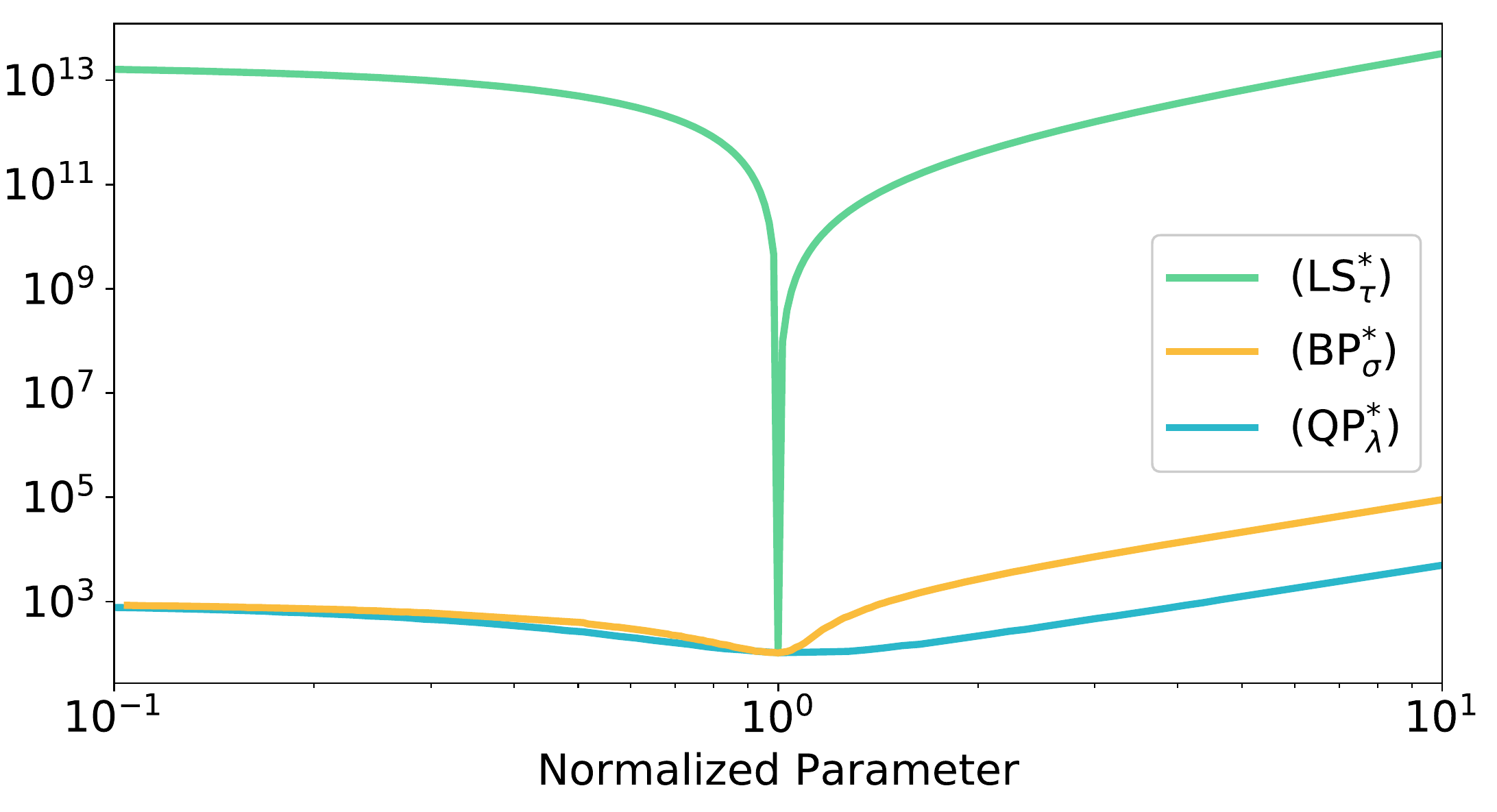}
    \subcaption{\label{fig:lspd-numerics-a}}
  \end{subfigure}
  \begin{subfigure}{.5\linewidth}
    \centering
    \includegraphics[width=\linewidth]{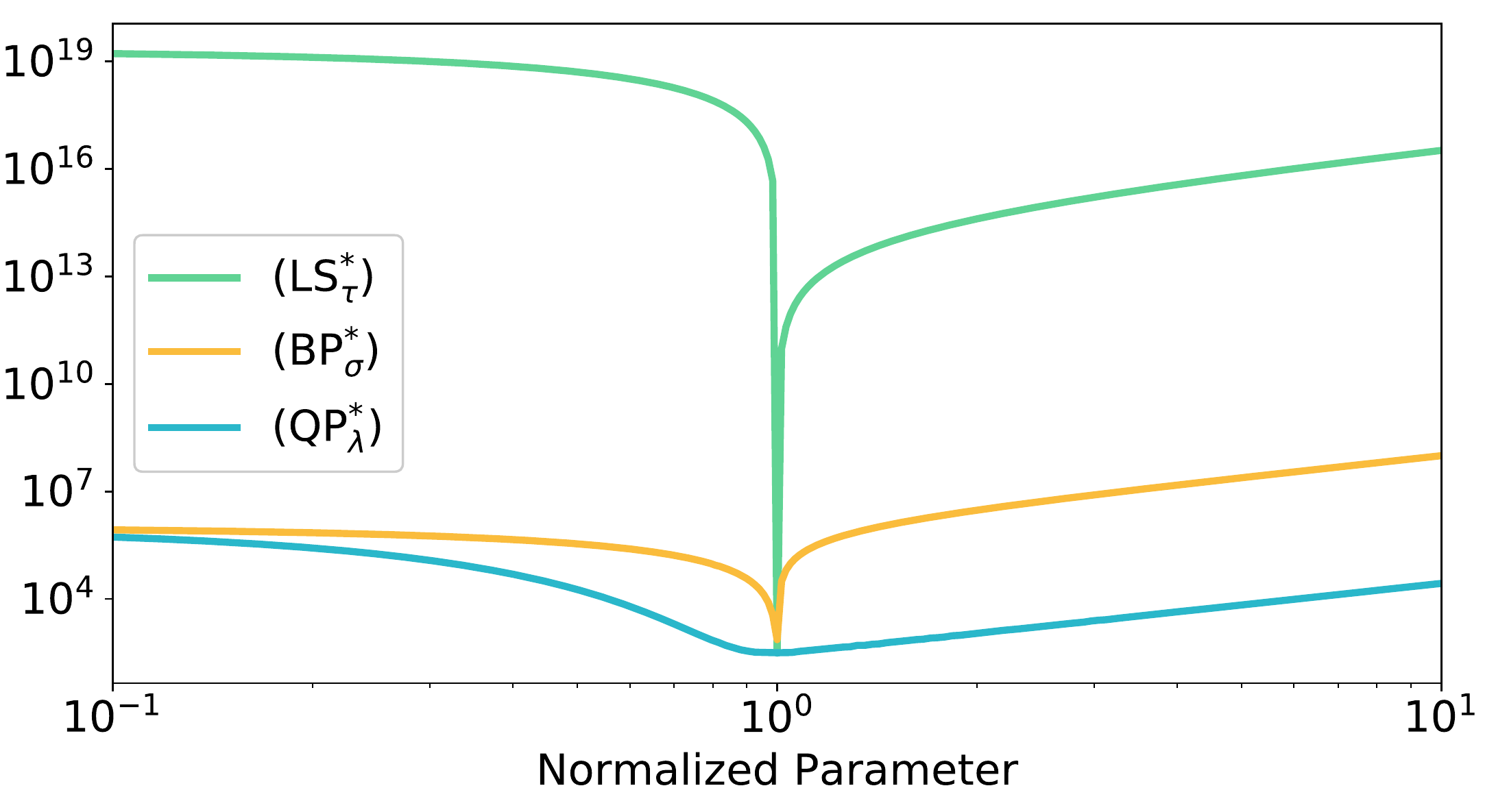}
    \subcaption{\label{fig:lspd-numerics-b}}
  \end{subfigure}
  \caption{%
    {\lspd} parameter instability in the low-noise regime. Average loss as per
    {\eqref{eq:avg-loss}} for each program plotted on a $\log$-$\log$ scale
    with respect to the normalized parameter. The data parameters for (a) are
    $(s, N, \eta, k, n) = (20, 10^{3}, 10^{-3}, 150, 301)$ and those for (b)
    are $(s, N, \eta, k, n) = (20, 10^{6}, 10^{-3}, 50, 301)$.}
  \label{fig:lspd-numerics}
\end{figure}

This section presents numerical simulations demonstrating parameter instability
of {\lspd} in the low-noise regime for two different ambient dimensions
$N = 10^{3}, 10^{6}$. This repetition has the benefit of showcasing the
behaviour of {\lspd} at two different sparsity levels, as well as contrasting
the behaviour of {\lspd} with {\qppd} and {\bppd} at relatively low and high
dimensions. Using the notation above, $n = 301$ points and $s = 20$;
$(k, N) = (150, 10^{3})$ for Figure \ref{fig:lspd-numerics-a}, while
$(k, N) = (50, 10^{6})$ for Figure \ref{fig:lspd-numerics-b}. In both regimes,
$x_{0}$ is quite sparse $s/N \sim 10^{-2}, 10^{-5}$ with entries that are well
separated from the noise $N/\eta \sim 10^{6}, 10^{9}$.

We may glean several pieces of information from these two plots. Most notably,
the {\lspd} parameter instability manifests in very low dimensions, relative to
practical problem sizes. Moreover, the curve for {\lspd} average loss seems to
approach something resembling the sharp asymptotic phase transition described by
Theorem \ref{thm:constr-pd}. One may also notice the behaviour of the other two
programs in the low-noise regime. It is apparent that the magnitude of the
derivative for the {\qppd} risk increases markedly on the left-hand side of the
optimal normalized parameter value (\ie below $1$) between the $N=10^{3}$ and
$N=10^{6}$ plots. This behaviour is consistent with the result in
Theorem \ref{thm:qppd-instability} that the left-sided risk scales as a power law of
$N$.

Finally, we observe that {\bppd} develops a shape resembling the instability of
{\lspd} when $N = 10^{6}$. We offer the plausible explanation that the relative
sparsity of the signal is small ($s/N = 2 / 10^{5}$) and thus this regime
coincides with the regime in which {\bppd} develops parameter
instability. Figure \ref{fig:bppd-numerics} demonstrates that such an instability
seems to occur in very large dimensions, a suspicion corroborated by the remark
at the end of \ref{ssec:bppd-numerics}.

We observe that the parameter instability developed by {\bppd} seems to manifest
in a way similar to that of {\lspd}. This is interesting, because
Theorem \ref{thm:bppd-maximin} shows that there is no good choice of parameter
$\sigma$, though Figure \ref{fig:bppd-numerics} supports that there is a single
best choice, albeit minimax suboptimal, when $N$ is moderately large.

\subsection{{\qppd} analytic plots}
\label{sec:qppd-numerics}

We plot $R^{\sharp}(\lambda; s, N)$ using the expressions derived in
\eqref{eq:pd-lambda}. Observe that the plot of the analytic expression for
$R^{\sharp}(\lambda; s, N)$ agrees well with the simulations of
$R^{\sharp}(\lambda; x_{0}, N, \eta)$ in Figure \ref{fig:lspd-numerics} and
Figure \ref{fig:bppd-numerics}.

In Figure \ref{fig:qppd-instability-a} we plot $R^{\sharp}(\lambda; s, N)$ for
$\lambda \in \{ 1 - 10^{-2}, 1- 10^{-3}, 2\}$. It is evident from the reference
lines $y \sim N^{2/5}$ and $y \sim \sqrt N$ that
$R^{\sharp}(u \bar \lambda; s, N)$ scales like a power law of $N$ for $u < 1$,
while $R^{\sharp}(2 \bar \lambda; s, N)$ appears to have approximately
order-optimal growth. The derivatives of these three functions are visualized in
Figure \ref{fig:qppd-instability-c}, with plotted reference lines
$y = N^{2/5}, \sqrt N$. Again, it is evident that the derivative scales as a
power law of $N$ for those risks with $\lambda < \bar \lambda$. In
Figure \ref{fig:qppd-instability-b} we plot $R^{\sharp}(\lambda; s, N)$ as a
function of $\lambda$ for $N = 10^{15}$. One may observe parameter instability
for $\lambda < \lambda^{*}$, for example by comparison to the plotted reference
line $y \sim \lambda^{-32}$. Similarly, one may observe right-sided parameter
stability of $R^{\sharp}(\lambda; s, N)$ by comparing with the second plotted
reference line, $y \sim \lambda^{2}$. From these simulations one may observe
that choosing $\lambda = .5 \lambda^{*}$ accrues at least a $10^{9}$ fold
magnification of the error.

Finally, we would like to clarify a potentially confusing issue. Though our
theory for {\qppd} refers to $\lambda^{*}$ only through its connection with
$\bar \lambda$, we were able to approximate $\lambda^{*}$ empirically in our
numerical simulations. Accordingly, we have made reference to it when discussing
parameter stability regimes.

\begin{figure}[t]
  \begin{subfigure}{.47\linewidth}
    \includegraphics[width=\linewidth]{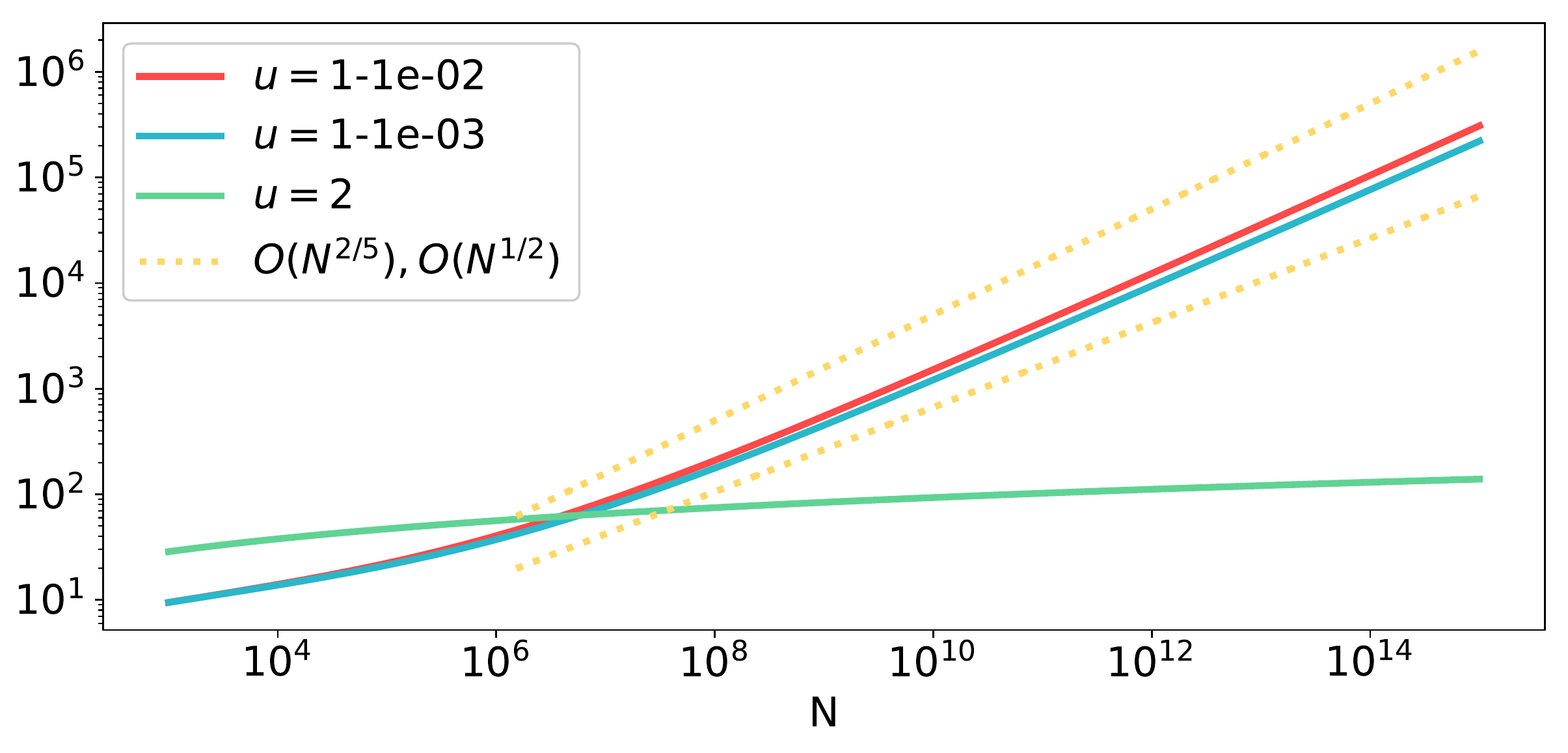}
    \subcaption{\label{fig:qppd-instability-a}}
  \end{subfigure}
  \qquad
  \begin{subfigure}{.47\linewidth}
    \includegraphics[width=\linewidth]{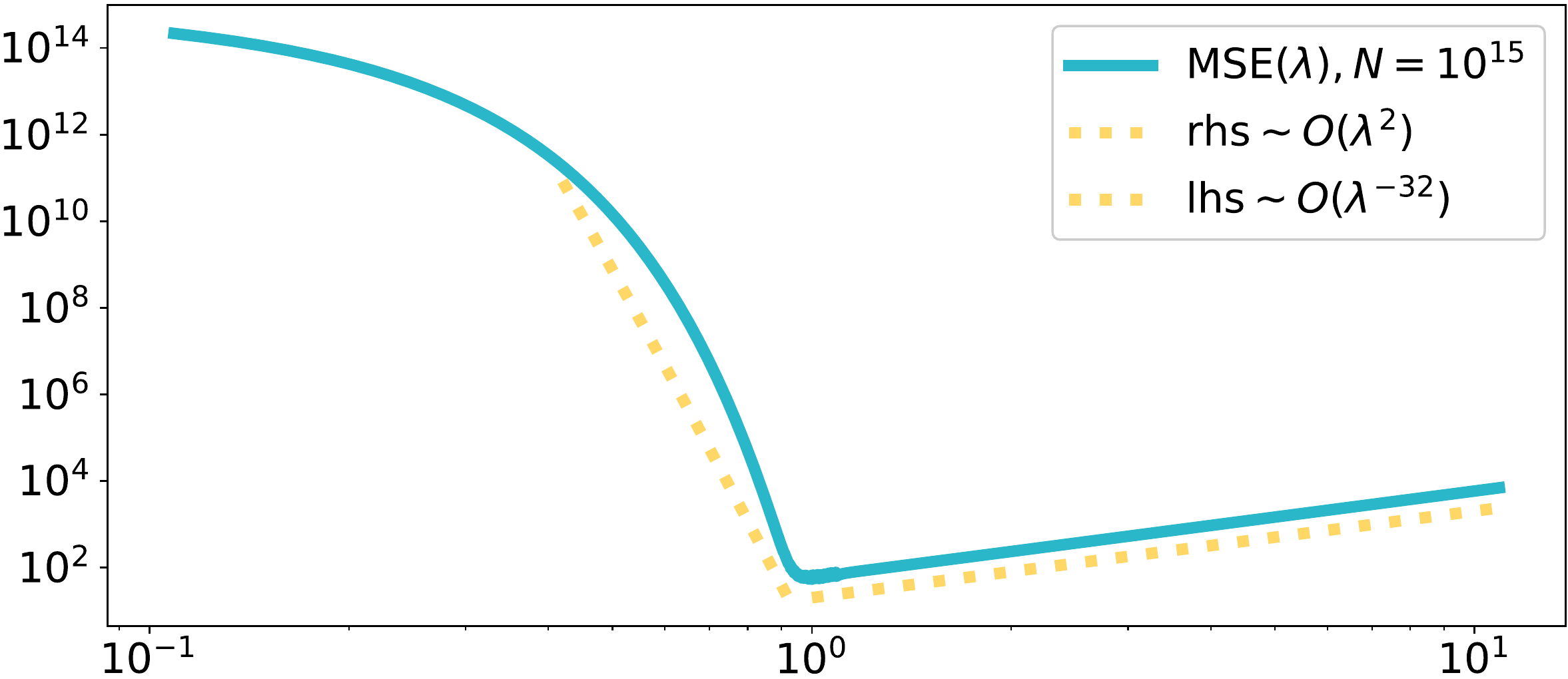}
    \subcaption{\label{fig:qppd-instability-b}}
  \end{subfigure}

  \begin{subfigure}{.47\linewidth}
    \includegraphics[width=\linewidth]{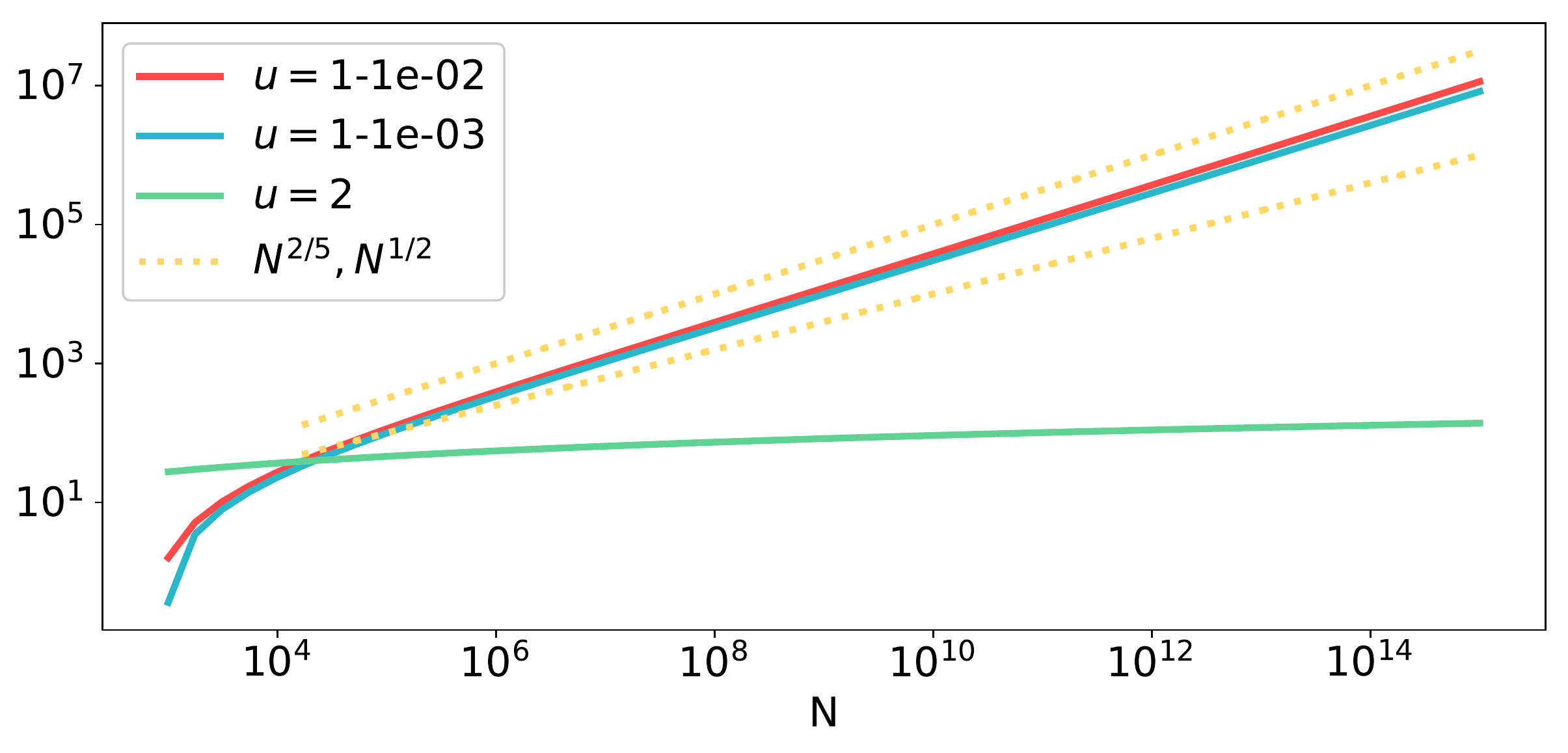}
    \subcaption{\label{fig:qppd-instability-c}}
  \end{subfigure}

  \caption{{\qppd} parameter instability in the low-noise regime. All curves are
    generated analytically using the expressions obtained in
    section \ref{sec:param-inst-qppd} and plotted on a $\log$-$\log$ scale. (a) A
    plot of $R^{\sharp}(u \lambda^{*}; s, 10^{15})$ as a function of $N$ for
    $u \in \{ 1- 10^{-2}, 1-10^{-3}, 2\}$. The lines $y = N^{2/5}/15$,
    $y = \sqrt N / 20$ are plotted for reference. (b) A plot of
    $R^{\sharp}(\lambda; s, 10^{15})$ as a function of $\lambda$. Two lines are
    plotted as reference for risk growth rate with respect to $\lambda$. (c) A
    plot of the magnitude of $\dee{}{u} R^{\sharp}(u\lambda^{*}; s, 10^{15})$ as
    a function of $N$ for $u \in \{ 1 - 10^{-2}, 1- 10^{-3}, 2\}$. The lines
    $y = N^{2/5}$, $y = \sqrt N / 20$ are plotted as reference.}
  \label{fig:qppd-instability}
\end{figure}

\subsection{{\bppd} numerical simulations}
\label{ssec:bppd-numerics}

This section presents numerical simulations demonstrating parameter instability
of {\bppd} in the regime where $x_{0}$ is very
sparse. Figure \ref{fig:bppd-numerics-a} is generated as described by the above
procedure in section \ref{sec:numerical-results}, with parameters
$(N, s, \eta, k, n) = (10^{7}, 1, 1, 10, 237)$, while
Figure \ref{fig:bppd-numerics-b} was generated in a way that mirrors the proof of
Theorem \ref{thm:bppd-maximin}, with parameters $(s, \eta, k, n) = (1, 1, 25, 31)$.

The thrust of Figure \ref{fig:bppd-numerics-a} is to resolve parameter instability
of {\bppd} about the optimal parameter choice. Because the theory suggests that
$\tilde R(\sigma; x_{0}, N, \eta)$ is surely resolved when the ambient
dimension is sufficiently large, we set $N = 10^{7}$; this value was expected
to resolve the instability, as per the discussion in
\ref{sssec:theorem-parameter-simulation} below. We limited the number of
realizations and grid points because the problem size was computationally
prohibitive. The minimal average loss observed on the plot was significantly
larger than the respective minimal average losses of {\lspd} and {\qppd} by a
factor of $82.2$, supporting the theory. We also noticed a cusp-like behaviour,
which would be an interesting object of further study.

Figure \ref{fig:bppd-numerics-b} was generated so as to mirror the theory backing
Theorem \ref{thm:bppd-maximin}. Specifically, noise realizations were constrained
to the constant probability event
$\{ \|z\|_{2}^{2} - N \in (.5\sqrt N, 5\sqrt N)\}$. Plotted in the figure is
the average best loss as a function $N$,
\begin{align*}
  \bar L_{\text{best}}(N; x_{0}, \eta, k, n) %
  := \frac{1}{k} \sum_{j=1}^{k} \min_{i \in [n]} %
  L(\sigma_{i}(N); x_{0}, N, \eta \hat z_{ij}). 
\end{align*}
The domain for $N$ ranges from $10^{2}$ to $10^{7}$, computed on a
logarithmically spaced grid composed of $51$ points. For each value $N$ in the
grid, the average loss was computed for $n = 31$ values of $\sigma$, each using
$k = 25$ realizations $\hat z$ of the noise. The standard deviations of the best
loss realizations were computed, and plotted as a grey ribbon about the average
best loss. Included for reference is a smoothed version of the average best
loss, computed as a rolling window average. In addition, we have plotted two
power laws of $N$ that lower- and upper-bound the averaged best loss, and nearly
bound that quantity up to a full standard deviation.

\begin{figure}[t]
  \begin{subfigure}{.47\linewidth}
    \centering
   \includegraphics[width=\linewidth]{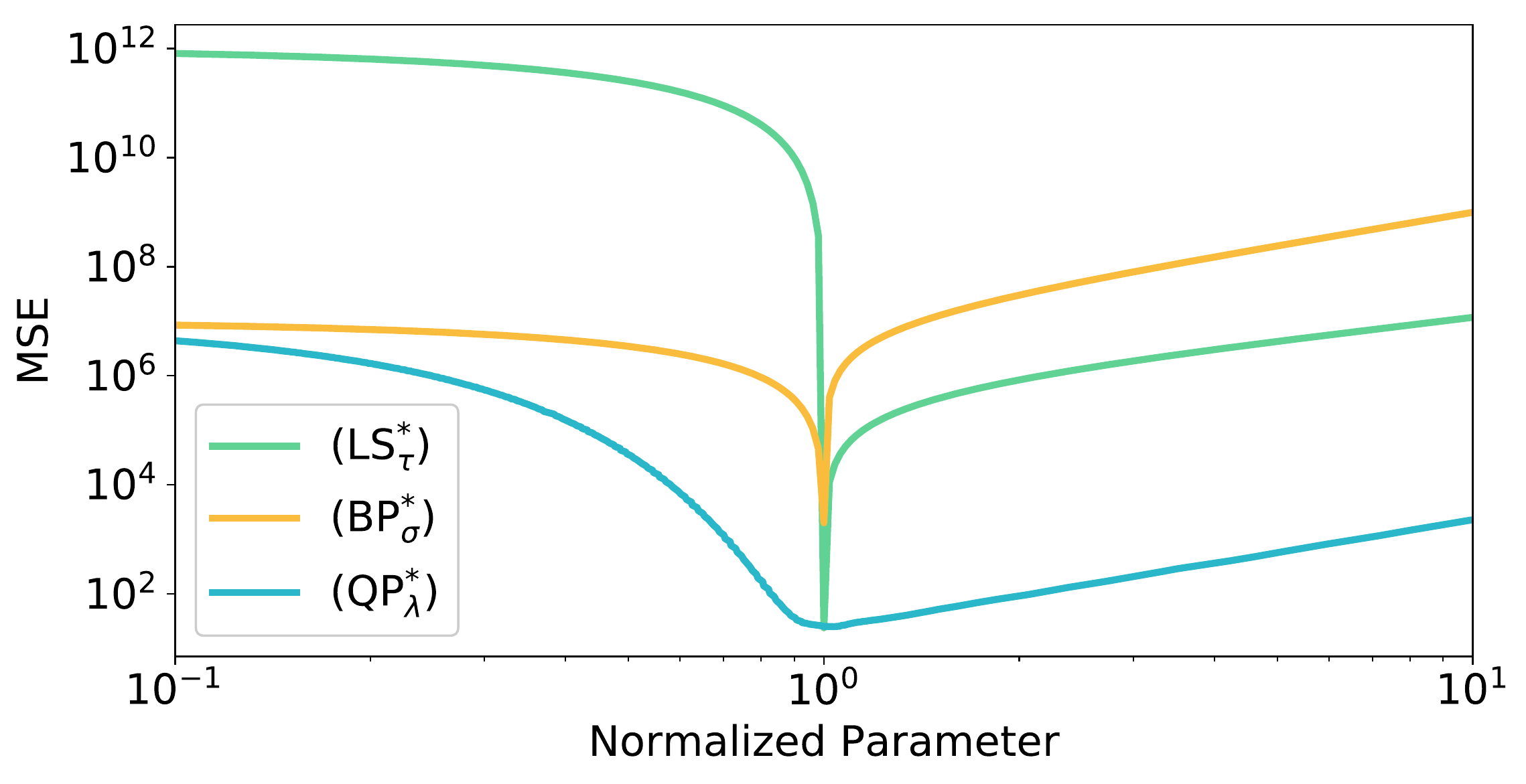}
    \subcaption{\label{fig:bppd-numerics-a}}
  \end{subfigure}
  \begin{subfigure}{.47\linewidth}
    \centering
    \includegraphics[width=\linewidth]{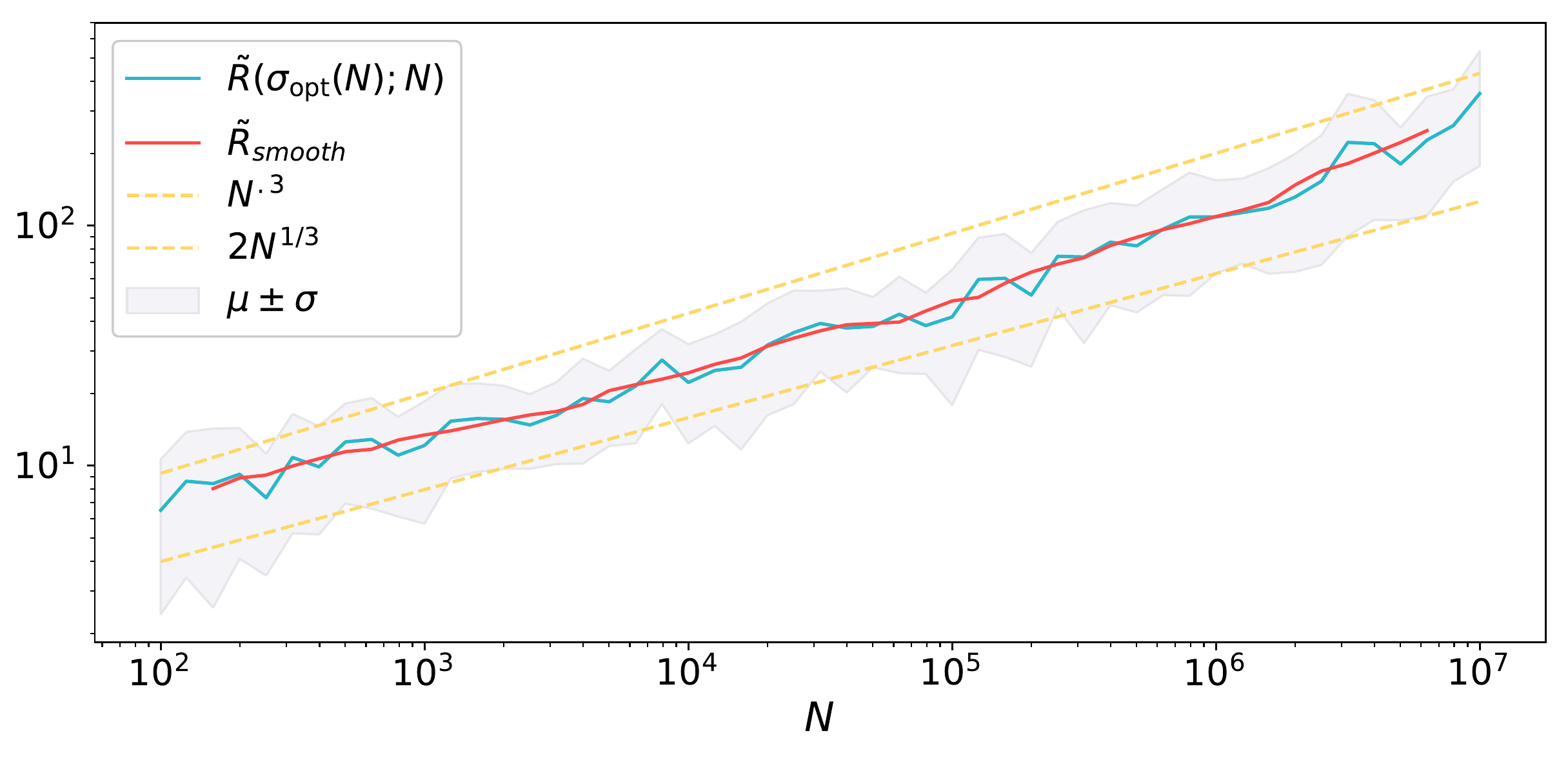}
    \subcaption{\label{fig:bppd-numerics-b}}
\end{subfigure}
\caption{%
  {\bppd} parameter instability in the very sparse regime.
  (a) Data parameters: $(s, N, \eta, k, n) = (1, 10^{7}, 1, 10, 237)$. Average
  losses plotted on a $\log$-$\log$ scale with respect to the normalized
  parameter. 
  (b) Average best loss for {\bppd} as a function of $N$. Data parameters:
  $(s, \eta, k, n) = (1, 1, 25, 31)$. The function $\sigma_{\text{opt}}(N)$ was
  obtained as the value of $\sigma$ bestowing minimal loss of the program for
  each $N$ and realization. These best losses were averaged, yielding average
  best loss. The standard deviation was computed for each $N$ from the same loss
  realizations and included as a ribbon about the mean; functions $y = N^{3/10}$
  and $y = 2 N^{1/3}$ are included for reference.}
  \label{fig:bppd-numerics}
\end{figure}

\subsubsection{Simulating theorem parameters}
\label{sssec:theorem-parameter-simulation}

Here we clarify the relationship between some of the constants appearing in the
proofs of Theorem \ref{thm:bppd-minimax} and Theorem \ref{thm:bppd-maximin}. We provide
two examples of minimal $N_{0}$ values guaranteeing parameter instability
behaviour of {\bppd} for given parameter choices. The theory does not claim
these values to be optimal, nor do we claim that the constants are tuned. In
particular, these demonstrations seem rather pessimistic, especially by
comparison with the numerical simulations in Figure \ref{fig:bppd-numerics}.

The following values were determined by computing
$N_{0} := \max\{ N_{0}^{\eqref{prop:bppd-oc-1a}}(a_{1}, C_{1}, L),
N_{0}^{\eqref{prop:bppd-oc-2a}}(C_{2}, L)\}$ for particular choices of
$a_{1}, C_{1}, C_{2}$ and $L$, using their definitions in the technical results
of \ref{sssec:supporting-propositions-for-bppd-oc-3a}. Thus, the theory of
section \ref{sec:param-inst-bppd} guarantees parameter instability for all
$N \geq N_{0}$ when
\begin{align*}
  N_{0} \approx 1.5\rm{e}6 &\quad\text{and}\quad (a_{1}, C_{1}, C_{2}, L) \approx (1.45, 5, 4, 3.78) \qquad\text{or}\\
  N_{0} \approx 4.9\rm{e}5 &\quad\text{and}\quad (a_{1}, C_{1}, C_{2}, L) \approx (1.58, 4.04, 4, 3.62). 
\end{align*}
These numbers appear pessimistic, given that $N_{0}$ is large, while
$(C_{2}, C_{1}) \approx (4, 5)$ implies the instability arises on the event
$\{ \|z\|_{2}^{2} - N \in (4 \sqrt N, 5 \sqrt N)\}$, which occurs with
relatively minute (but constant) probability. Thus, it may not be all that
surprising that {\bppd} suboptimality is difficult to ascertain empirically from
a small number of realizations in only moderately large dimension when
$\sigma \approx \sigma^{*}$.

\subsection{Parameter stability in sparse proximal denoising}
\label{sec:param-stab-sparse}

In this section we show numerical simulations in which the three programs appear
to exhibit better parameter stability. For these simulations,
$\eta \approx 233.0$, $s = 2500$ and $N = 10^{4}$. Average loss was computed
from $k = 25$ realizations for $n = 401$ grid points. As the noise is large,
this setting lies (mostly) outside the regime in which {\lspd} and {\qppd}
exhibit parameter instability. Moreover, the signal is not very sparse, since
$s/N = .25$. Thus, this setting also lies outside the regime in which {\bppd}
exhibits parameter instability. Accordingly, smooth risk curves are seen for
{\bppd} and {\qppd}. While {\qppd} and {\bppd} appear relatively gradual,
{\lspd} appears at least to avoid a cusp-like point about $\tau / \tau^{*} = 1$. These data visualized in Figure \ref{fig:parameter-stability}. 

\begin{figure}[t]
  \begin{minipage}{.47\linewidth}
    \includegraphics[width=\linewidth]{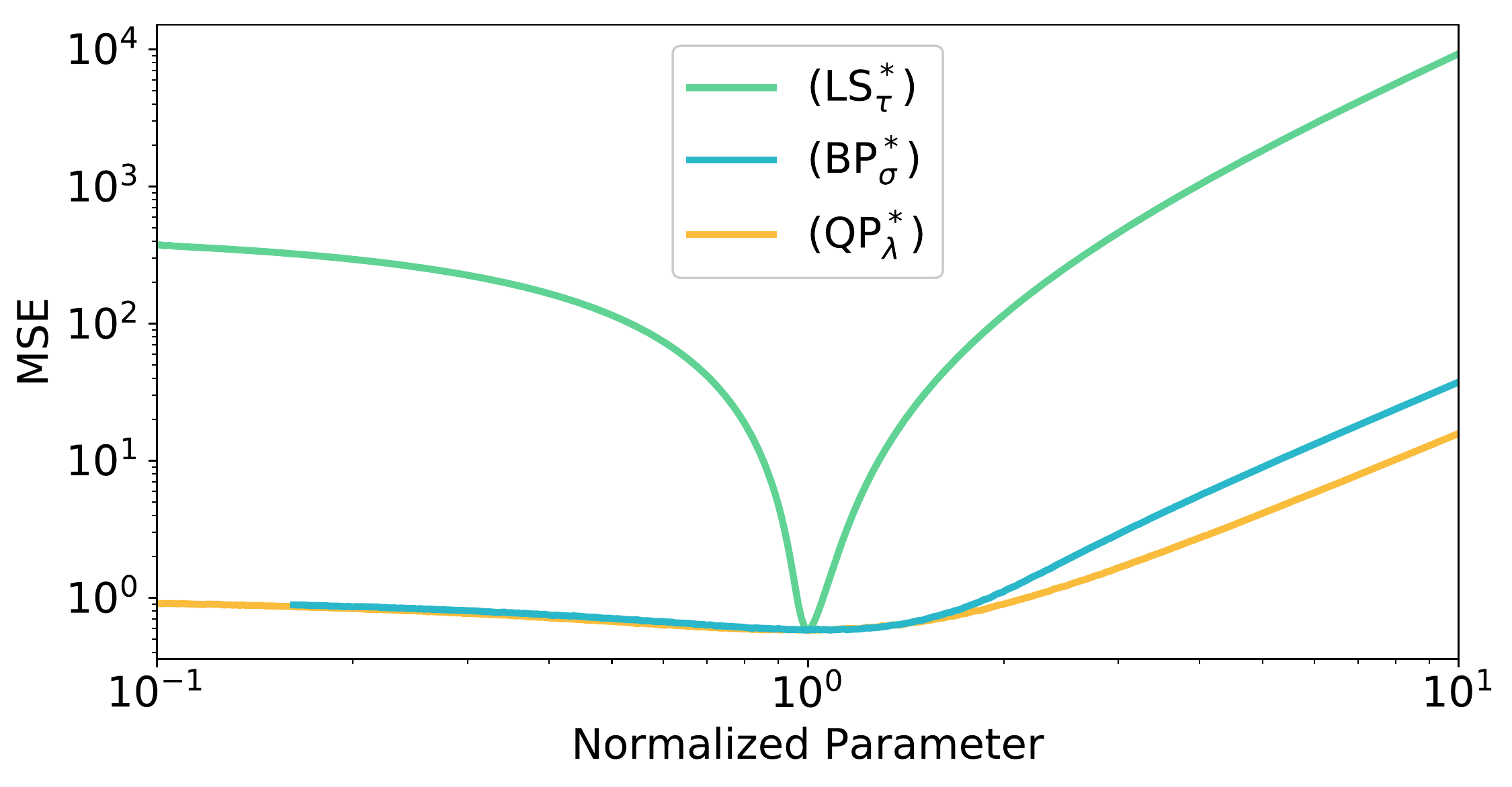}
  \end{minipage}
  \begin{minipage}{.47\linewidth}
    \caption{Parameter stability of sparse PD programs when
      $(s, N, \eta, k, n) = (2500, 10^{4}, 233, 25, 401)$. Plotted curves
      represent average loss, plotted on a $\log$-$\log$
      scale. \label{fig:parameter-stability}}
  \end{minipage}

\end{figure}

\subsection{Realistic denoising examples}
\label{sec:realistic-denoising}

\subsubsection{Image-space denoising}
\label{sec:image-space-deno}

We visualize how proximal denoising behaves for a realistic denoising
problem. The ground truth signal is the standard $512\times 512 \times 3$
colour image of a mandrill face, ravelled to a vector
$x_{0} \in [0,1]^{N} \subseteq \reals^{N}$, $N = 786\,432$. The denoising is
performed in image space. Specifically, the signal $x_{0}$ is not sparse:
$99.98 \%$ of its coefficients are nonzero.  We set
$y_{j} = x_{0, j} + \eta z_{j} \in \reals^{N}$ where $\eta = 10^{-5}, 1$ and
$z_{j} \iid \mathcal{N}(0, 1)$. The results of this example are displayed in
Figure \ref{fig:mandrill}: the ground truth and noisy images in the top row, and
quantitative results captured by plots of the average loss \eqref{eq:avg-loss}
in the bottom row.

The plots of average loss were generated from $k = 25$ realizations of noise
$z$, with a logarithmically spaced grid of $n = 501$ points centered about the
optimal parameter value for each of the three proximal denoising programs. The
optimal parameter value for each program was determined analytically where
possible, or numerically using standard solvers \cite{scipy}. A smooth
approximating curve of the non-uniformly spaced point cloud of loss
realizations was computed using radial basis function approximation. The RBF
approximation used multiquadric kernels with parameters
$(\varepsilon_{\text{rbf}}, \mu_{\text{rbf}}, n_{\text{rbf}}) = (10^{-3},
10^{-2}, 301)$. Here, $\varepsilon_{\text{rbf}}$ is the associated RBF scale
parameter, $\mu_{\text{rbf}}$ is a smoothing parameter and $n_{\text{rbf}}$ is
the number of grid points at which to approximate \cite{scipy}. The RBF
parameters for the approximation were selected so as to generate a smooth line
that best represents the path about which the individual (noisy) data points
concentrate.

About the optimal average loss (where the normalized parameter is $1$), an
average difference of $1.382 \%$ in the value of $\tau$ results in a
$4.694\times 10^{5}$ fold difference in nnse on average when $\eta =
10^{-5}$. In contrast, that error varies by no more than a factor of three in
the large noise regime ($\eta = 1$). Moreover, we observe that the average
losses computed for $\eta = 10^{-5}$ upper bound those computed for $\eta =
1$. These results suggest not to use {\lspd} for proximal denoising when $\eta$
is small, even when the underlying data are not sparse.

\begin{figure}[t]
  \centering
  \includegraphics[width=.7\textwidth]{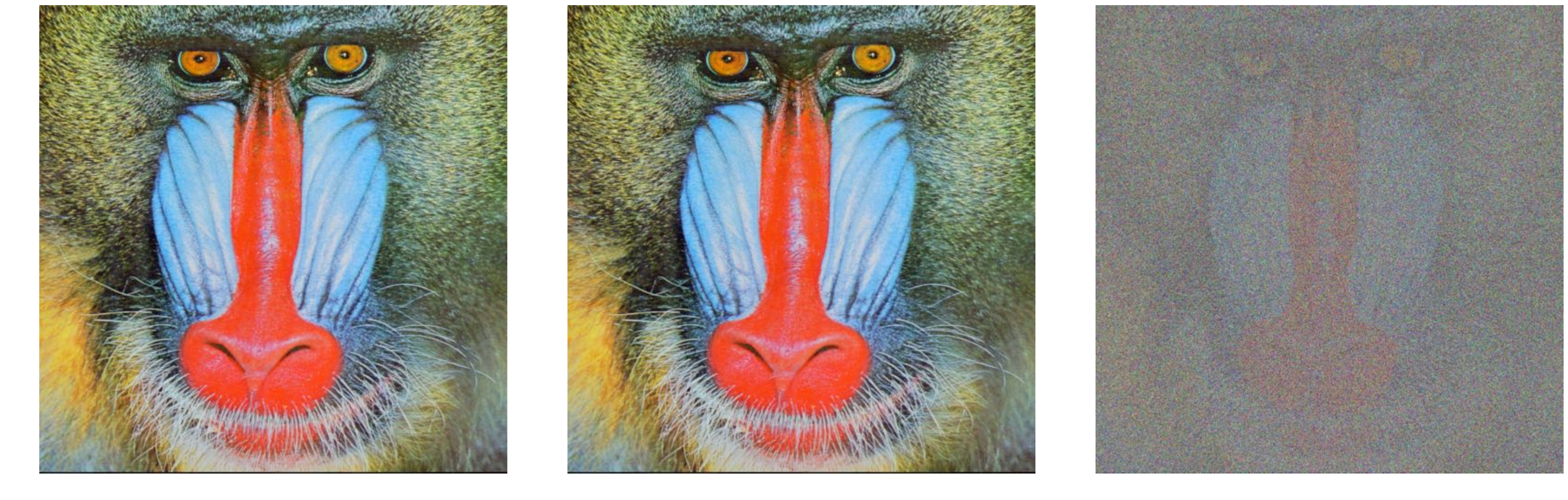}

\raisebox{-.5\height}{\includegraphics[width=.45\textwidth]{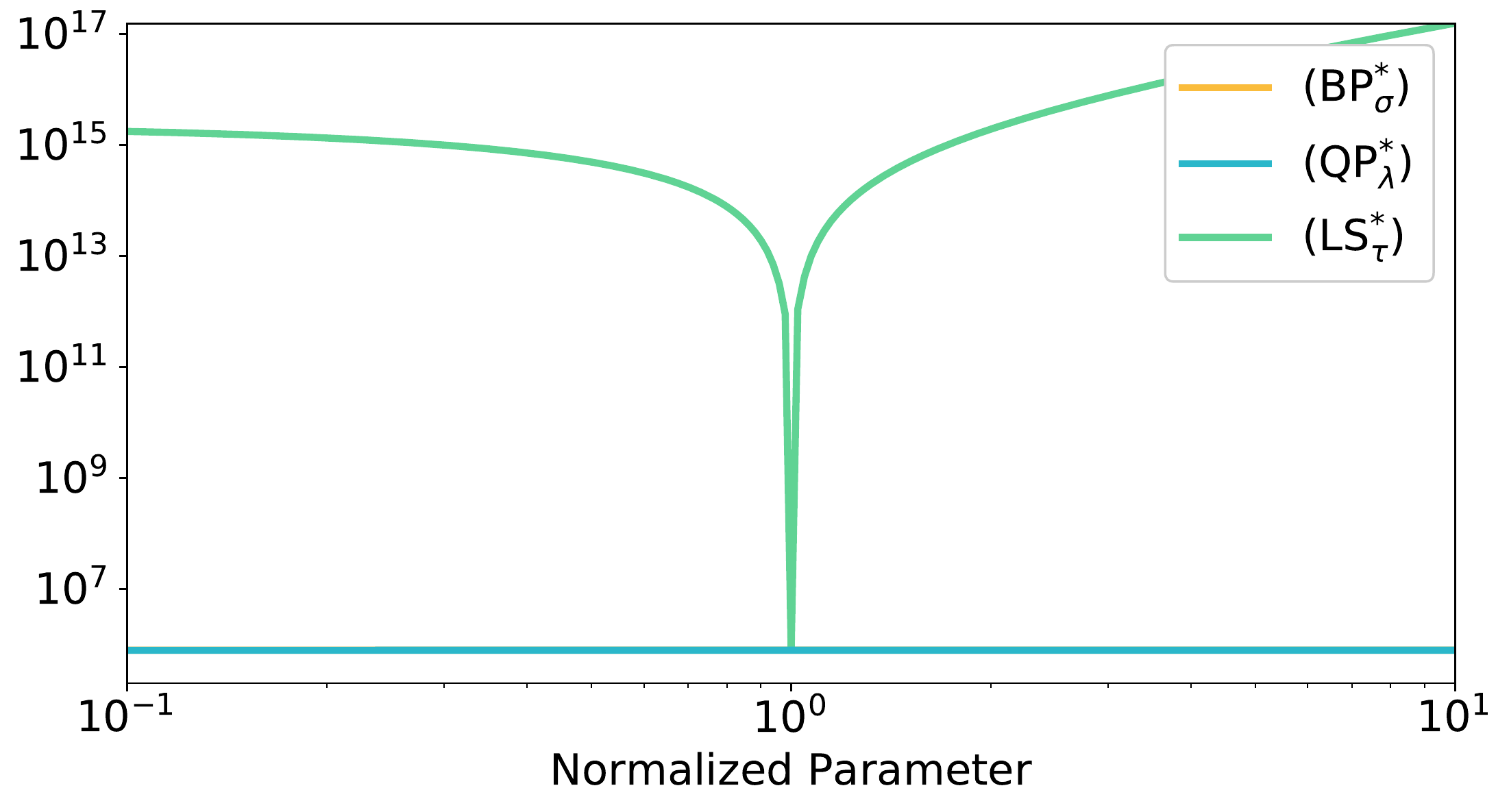}}
\raisebox{-.5\height}{\includegraphics[width=.45\textwidth]{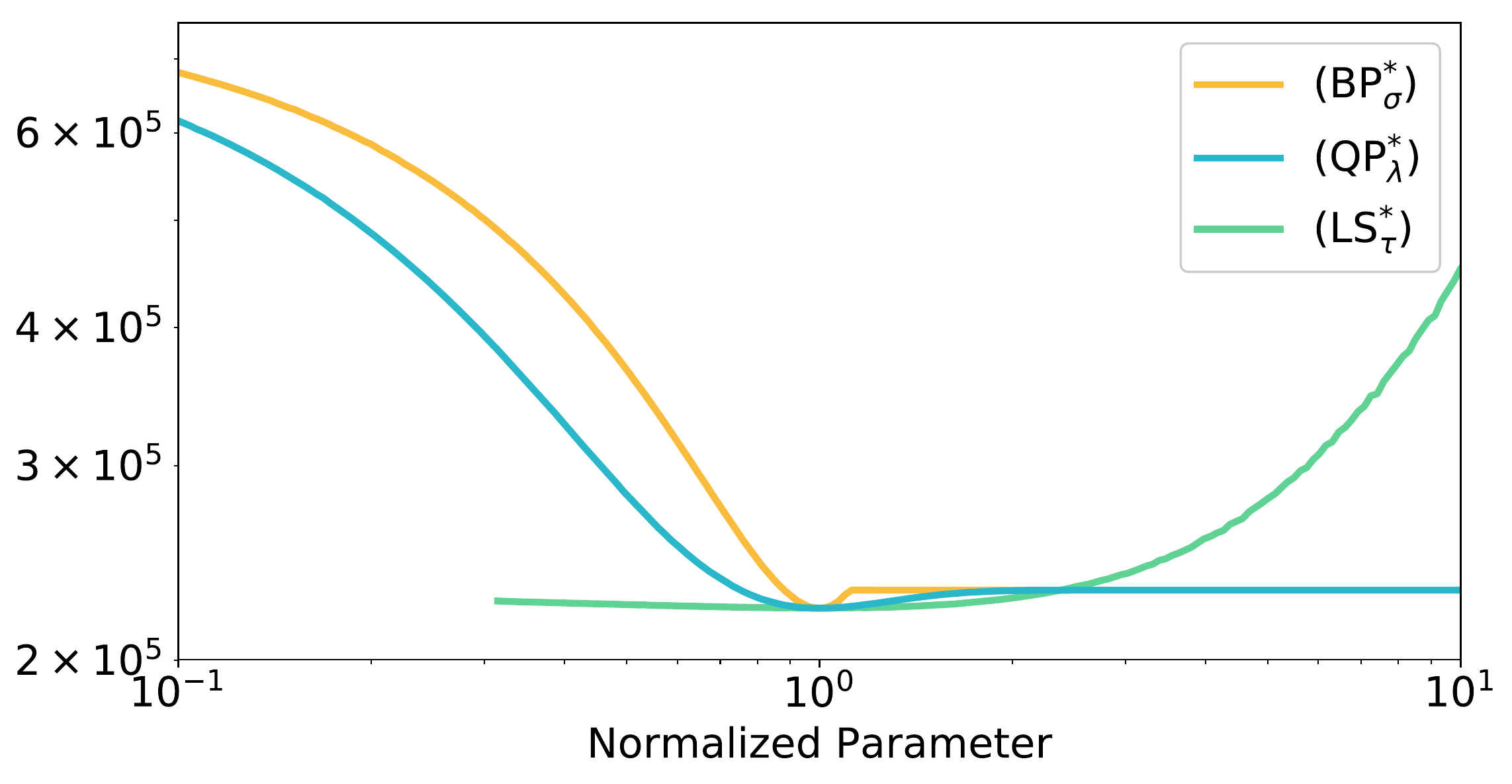}}\qquad
  
\caption{\textbf{Top (left-to-right):} The underlying signal is the
  $512 \times 512 \times 3$ mandrill image; the middle image is corrupted by
  iid normally distributed noise ($\eta = 10^{-5}$); the right-most image is
  corrupted by iid normally distributed noise ($\eta = 1$). The pixel values of
  the original image lie in $[0,1]^{3}$; those of the noisy images are scaled
  to this range for plotting. \textbf{Bottom:} Average loss is plotted with
  respect to the normalized parameter for {\lspd}, {\qppd} and {\bppd}
  respectively when $\eta = 10^{-5}$ (left) and $\eta = 1$ (right). The
  associated parameters are $(N, k, n) = (786\,432, 25, 501)$. Plotted lines
  are smoothed approximations of loss realization data using multiquadric RBFs.
}
\label{fig:mandrill}
\end{figure}

\subsubsection{1D denoising example}
\label{sec:1d-denoising-example}

In this section, we demonstrate parameter instability regimes for a realistic
example of a 1D signal using wavelet domain denoising. Specifically, an
$s$-sparse 1D signal $x_{0} \in \reals^{N}$ was generated in the Haar wavelet
domain, where $(s, N) = (10, 4096)$. In the signal domain, iid normal random
noise was added to the signal to generate
$\mathcal{W}^{-1}y := \mathcal{W}^{1}x_{0} + \eta z$ where
$\eta = \frac{N}{100}, \frac{N}{10}$. The denoising problem was solved in the
wavelet domain on a grid of size $501$ centered about the optimal normalized
parameter and logarithmically spaced. Namely, the input to each program was
$y$. The loss was computed in the signal domain after applying the inverse
transform to the estimated solution:
\begin{align*}
  L(\rho; x_{0}, N, \eta \hat z) %
  := \eta^{-2} \|\mathcal{W}^{-1} (x^{*}(\rho) - x_{0})\|_{2}^{2}
\end{align*}
A smooth approximation to the average loss
$\bar L(\rho_{i}; x_{0}, N, \eta, k)$ was computed from $k = 25$ realizations
of the noise using linear radial basis function approximation with
parameters $(\mu_{\text{rbf}}, n_{\text{rbf}}) = (0.01, 501)$.

In Figure \ref{fig:realistic-example-1d-plots}, we visualize how the three
programs behave for denoising a 1D signal, sparse in the Haar wavelet domain,
which has been corrupted by one of two different noise levels in the signal
domain. The top row visualizes the ground truth signal with a realization of
the corrupted signal for $\eta = N / 100$ (top-left) and $\eta = N / 10$
(top-right). The bottom row visualizes the average loss with respect to the
normalized parameter of each program. In the high-noise regime (bottom-right),
it is clearly seen that {\bppd} is the most parameter unstable about the
optimal parameter choice. Moreover, the best average loss for {\bppd} is
greater than that for {\qppd} or {\lspd}, as suggested by the supporting
theory. We note that {\qppd} also has an average loss greater than the minimal
one, and suggest --- noting the local variability in the curve --- that this is
an artifact of the RBF approximation through the optimality region. In the
moderate-noise regime, we see a situation in which {\qppd} appears to be the
most parameter stable --- again consistent with our reasoning that
unconstrained programs should exhibit better stability. In contrast, {\lspd} is
most parameter unstable below the optimal parameter, while it is {\bppd} that
is most parameter unstable above the optimal parameter. This behaviour may be
indicative of a regime intermediate to those we have previously discussed (\ie
lying between strictly low-noise and strictly very sparse).

With the grid in Figure \ref{fig:realistic-example-1d}, we intend to elucidate how
parameter instability manifests for each program as a function of the
normalized parameter, by visualizing the recovered signal for different values
of the normalized parameter. The top plot shows the same average losses that
are plotted in the bottom-left of Figure \ref{fig:realistic-example-1d-plots}. The
dotted lines at $\rho = .5, .75, 1, 4/3, 2$ and the markers located
approximately at the intersection of these lines with the loss curves visualize
sections of the loss for which the solution to the program will be
visualized. Indeed, for each value of $\rho$ and each program, there is a
corresponding plot in the grid that depicts the solution to the program for
that normalized parameter value $\rho$, along with the original signal $x_{0}$,
which is depicted as a black dotted line in each of the 15 plots. When $\rho$
is too small for {\bppd} and {\qppd}, it is clear that the noise fails to be
thresholded away. In contrast, this occurs for {\lspd} when $\rho$ is too
large. On the other hand, the signal content is thresholded away by {\bppd}
when $\rho$ too large, and by {\lspd} when $\rho$ is too small. Notice that
this behaviour does not seem to occur with {\qppd}, further supporting that
{\qppd} admits right-sided parameter stability.

\begin{figure}[h]
  \centering
  \includegraphics[width=.45\textwidth]{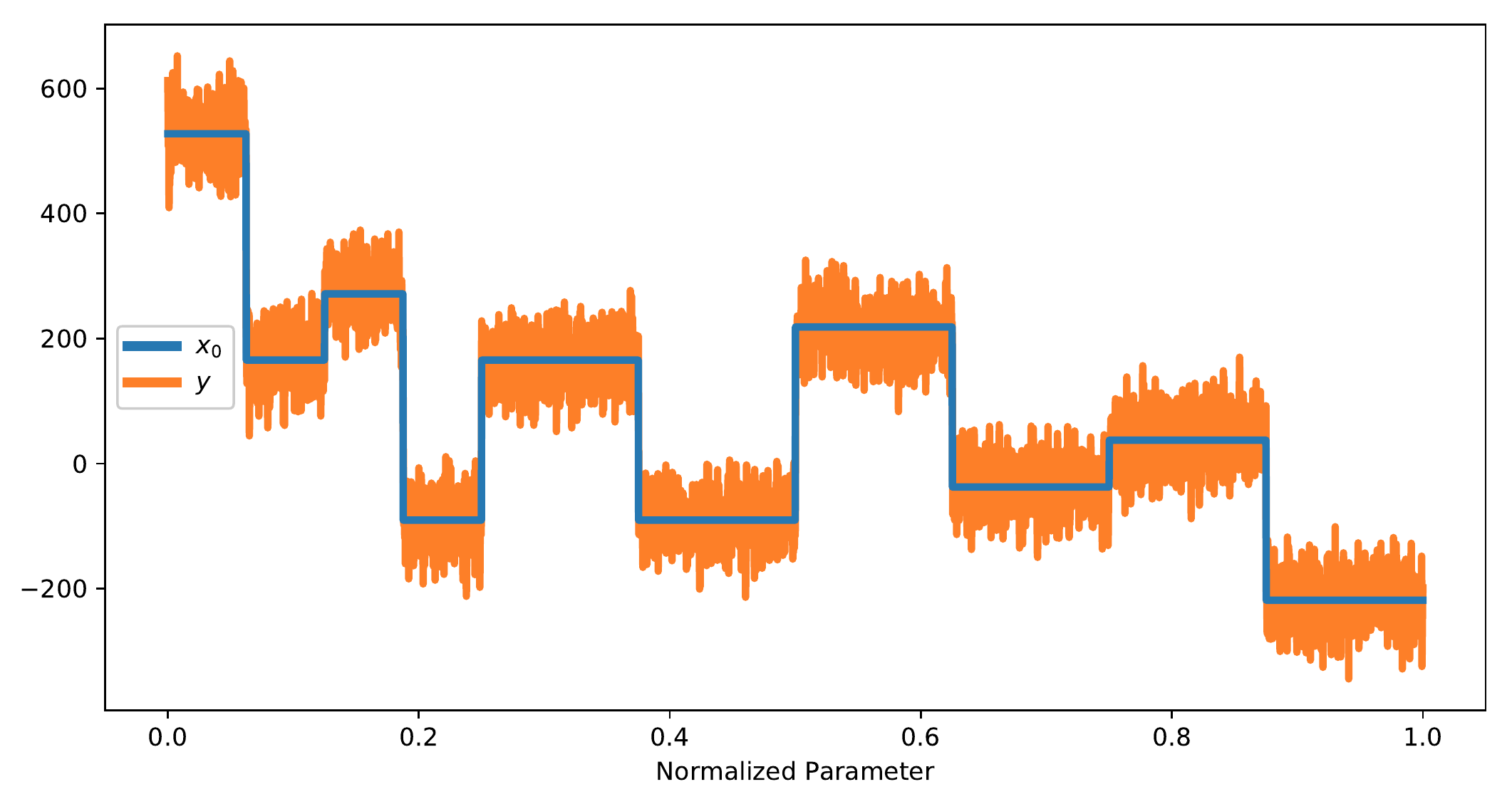}\qquad
  \includegraphics[width=.45\textwidth]{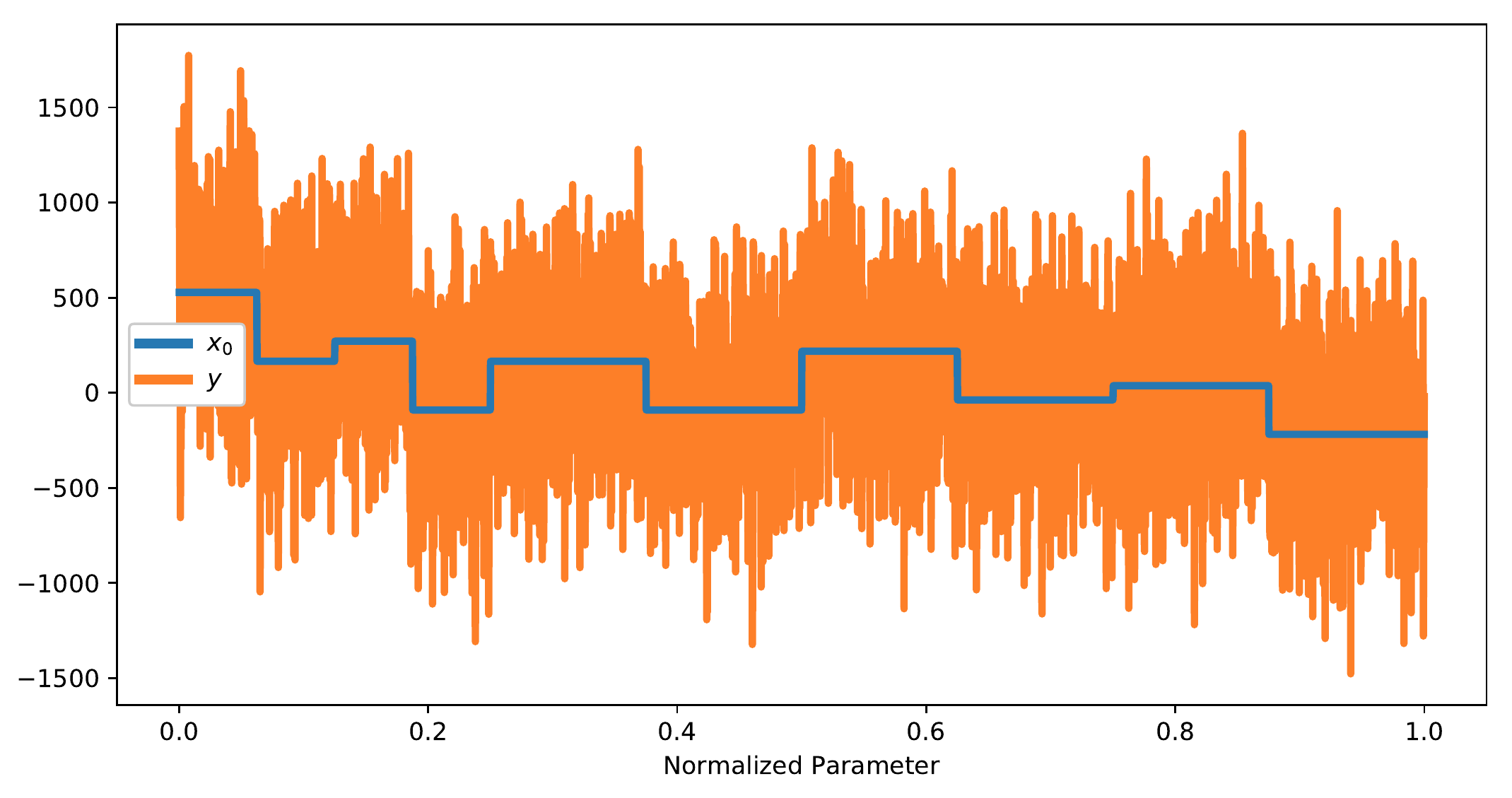}

  \includegraphics[width=.45\textwidth]{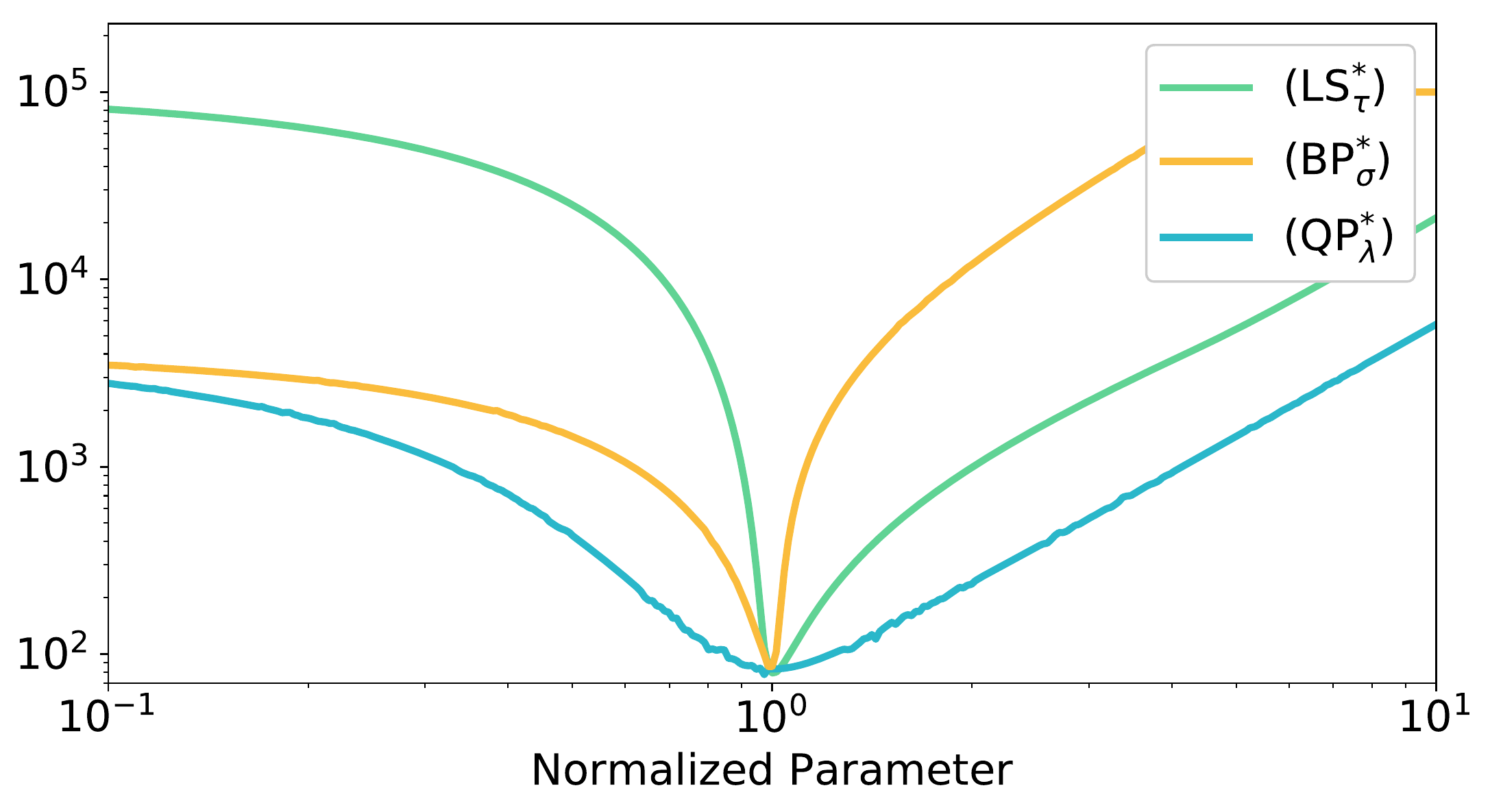}\qquad
  \includegraphics[width=.45\textwidth]{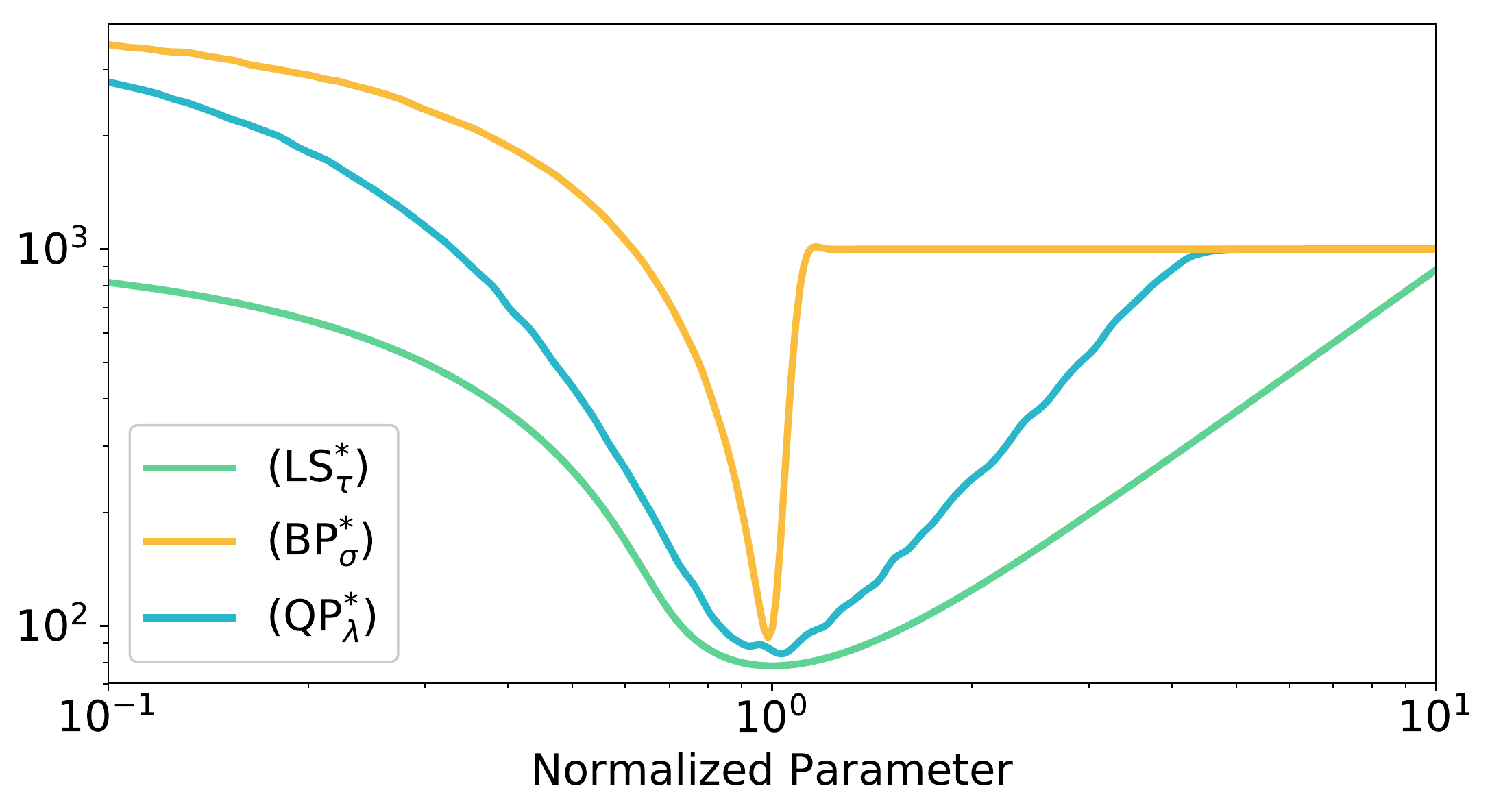}
  
  \caption{Haar wavelet space denoising of a 1D signal that is sparse in the
    Haar wavelet domain for two different noise levels, $\eta = N/100$ (left
    column) and $\eta = N/10$ (right column). \textbf{Top:} each plot contains
    a realization of the noisy 1D signal plotted in orange with the ground
    truth signal in blue. \textbf{Bottom:} visualizations of the average loss
    curve, computed using RBF approximation. The parameters for this example
    are: $(s, N, k, n) = (10, 4096, 25, 501)$.}
  \label{fig:realistic-example-1d-plots}
\end{figure}

\begin{figure}[h]
  \centering
  \includegraphics[width=.5\textwidth]{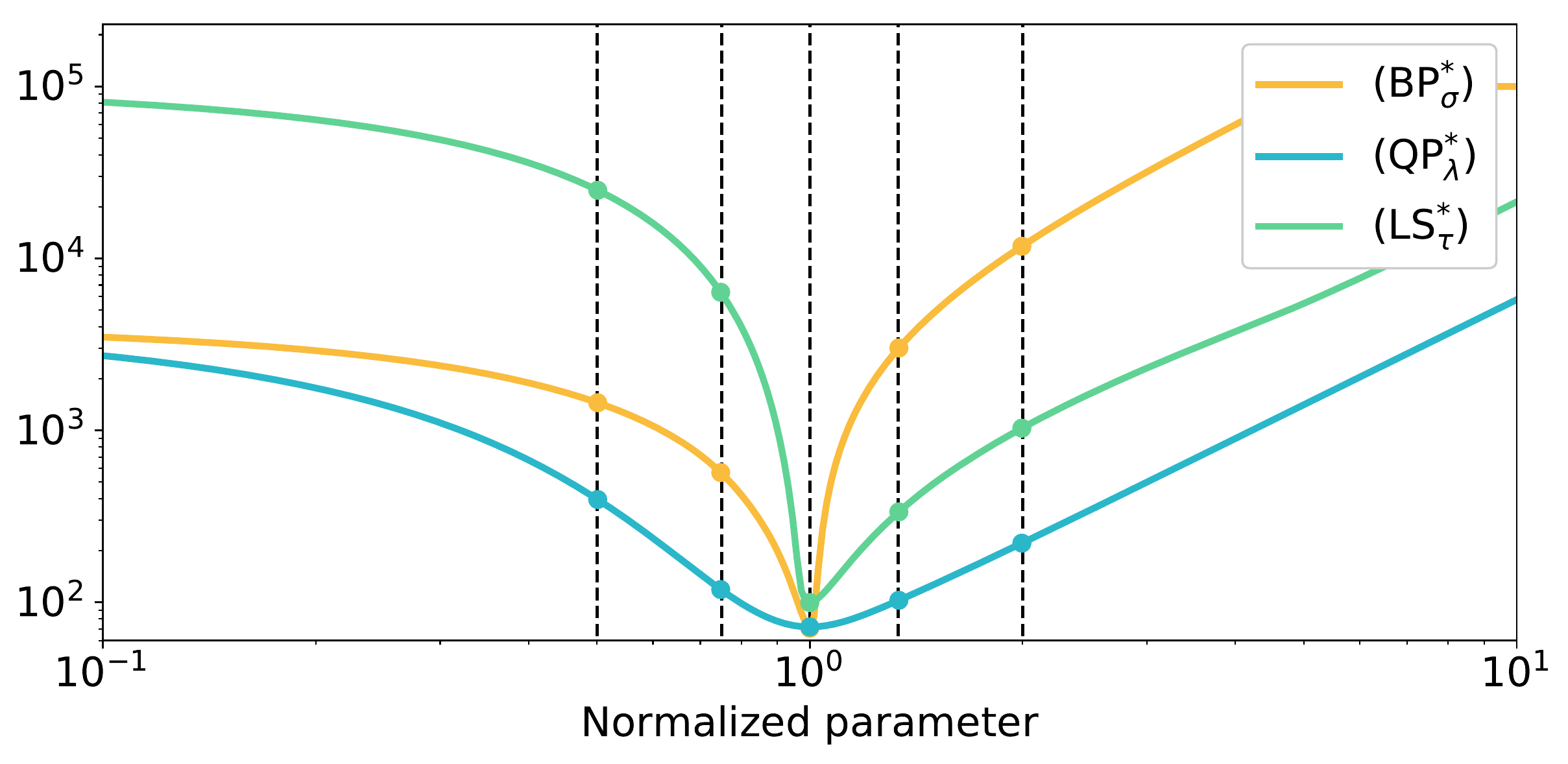}

  \includegraphics[width=.8\textwidth]{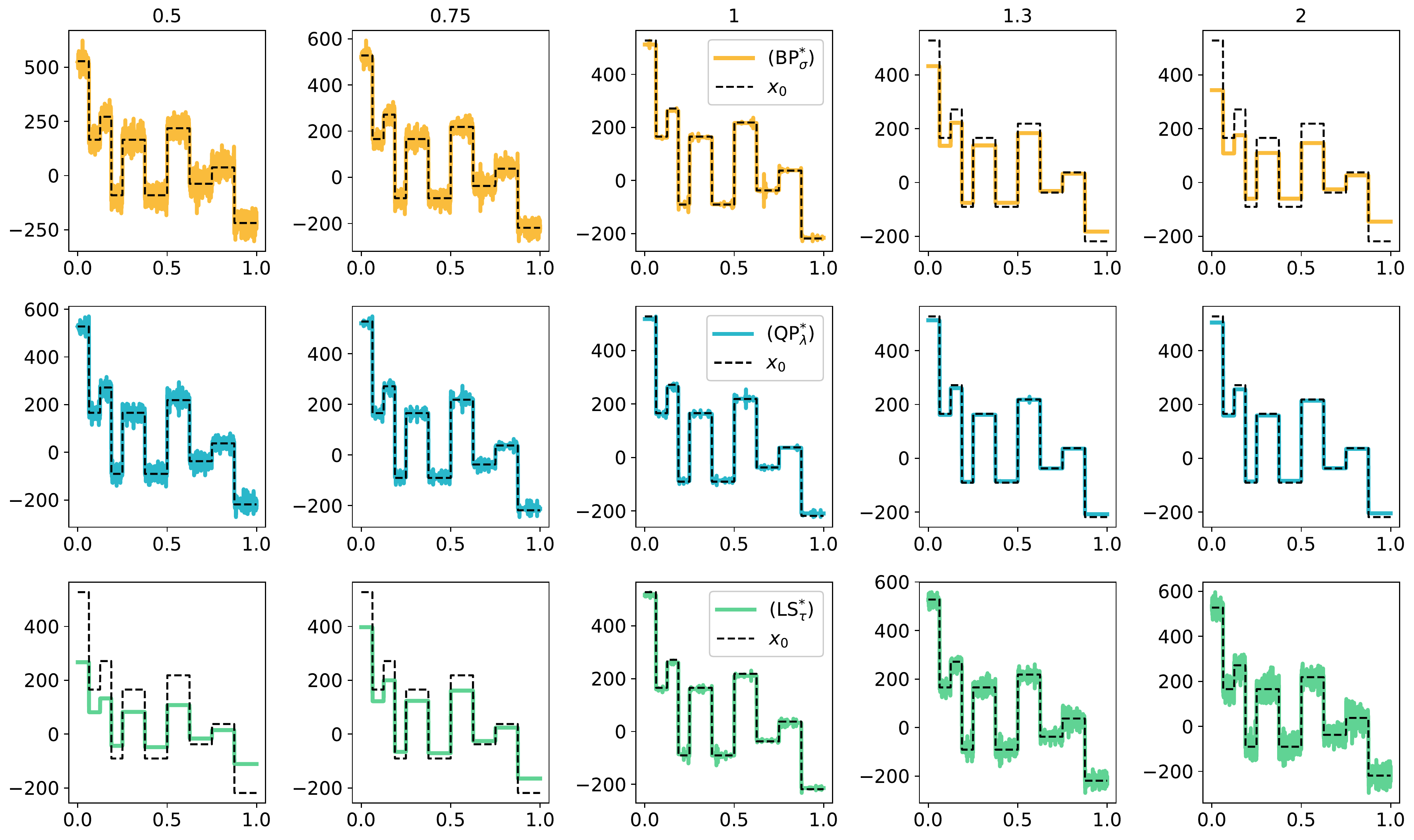}

  \caption{Wavelet space denoising of a 1D signal for different values of the
    normalized parameter when $\eta \approx 41$. \textbf{Top:} The sections of
    the average-loss surface for which estimator recovery will be visualized
    are depicted by the dots which lie nearly on the blacked dotted lines,
    themselves located at $\rho = 0.5, 0.75, 1, 4/3, 2$. \textbf{Bottom:} This
    group of fifteen plots represents a program's solution for a particular
    value of the normalized parameter, arranged in a grid. Each row of the 15
    plot grouping represents a program, as denoted by the legend label; each
    column a value of the normalized parameter, as determined by the heading
    above the top row. }
  \label{fig:realistic-example-1d}
\end{figure}

\subsubsection{Wavelet-space denoising}
\label{sec:realistic-sslp-wdn}

In this section, we demonstrate parameter instability regimes for a realistic
example using proximal denoising of an image signal in a wavelet
domain. Namely, noise is added in the image domain, the data denoised in Haar
wavelet space, and performance of the back-transformed estimator is evaluated
in the image domain. The image was designed to resemble a Shepp-Logan phantom,
but to admit a very sparse expansion in Haar wavelets. This modified phantom,
which we coin the ``Square Shepp-Logan phantom'', was created so as to be
sparse enough to allow for better visualization of {\bppd} parameter
instability. Specifically, if one were to generate the same figures for the
Shepp-Logan phantom, one would see that {\bppd} is less parameter stable than
{\qppd}, but that the behaviour is markedly less pronounced than the behaviour
we visualize in Figure \ref{fig:square-shepp-logan-wavelet-denoising-example} or
Figure \ref{fig:sslp-grid-plots-2}. Indeed this discrepancy results from the
standard Shepp-Logan phantom being less sparse (having more non-zero entries)
in its Haar wavelet transform than our modification. An alternative
demonstration using the standard Shepp-Logan phantom might proceed using a
different transform domain in which its representation is sparser.

A corrupted Square Shepp-Logan phantom was obtained by adding iid noise
$z_{i,j} \iid \mathcal{N}(0,1)$ to the image pixels $I = (I_{i,j})_{i,j}$,
yielding $y$ where $y_{i,j} = I_{i,j} + \eta z_{i,j}$ with
$\eta = 10^{-5}, 0.5$ and where $I_{i,j} \in [0,1]$ is the $(i,j)$th pixel of
the uncorrupted Square Shepp-Logan phantom. The input signal to each recovery
program was the vectorized 2D Haar wavelet transform of $y_{i,j}$:
$w = \mathcal{W}(y_{i,j})_{i,j}$ where $\mathcal{W}$ is the operator connoting
a Haar wavelet transform to (vectorized) Haar wavelet coefficients. Loss was
computed in the image domain, using the nnse of the inverse-transformed
proximal denoising estimator. For example, the loss for {\bppd} is given by
$\eta^{-2}\|\mathcal{W}^{-1}(\tilde x(\sigma)) - I\|_{2}^{2}$. Average loss
\eqref{eq:avg-loss} was thus computed by averaging the loss over $k = 25$
realizations of the noise $z$.

The associated parameters of the problem are
$(s, N, k, n) = (5188, 409618, 25, 501)$, implying a relative sparsity of
$1.27\%$. To create effective visualizations of the parameter instability
behaviour, the noisy images seen in the top row of
Figure \ref{fig:square-shepp-logan-wavelet-denoising-example} are scaled to the
interval $[0,1]$. Subsequent visualizations do not perform this rescaling so that a perceptual evaluation of the recovery is better facilitated.

The plots in the bottom row of
Figure \ref{fig:square-shepp-logan-wavelet-denoising-example} depict the average
loss as a function of the normalized parameter $\rho$ of each program. For each
of the $k$ realizations, the loss was computed on a logarithmically spaced grid
of $n = 501$ points about the optimal parameter. As in
section \ref{sec:image-space-deno}, a smooth approximating curve to the
non-uniformly spaced point cloud of loss realizations was computed using RBF
approximation. The RBF approximation used multiquadric kernels with parameters
$(\varepsilon_{\text{rbf}}, \mu_{\text{rbf}}, n_{\text{rbf}}) = (10^{-3},
10^{-2}, 301)$ \cite{scipy}. The RBF parameters for the approximation were
selected so as to generate a smooth line that best represents the path about
which the individual (noisy) data points concentrate, especially so as to
resolve the behaviour of the loss about $\rho = 1$.

About the optimal average loss, an approximate $10^{6}$ fold difference in nnse
results from a less than $2\%$ perturbation of $\tau$ in the low-noise and very
sparse regime ($\eta = 10^{-5}, s/N \approx 1.27\%$). In this regime, we
observe that {\bppd} is less stable than {\qppd}, especially for values of the
normalized parameter greater than $1$, as suggested by our theory. In the very
sparse regime with large noise ($\eta = 0.5$), {\bppd} is markedly more
parameter unstable than {\lspd} or {\qppd}, especially for values of the
normalized parameter exceeding $1$. Moreover, we observe that the minimal
average loss for {\bppd} is greater than that for {\lspd} or {\qppd}. This
numerical behaviour is consistent with our theoretical results.

In Figure \ref{fig:sslp-grid-plots-1} and Figure \ref{fig:sslp-grid-plots-2} we
depict estimator performance by visualizing the solution to each program at
specific values of the normalized parameter. The description of each figure is
identical, but the noise levels $\eta$ differ between them. Specifically, for
each program we show the recovered image and its pixel-wise nnse for values of
the normalized parameter $\rho = 0.5, 0.75, 1, 4/3, 2$. The plot in the top row
of the figure depicts a loss curve for each program (\ie a curve generated from
one realization of the noise $z$), along with reference lines for the
corresponding values of the normalized parameter whose recovered image are
visualized. The middle row contains a grid of 15 images; each column
corresponds to a value of the normalized parameter as denoted by the title
heading, while each row corresponds to a proximal denoising program as denoted
by the labels along the left-most $y$-axis. The bottom grouping of 15 images
depicts the pixel-wise nnse, arranged identically to the middle row. Because
the average loss curves were computed on a grid of $n$ logarithmically spaced
points centered about the optimal parameter value, we do not visualize the
recovered image for the exact values of $\rho$ given above, but for those
values represented by the coloured points seen in the plot of the top
row. These points are sufficiently close to the quoted values of $\rho$ so as
to visualize the program behaviour all the same.

The numerics of Figure \ref{fig:sslp-grid-plots-1} occur in the low-noise regime
$(\eta = 10^{-5})$, and so, as expected, demonstrate parameter instability of
{\lspd}. We note that pixel-wise nnse for {\bppd} is approximately $20$ times
worse than {\qppd} when $\rho \approx 2$. Moreover, the pathologies (in the
sense of pixel-wise nnse) of these latter two programs appear similar. We also
observe that the pixel-wise nnse varies more greatly for {\bppd} than for
{\qppd} as $\rho$ varies from $0.75$ to $4/3$. This is consistent with our
theory for the behaviour of {\bppd} in the very sparse regime. The numerics of
Figure \ref{fig:sslp-grid-plots-2} occur in the high-noise regime $(\eta =
0.5)$. Failure of {\bppd} in the very sparse regime is seen from examining the
solution itself. For example, when $\rho < 1$, pixel values of the solution to
{\bppd} may reach more than $2$ or even be negative. This pathology manifests
as large-magnitude pixelation in the corresponding plots of pixel-wise
nnse. Catastrophic failure of {\bppd} is observed for $\rho > 1$, in which the
program fails to recover any semblance of the original image. Specifically,
large $\sigma$ shrinks the wavelet coefficients to near the origin, enforcing
few non-zero components that are small in magnitude. This yields the
rectangular pattern observed in the solutions for {\bppd} (top-right of the
middle row). In contrast, moderate deformation of the image is observed for
$\rho \neq 1$ for both {\qppd} and {\lspd}.

\begin{figure}[t]
  \centering
  \includegraphics[width=.8\textwidth]{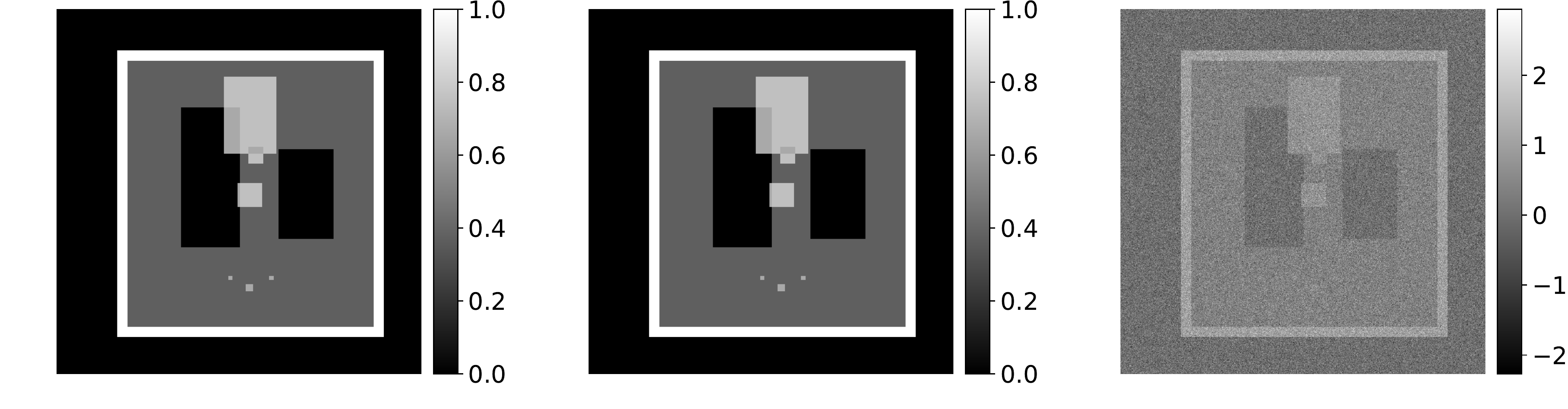}

  \includegraphics[width=.45\textwidth]{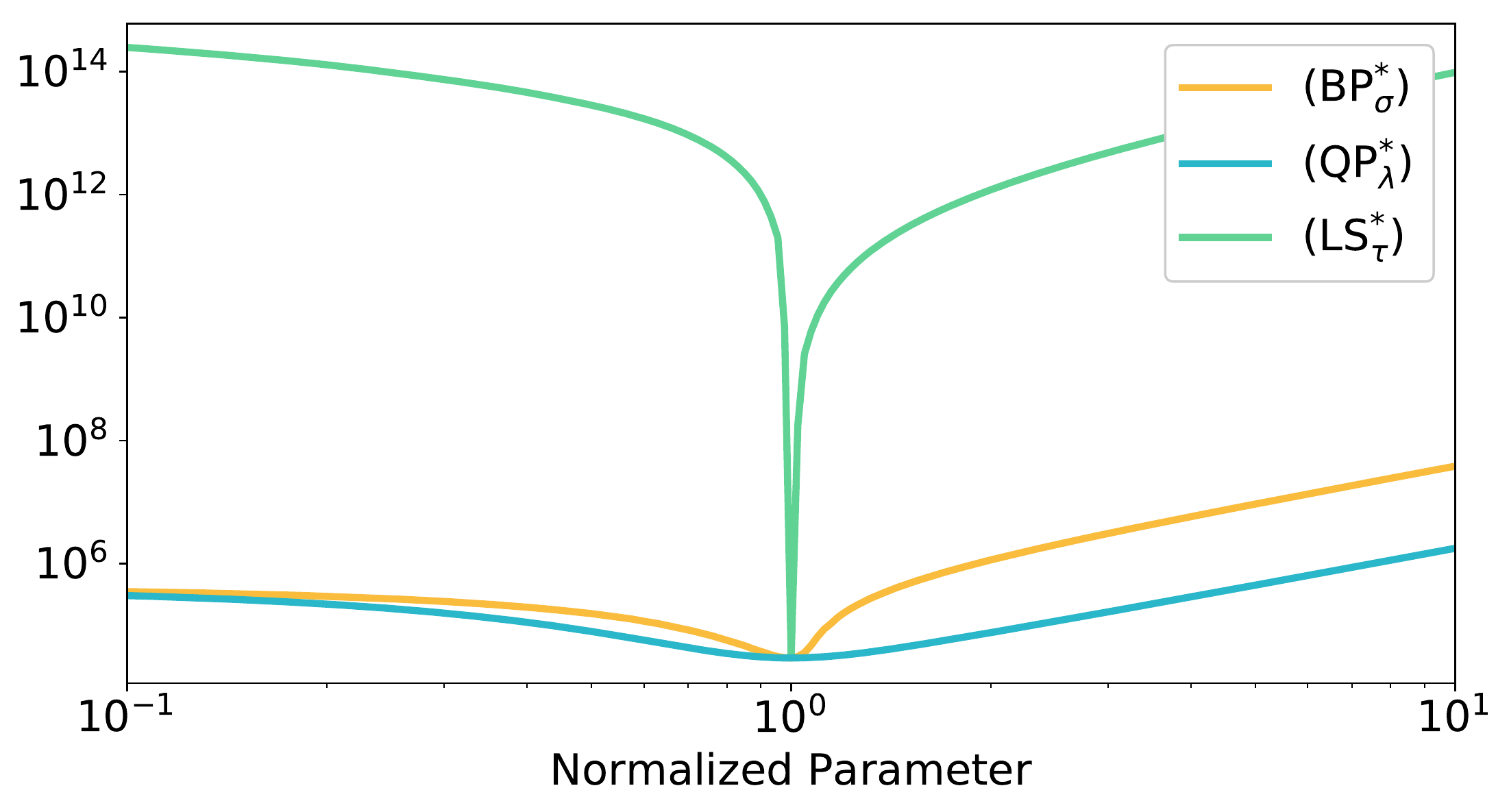}
\includegraphics[width=.45\textwidth]{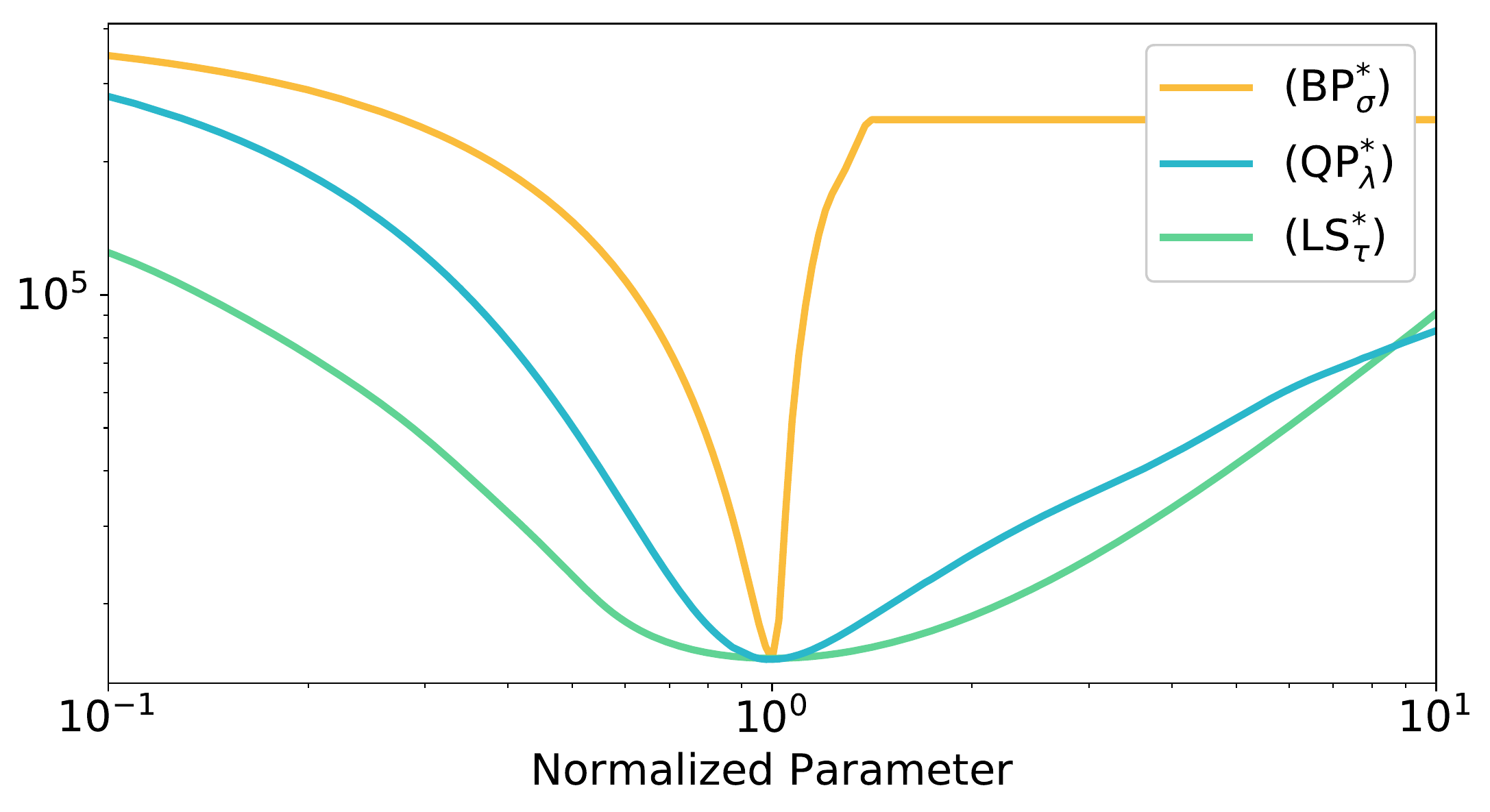}
  
\caption{\textbf{Top (left-to-right):} The underlying signal is the
  $640 \times 640$ Square Shepp-Logan phantom image; the middle image is
  corrupted by iid normally distributed noise ($\eta = 10^{-5}$); the
  right-most image is corrupted by iid normally distributed noise
  ($\eta = 0.5$). The pixel values of the original image lie in $[0,1]$; those
  of the noisy images are scaled to $[0,1]$. \textbf{Bottom:} Average loss is
  plotted with respect to the normalized parameter for {\lspd}, {\qppd} and
  {\bppd} respectively when $\eta = 10^{-5}$ (left) and $\eta = 0.5$
  (right). The associated parameters are
  $(s, N, k, n) = (5188, 409618, 25, 501)$, implying relative sparsity of
  $1.27\%$. Plotted lines are smoothed approximations of loss realization data
  using multiquadric RBFs. }
  \label{fig:square-shepp-logan-wavelet-denoising-example}
\end{figure}

\begin{figure}[h]
  \centering
  
\includegraphics[width=.8\textwidth]{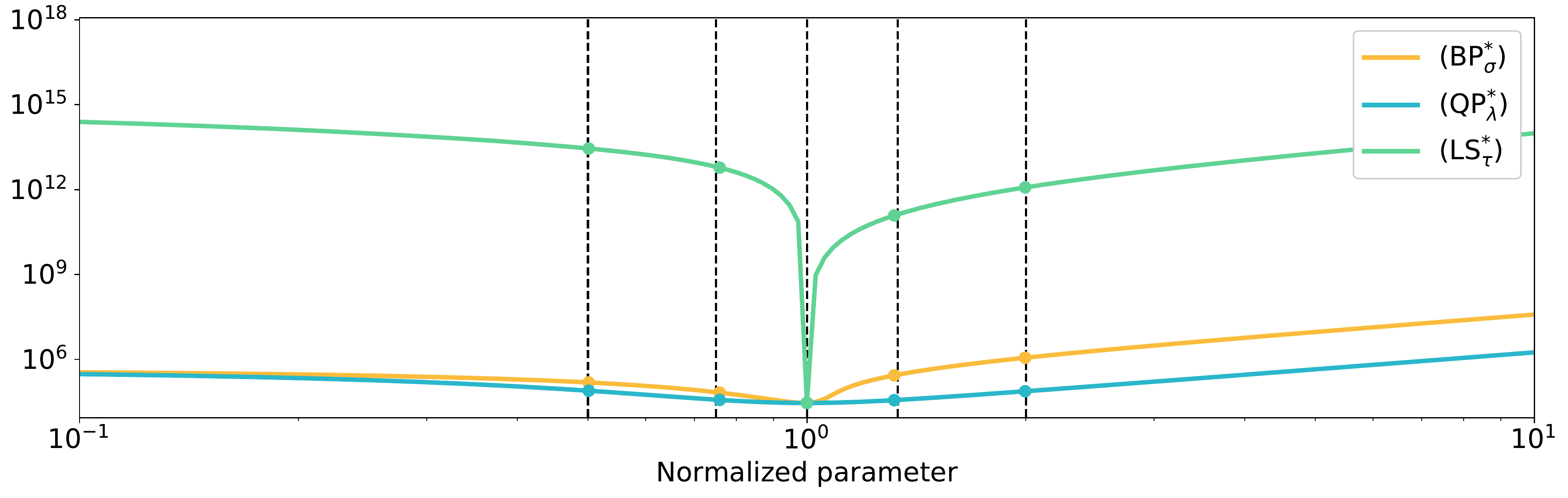}\qquad
  
   \includegraphics[width=.8\textwidth]{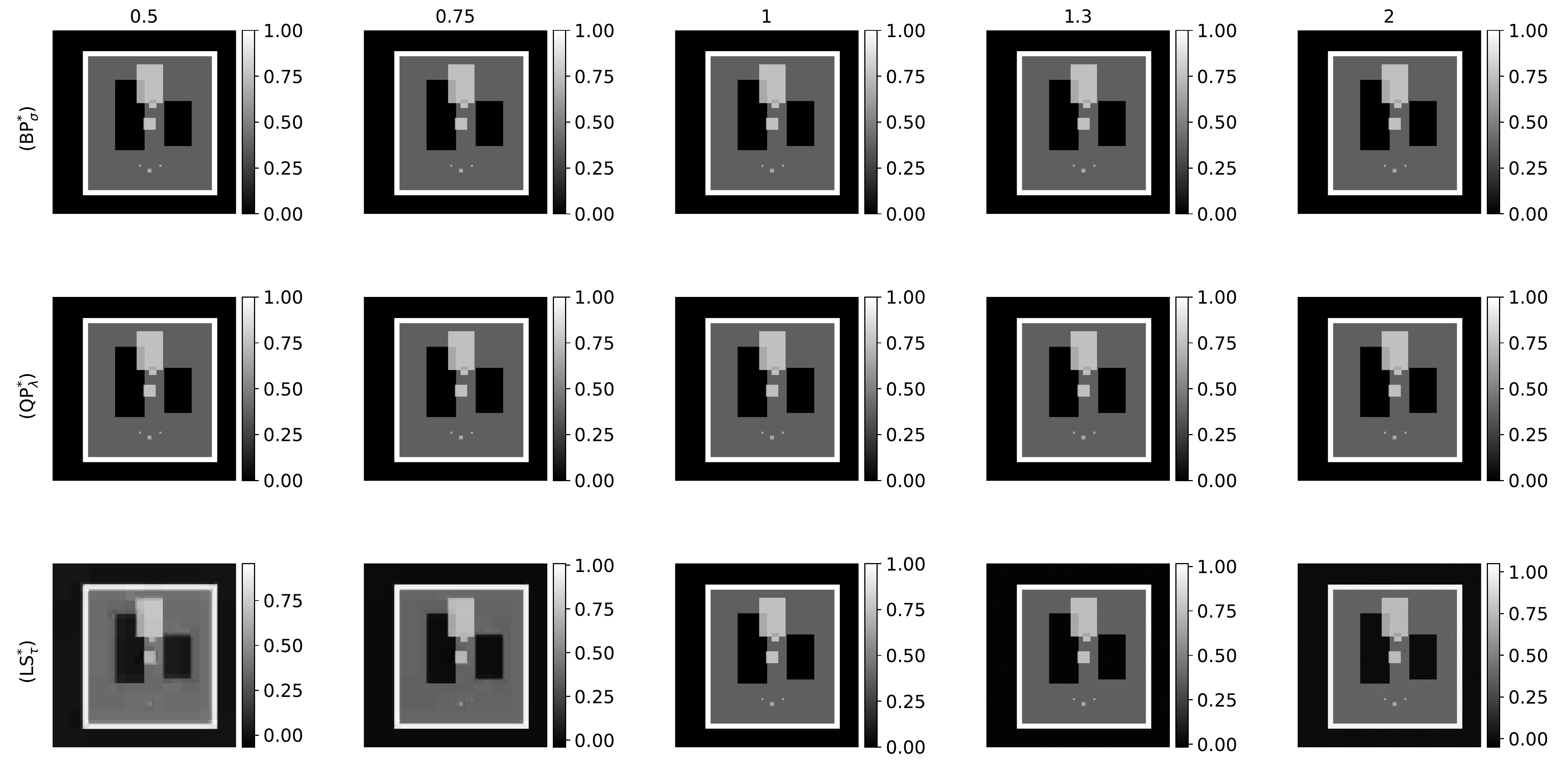}\qquad
  
   \includegraphics[width=.8\textwidth]{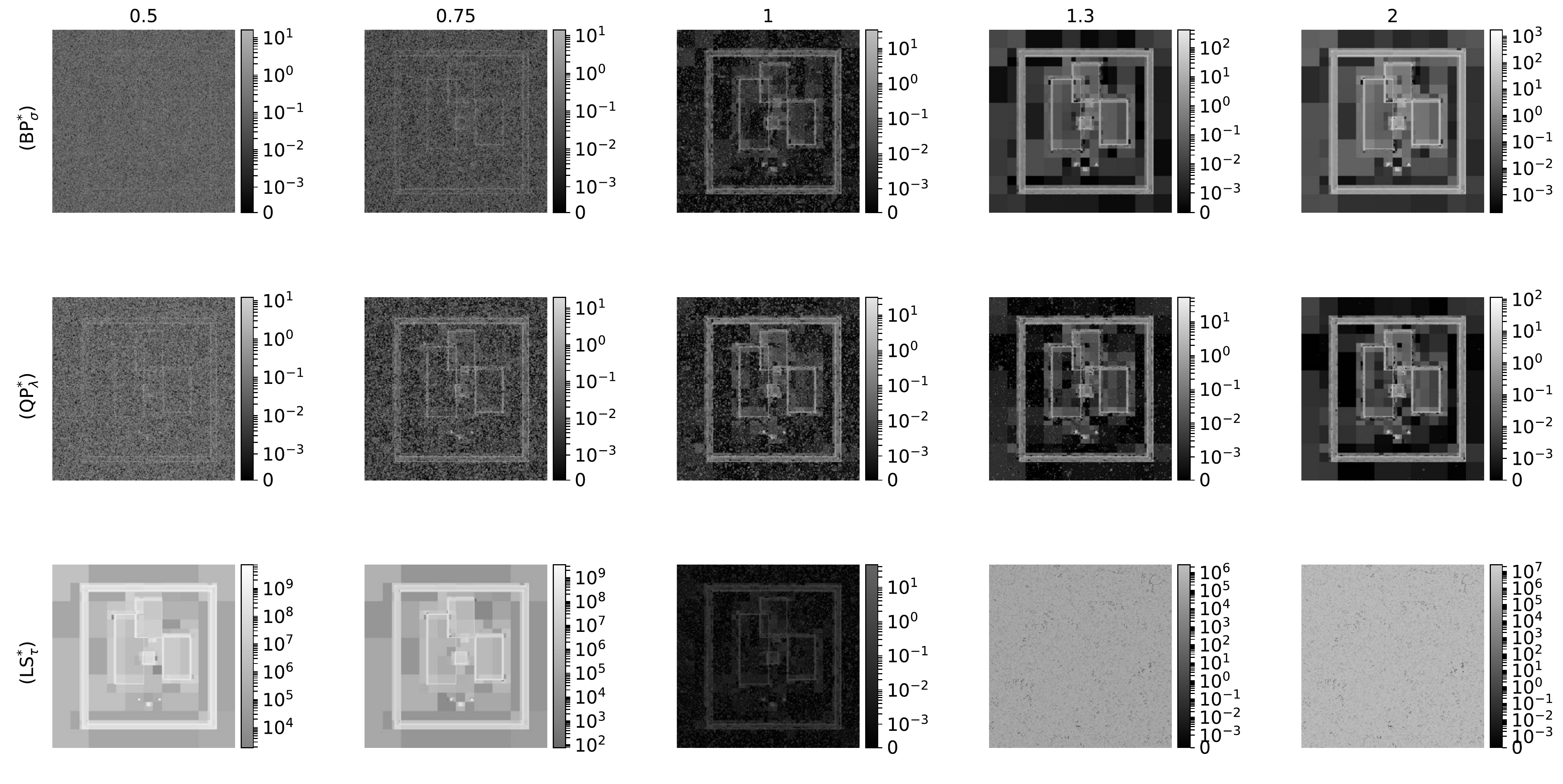}
  
   \caption{Wavelet space denoising of the square Shepp-Logan phantom for
     different values of the normalized parameter when $\eta =
     10^{-5}$. \textbf{Top:} The sections of the average-loss surface for which
     estimator recovery will be visualized are depicted by the dots which lie
     nearly on the blacked dotted lines, themselves located at
     $\rho = 0.5, 0.75, 1, 4/3, 2$. \textbf{Middle:} This group of fifteen
     plots represents a program's solution for a particular value of the
     normalized parameter, arranged in a grid. Image pixel values are not
     scaled to $[0,1]$; their range is given by the associated colour
     bar. \textbf{Bottom:} This group of fifteen plots depicts pixel-wise nnse
     for each (program, normalized parameter) pairing. In both the middle and
     bottom groups, the program is denoted along the left-hand side, while the
     normalized parameter value is denoted along the top row of each group.}
  \label{fig:sslp-grid-plots-1}
\end{figure}

\begin{figure}[h]
  \centering
 \includegraphics[width=.8\textwidth]{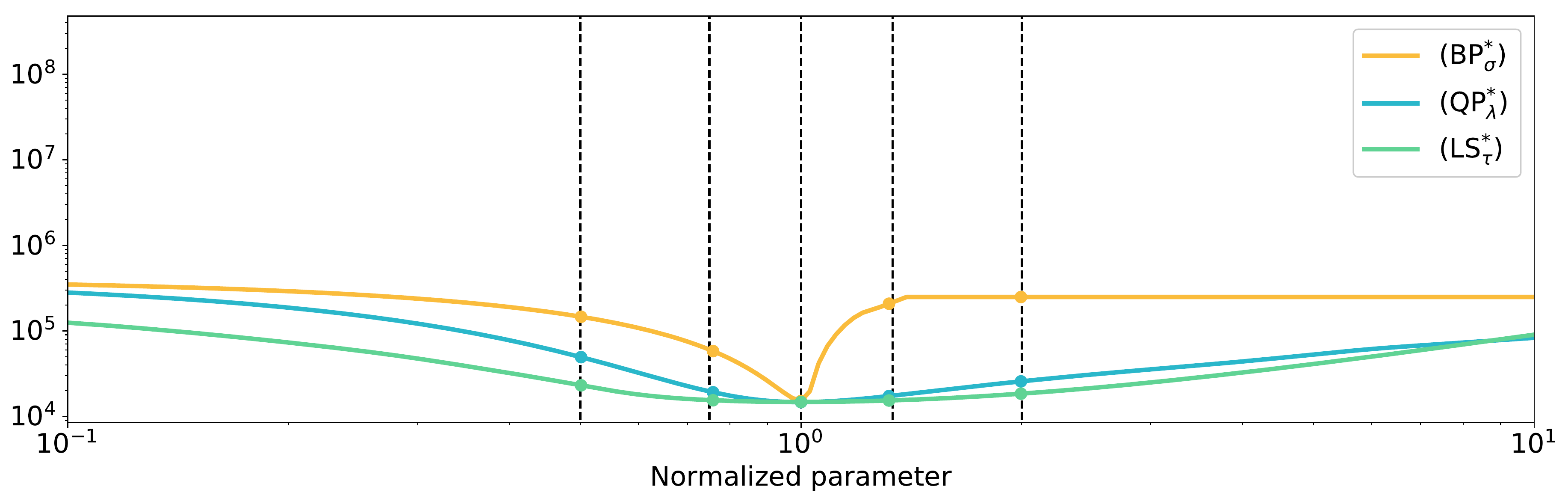}

 \includegraphics[width=.8\textwidth]{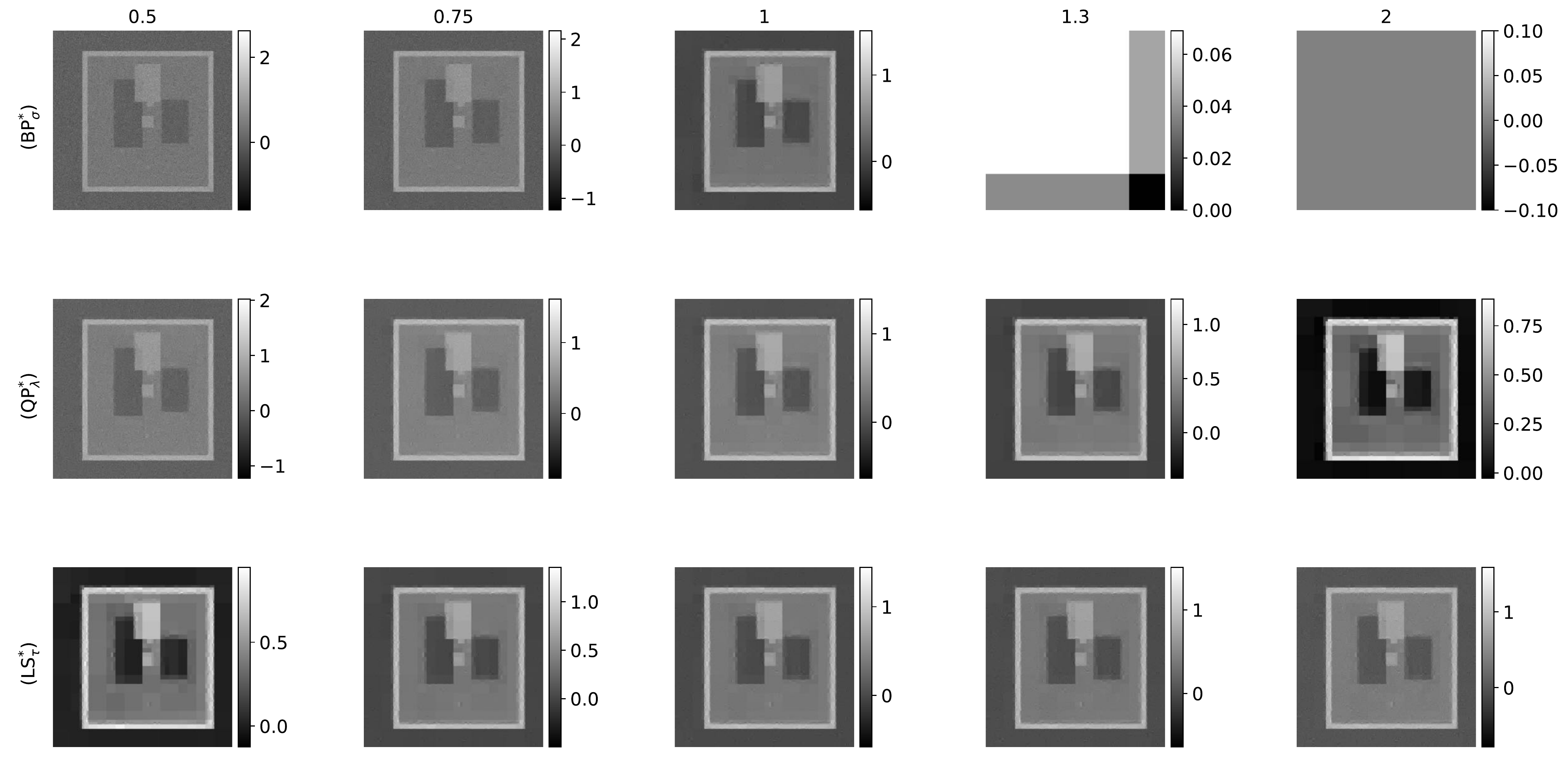}
   \includegraphics[width=.8\textwidth]{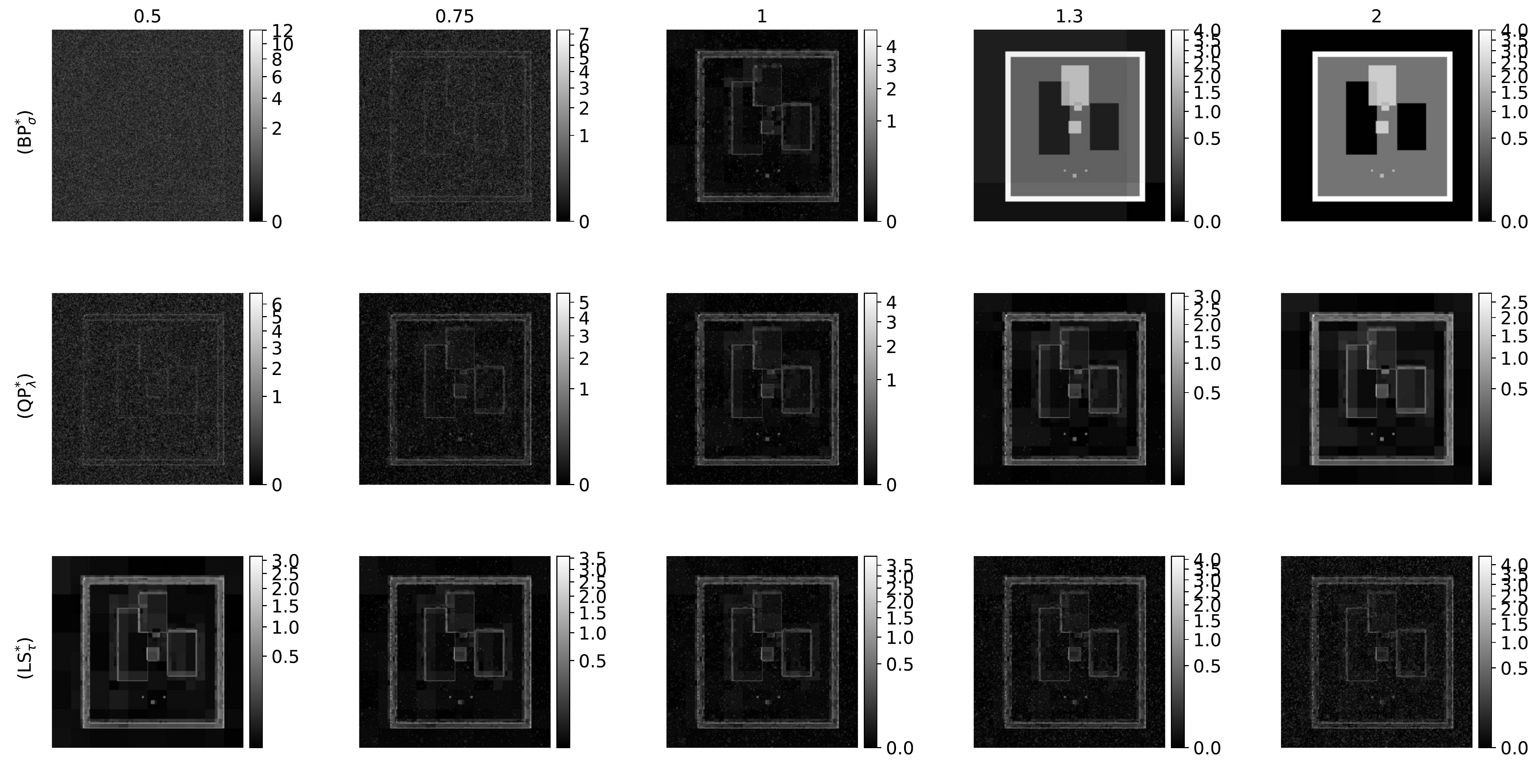}

   \caption{Wavelet space denoising of the square Shepp-Logan phantom for
     different values of the normalized parameter when $\eta =
     0.5$. \textbf{Top:} The sections of the average-loss surface for which
     estimator recovery will be visualized are depicted by the dots which lie
     nearly on the blacked dotted lines, themselves located at
     $\rho = 0.5, 0.75, 1, 4/3, 2$. \textbf{Middle:} This group of fifteen
     plots represents a program's solution for a particular value of the
     normalized parameter, arranged in a grid. Image pixel values are not
     scaled to $[0,1]$; their range is given by the associated colour
     bar. \textbf{Bottom:} This group of fifteen plots depicts pixel-wise nnse
     for each (program, normalized parameter) pairing. In both the middle and
     bottom groups, the program is denoted along the left-hand side, while the
     normalized parameter value is denoted along the top row of each group.}
  \label{fig:sslp-grid-plots-2}
\end{figure}

\subsubsection{\textsc{Lasso} Example}
\label{sec:lasso-example}

This section includes a realistic example comparing parameter instability of
\textsc{Lasso} programs in the very sparse regime, both in the low noise regime
and when the noise is relatively large. Specifically, we assume the model
\begin{align*}
  y = Ax_{0} + z, \qquad x_{0} = \mathcal{W}(I)
\end{align*}
where $\mathcal{W}(I)$ connotes the 2D Haar wavelet transform of $I$, the
$80\times 80$ square Shepp-Logan phantom. This image size was reduced from that
of section \ref{sec:realistic-sslp-wdn}, because using the full image for the
examples in this section would have been computationally prohibitive. The
measurement matrix $A \in \reals^{m\times N}$ has entries
$A_{ij} \iid Z / \sqrt m$ where $Z \sim \mathcal{N}(0,1)$. The parameters for
the problem are $(N, s, m) = (6418, 416, 3110)$, implying a sparsity ratio of
$6.48\%$ in the Haar wavelet domain, and a measurement matrix aspect ratio of
$48.46\%$ with $m \gtrsim s \log (N/s)$. The wavelet coefficients $x_{0}$ are
recovered according to {\ls}, {\bp} and {\qp} where $K = B_{1}^{N}$.

\begin{figure}[t]
  \centering
  \includegraphics[width=.47\textwidth]{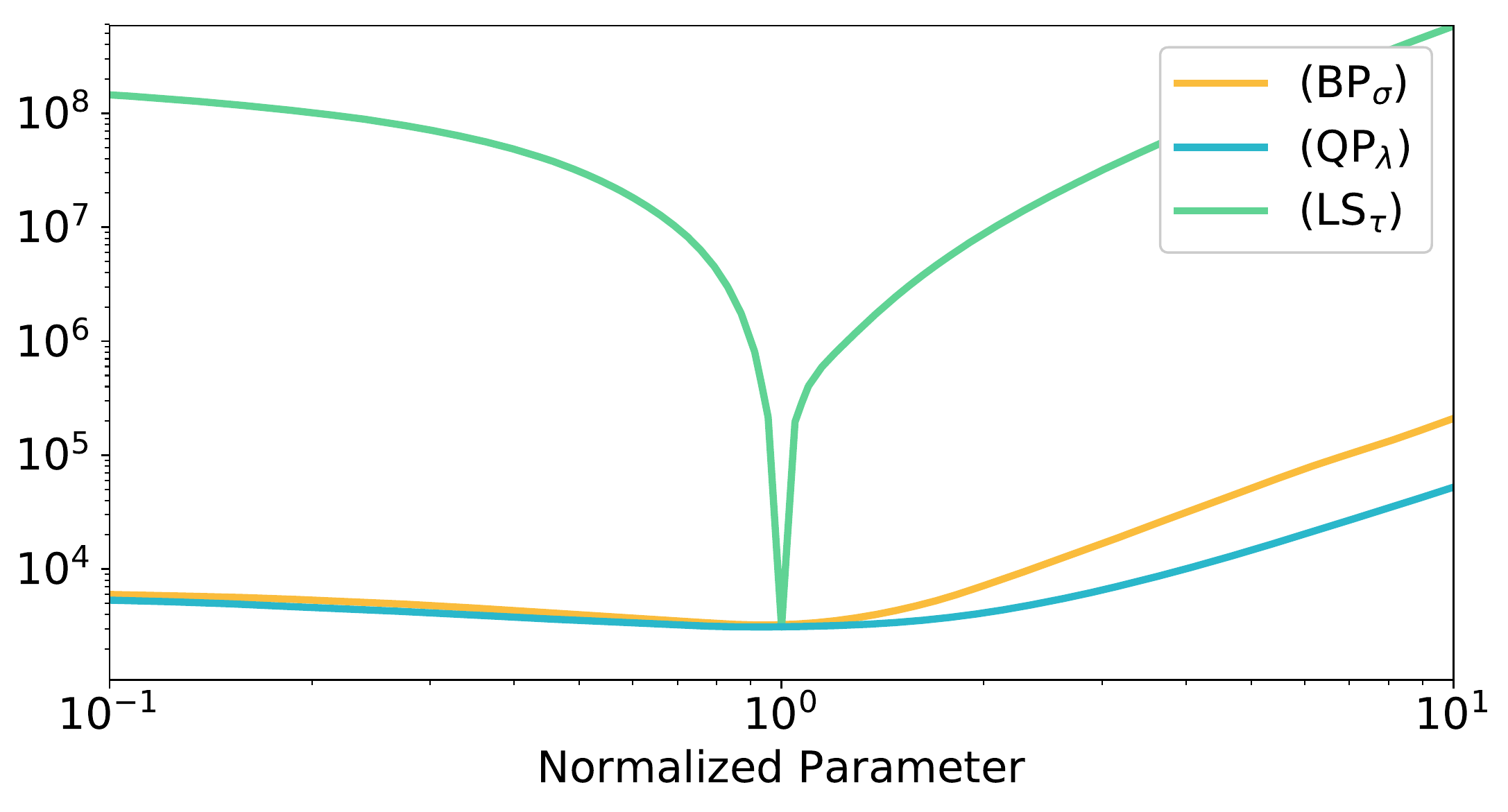}\qquad
\includegraphics[width=.47\textwidth]{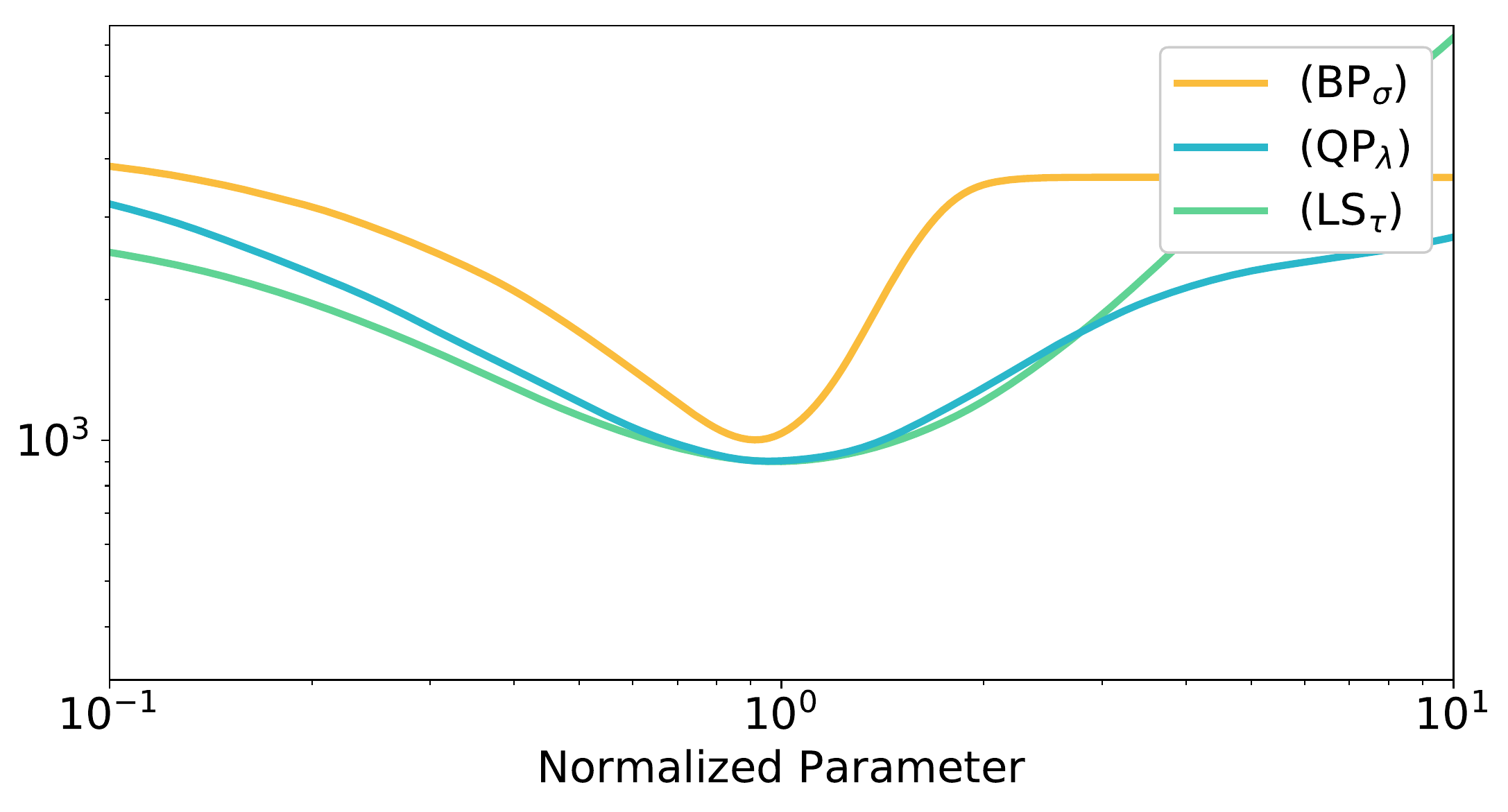}

\includegraphics[width=.47\textwidth]{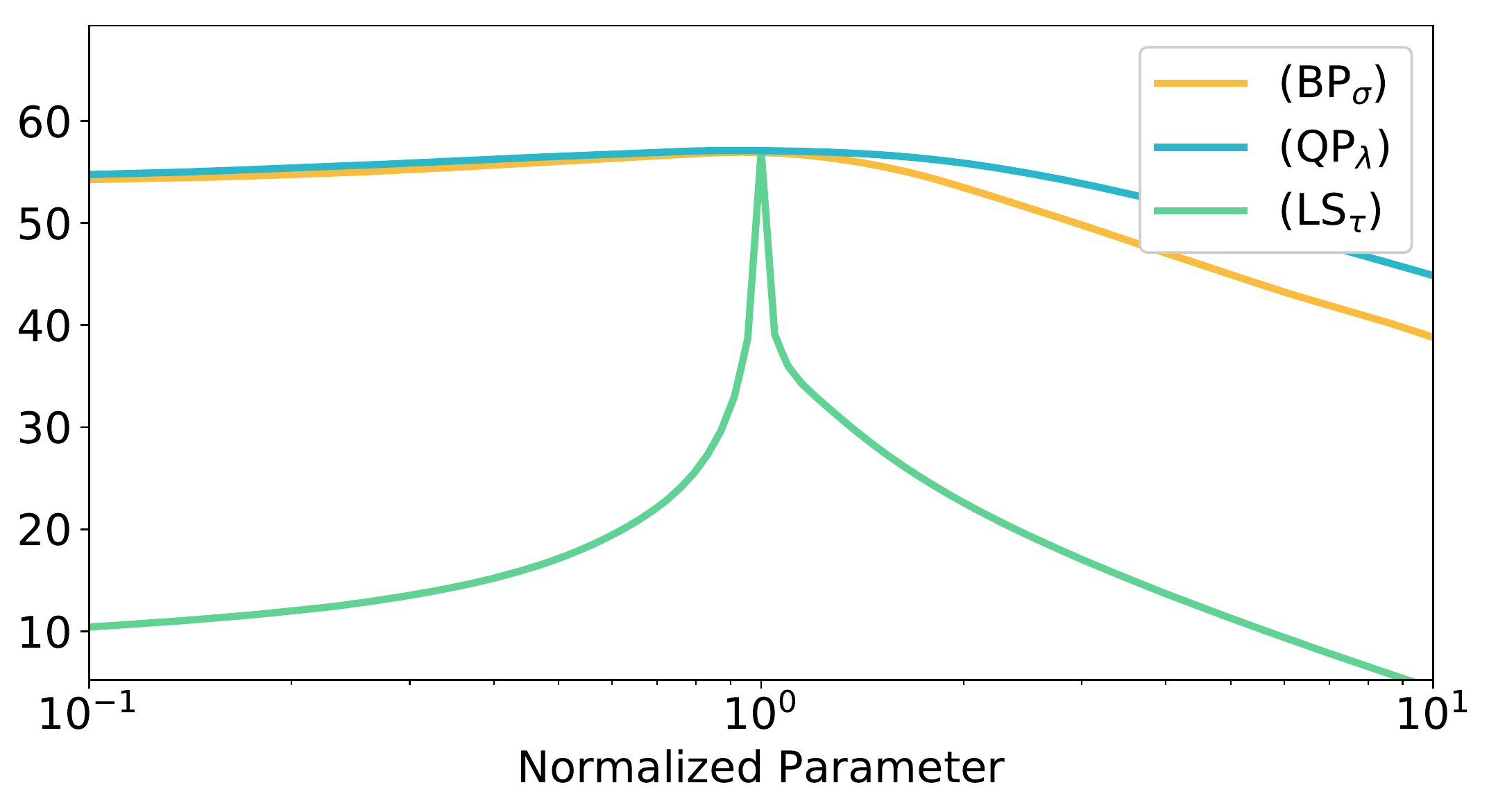}\qquad
\includegraphics[width=.47\textwidth]{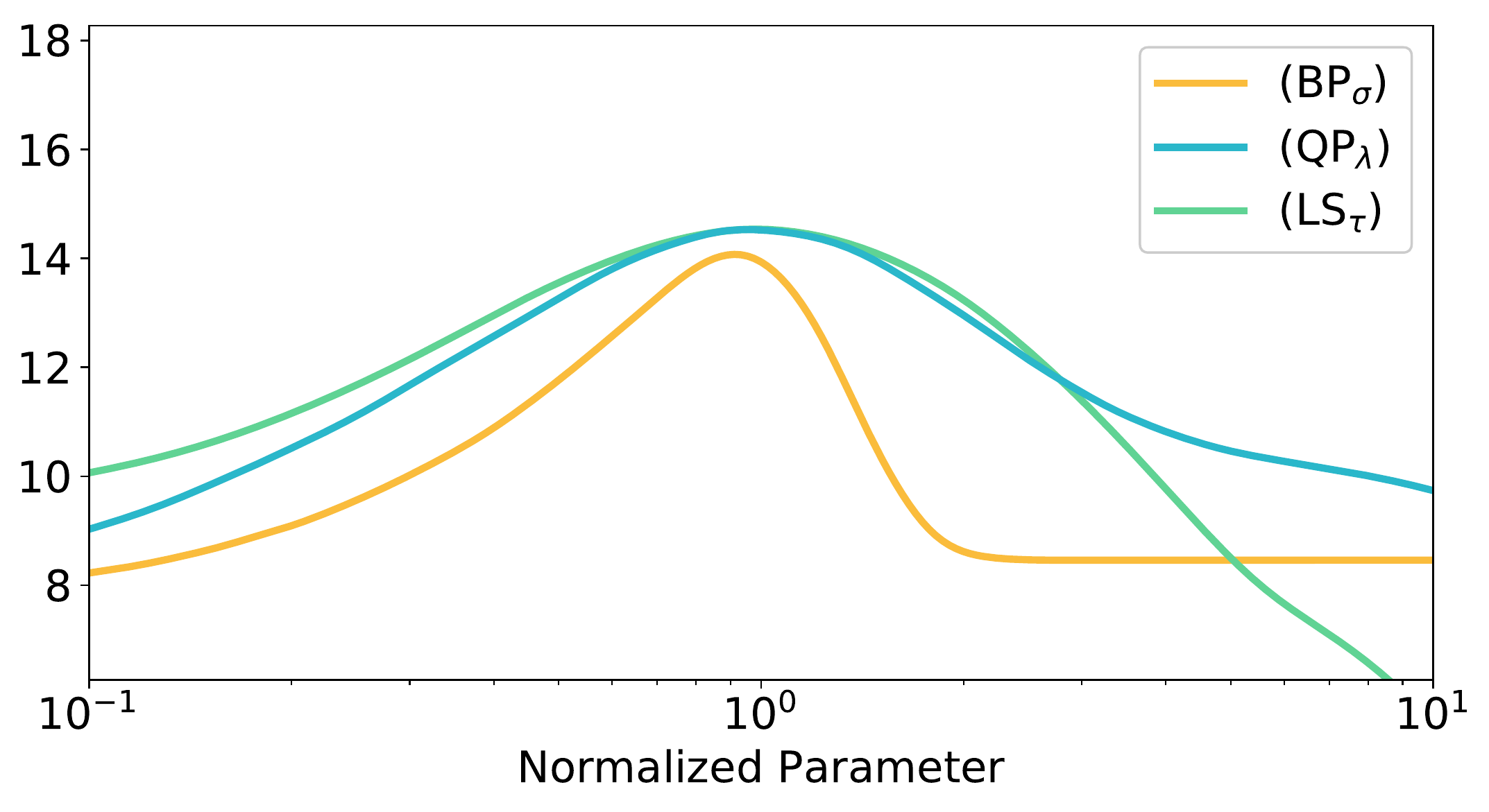}

\caption{Average loss (top) and average psnr (bottom) \emph{vs.} normalized
  parameter for {\ls}, {\qp} and {\bp} respectively when
  $\eta = 2\cdot 10^{-3}$ (left) and $\eta = 0.5$ (right). Associated
  parameters: $(s, N, m, k, n) = (416, 6418, 3110, 25, 301)$; relative sparsity
  $6.48\%$, aspect ratio $48.46\%$. Plotted lines approximate average of
  realization data using multiquadric RBFs.}
  \label{fig:average-loss-lasso-sslp}
\end{figure}

Given two signals $x^{*}, x_{0}$ define the peak signal-to-noise ratio (psnr) by
\begin{align}
  \label{eq:psnr}
  \mathrm{psnr}(x^{*}, x) %
  := 10 \log_{10}\Big( \frac{\max_{i \in [N]} x_{i}^{2}}{%
  \mathrm{mse}(x^{*}, x_{0})}\Big), %
  \qquad %
  \mathrm{mse}(x^{*}, x_{0}) %
  := \frac{1}{N}\sum_{i = 1}^{N} (x^{*}_{i} - x_{i})^{2}.
\end{align}
As with defining loss $L(\rho; x_{0}, N, \eta \hat z_{ij})$, by abuse of
notation we define psnr as a function of the normalized parameter $\rho$,
$\mathrm{psnr}(\rho) := \mathrm{psnr}(\rho; x_{0}, N, \eta \hat z_{ij}) :=
\mathrm{psnr}(x^{*}(\varrho), x_{0})$.

In Figure \ref{fig:average-loss-lasso-sslp}, we compute average loss and average
psnr as a function of the normalized parameter, where average loss is measured
using nnse in the image domain, as in section \ref{sec:realistic-sslp-wdn} and
average psnr is defined as in \eqref{eq:psnr}. The data was simulated for
$k = 25$ realizations of noise for both $\eta = 2\cdot 10^{-3}$ (left) and
$\eta = 0.5$ (right). Each curve visualized is a radial basis function
approximation to the true average curve, obtained from the non-uniformly spaced
point cloud of realization data. Specifically, the loss and psnr for each
realization was computed on a logarithmically spaced grid of $n = 301$ points
about the optimal normalized parameter. For both average psnr and average loss,
the resultant point cloud of $7\,525$ points for each program was used as input
for a multiquadric RBF approximation with parameters
$(\varepsilon_{\mathrm{rbf}}, \mu_{\mathrm{rbf}}, n_{\mathrm{rbf}}) = (10^{-1},
10, 301)$ (except for {\ls} when $\eta = 2\cdot 10^{-3}$, for which the
parameters were
$(\varepsilon_{\mathrm{rbf}}, \mu_{\mathrm{rbf}}, n_{\mathrm{rbf}}) = (10^{-3},
20^{-1}, 301)$, selected so as to properly resolve the cusp about $\rho = 1$)
\cite{scipy}.

About the optimal choice of normalized parameter, $\rho = 1$, an approximate
$1.51\cdot 10^{5}$ fold difference in average loss results from a $2.73\%$
average perturbation of the normalized parameter for {\ls} in the low-noise
regime ($\eta = 2\cdot 10^{-3}$). In contrast, the error difference is no more
than $10^{3}$ for the other two programs. In particular, {\ls} undergoes an
approximate $30\,\mathrm{dB}$ drop in psnr for this small variation. In the
very sparse regime with $\eta = 0.5$ (right), one observes that {\bp} is the
least stable when compared with the other two programs. If $\sigma$ is $10\%$ larger than the optimal choice, the psnr is approximately halved. Moreover,
its best average loss is observed to be strictly greater than that for either
of the other two programs. This observation mirrors the numerics for {\bppd} in
section \ref{sec:realistic-sslp-wdn} and is consistent with the theoretical results
of section \ref{sec:param-inst-bppd}.

Finally, we observe that the numerics for {\qp} exhibit parameter stability,
though the data regime is low-noise and very sparse
(Figure \ref{fig:average-loss-lasso-sslp}, left). We claim this behaviour is not
contrary to \ref{ssec:left-sided-parm-inst}, and use the following
intuition from {\qppd} to elucidate. When $\lambda < \bar \lambda$,
Theorem \ref{thm:qppd-instability} demonstrates parameter instability for {\qppd}
behaving as $R^{\sharp}(\lambda; s, N) \gtrsim N^{\varepsilon}$ where
$\lambda = (1- \varepsilon) \bar \lambda$. This term dominates
$R^{\sharp}(\lambda; s, N)$ only for relatively high dimensional problems
(\emph{cf.}  section \ref{sec:qppd-numerics}), roughly requiring that
$N^{\varepsilon} \gtrsim s\log(N/s)$. By this heuristic, in the present
example, the instability does not dominate for $\varepsilon \leq C\cdot 0.80$
where $C > 1$ is a constant. In particular, the theoretical results of
\ref{ssec:left-sided-parm-inst} show that in high dimensions, if the signal
is very sparse and the noise level is small, then {\qppd} is an appropriate
choice only if $\lambda > \bar \lambda$, while our numerics support that
$\lambda < \bar \lambda$ remains a safe choice for {\qppd} (\emph{cf.}
section \ref{sec:realistic-sslp-wdn}) and {\qp} in relatively lower dimensional
problems. This observation is particularly advantageous given that parameter
instability may still be expected of both {\ls} and {\bp}.

In Figure \ref{fig:lasso-sslp-grid-plots-1} and
Figure \ref{fig:lasso-sslp-grid-plots-2}, a grid of plots similar to those of
section \ref{sec:realistic-sslp-wdn} were generated to visualize the solution to
each program as a function of the normalized parameter
$\rho \in \{0.5, 0.75, 1, 4/3, 2\}$. As before, the topmost image in each
figure is a reference plot to depict the locations on the curve to which the
displayed images correspond; however, these plots now depict psnr as a function
of the normalized parameter. The images displayed below the reference plot do
not correspond exactly to the quoted normalized parameters, but to a closest
approximation obtained from a logarithmically spaced grid of $n = 301$ points
centered about the optimal parameter. These true normalized parameter values
are visualized as large coloured dots on the reference plot; that they
approximate well the quoted normalized parameter values for $\rho$ is verified
by their proximity to the black dotted lines in the reference plot. Showing the
estimator corresponding to a normalized parameter has the twofold purpose of
visualizing the pathology of each program as its parameter varies, and
demonstrating when a program is relatively unstable in a given regime.

The images in Figure \ref{fig:lasso-sslp-grid-plots-1} portray the setting of
low-noise regime, with $\eta = 2\cdot 10^{-3}$. Indeed, the reference plot
displays a cusp for the {\ls} loss that was characteristic of {\lspd} in the
low-noise regime. Moreover, one observes similar pathologies in both the
recovered image for {\ls} as well as the point-wise nnse. The image recovered
using {\ls} is blurry for $\rho < 1$, which is indicative of incompletely
recovered wavelet coefficients. For $\rho > 1$, the noise was not suppressed on
the off-support of the wavelet coefficients yielding a noisy pixelated image in
both the recovered image and the error image. This behaviour is observed to a
significantly lesser degree for the corresponding {\bp} and {\qp}
images. 

The images in Figure \ref{fig:lasso-sslp-grid-plots-2} portray the very sparse
regime where $\eta = 0.5$. As is consistent with the asymptotics in
section \ref{sec:param-inst-bppd} for {\bppd}, it is difficult to visualize {\bp}
parameter instability for relatively low dimensional problems. We suspect that
this instability would have been markedly more apparent were it possible to run
these simulations for the full $640 \times 640$ square Shepp-Logan phantom
image. Nevertheless, one observes that {\bp} is the least stable of the three
programs for $\rho > 1$. In particular, the visualized estimator and point-wise
nnse both depict catastrophic failure of {\bp} for $\rho = 2$.

Lastly, it is readily observed that the curves computed from single
realizations resemble very closely those computed to approximate the average of
several realizations. Notably, the nnse curves in
Figure \ref{fig:lasso-sslp-grid-plots-1} and Figure \ref{fig:lasso-sslp-grid-plots-2}
strongly resemble the corresponding average nnse curves in
Figure \ref{fig:average-loss-lasso-sslp}.


\begin{figure}[h]
  \centering
  \includegraphics[width=.8\textwidth]{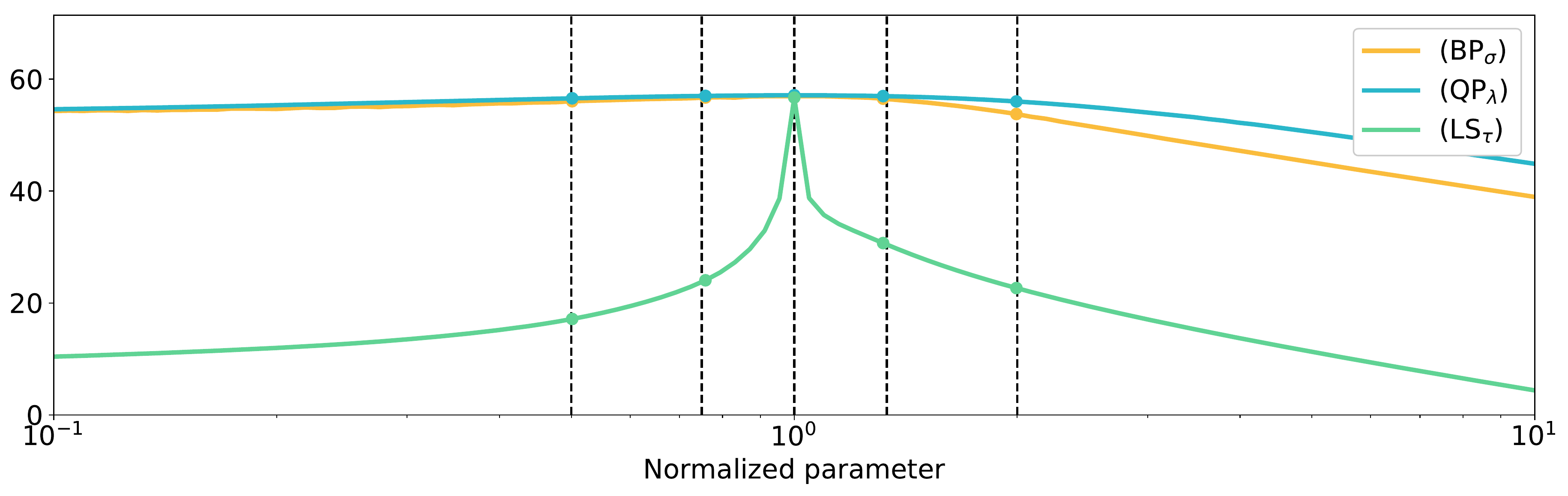}
  
  \includegraphics[width=.8\textwidth]{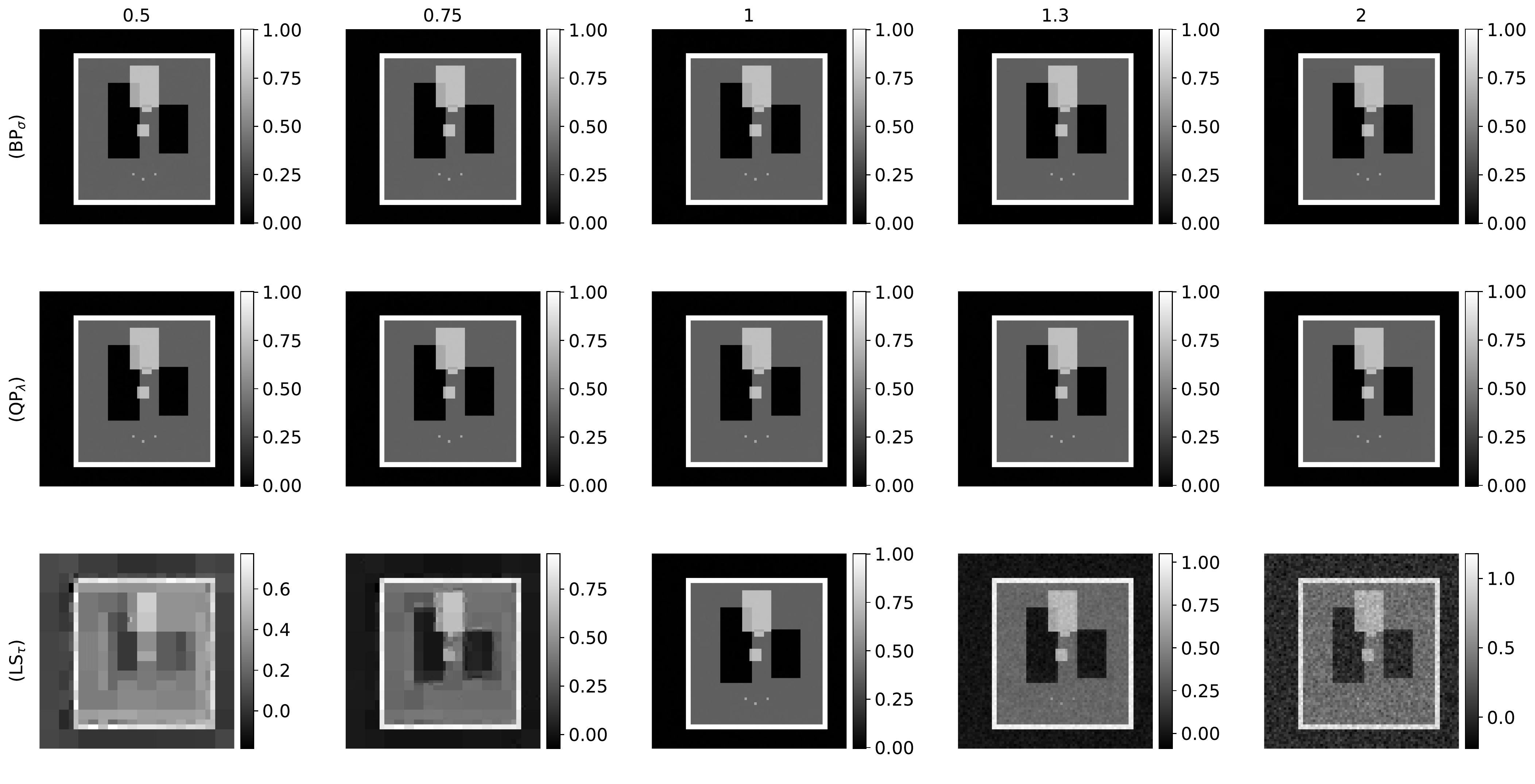}
  \includegraphics[width=.8\textwidth]{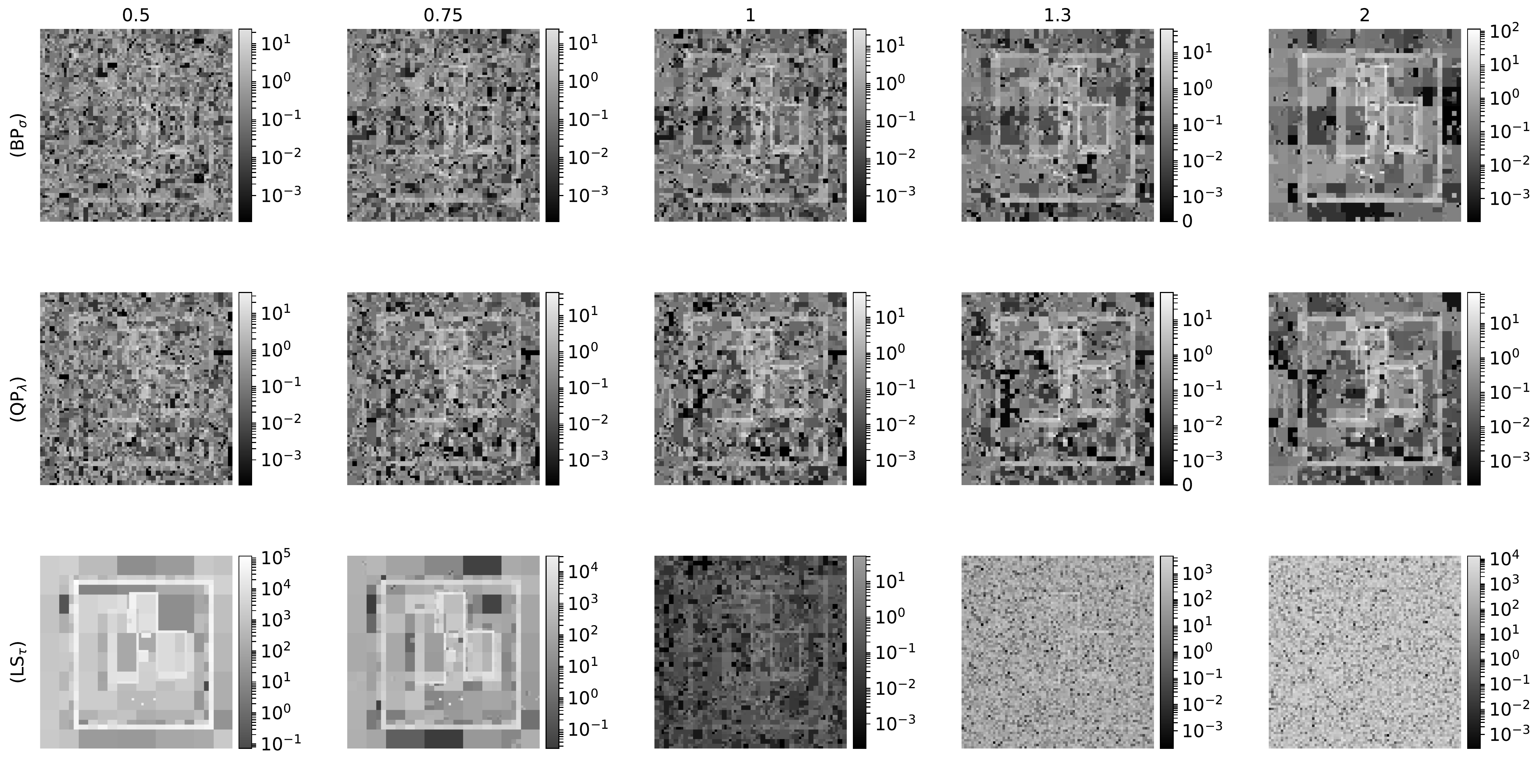}
  
  \caption{Wavelet space compressed sensing problem with the square Shepp-Logan
    phantom for different values of the normalized parameter when
    $(s, N, m, \eta) = (416, 6418, 3110, 2\cdot 10^{-3})$. \textbf{Top:} The
    sections of the psnr surface for which estimator recovery will be
    visualized are depicted by the dots which lie nearly on the black dotted
    lines, themselves located at $\rho = 0.5, 0.75, 1, 4/3,
    2$. \textbf{Middle:} This first $3 \times 5$ group of plots shows each
    program's solution for a particular value of the normalized
    parameter. Pixel values in an image are associated to its colour
    bar. \textbf{Bottom:} This $3\times 5$ group of plots depicts pixel-wise
    nnse for each (program, parameter) pairing. The program is denoted along
    the left-hand side, while the normalized parameter value is denoted along
    top row.}
  \label{fig:lasso-sslp-grid-plots-1}
\end{figure}

\begin{figure}[h]
  \centering
  \includegraphics[width=.8\textwidth]{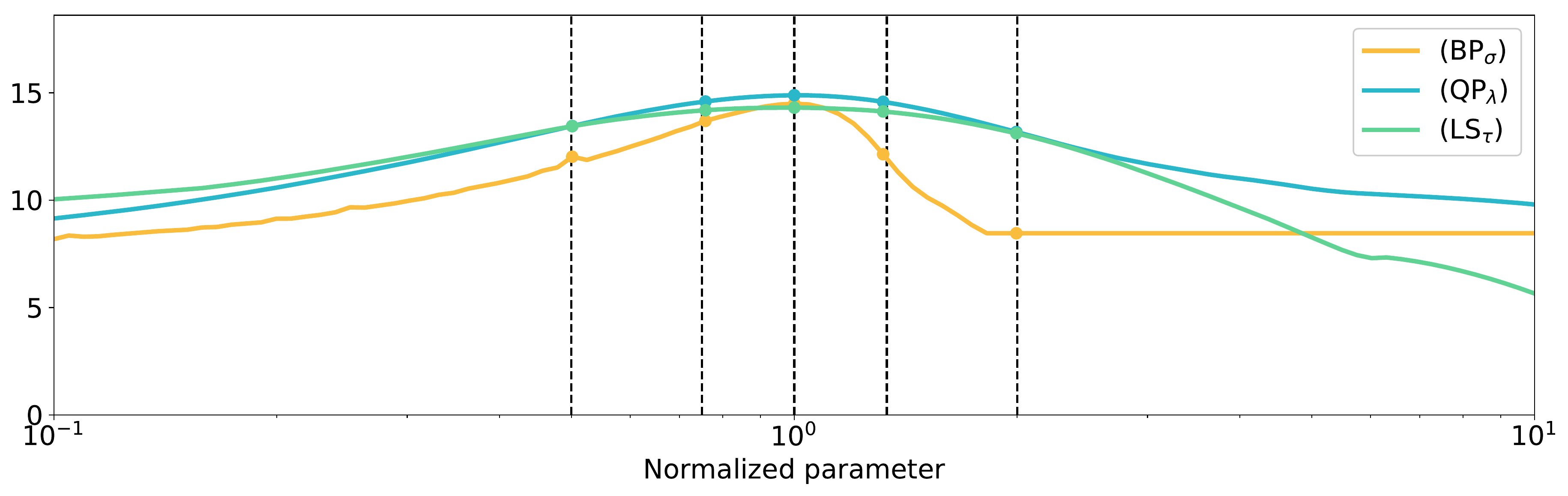}

  \includegraphics[width=.8\textwidth]{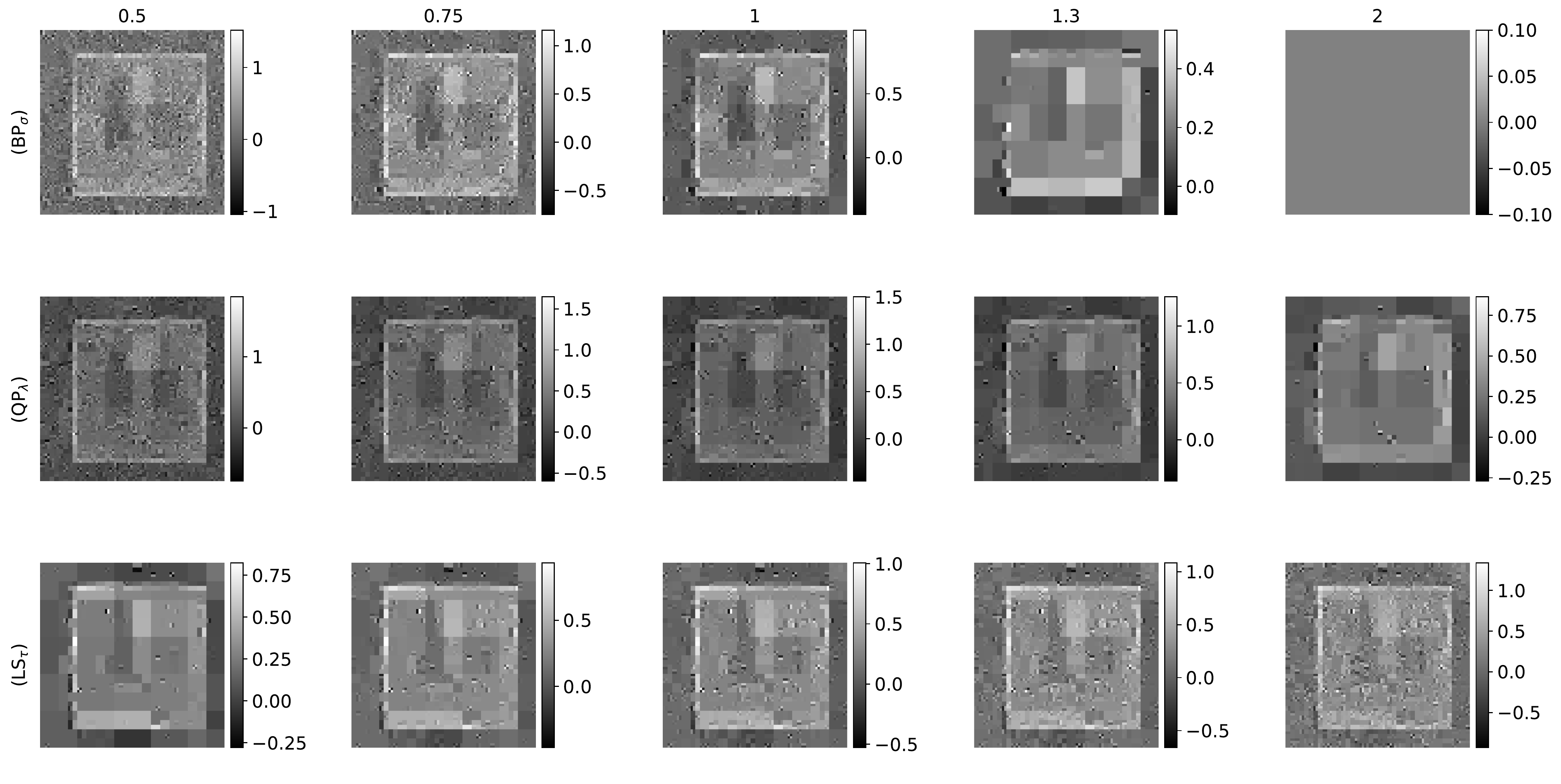}
  \includegraphics[width=.8\textwidth]{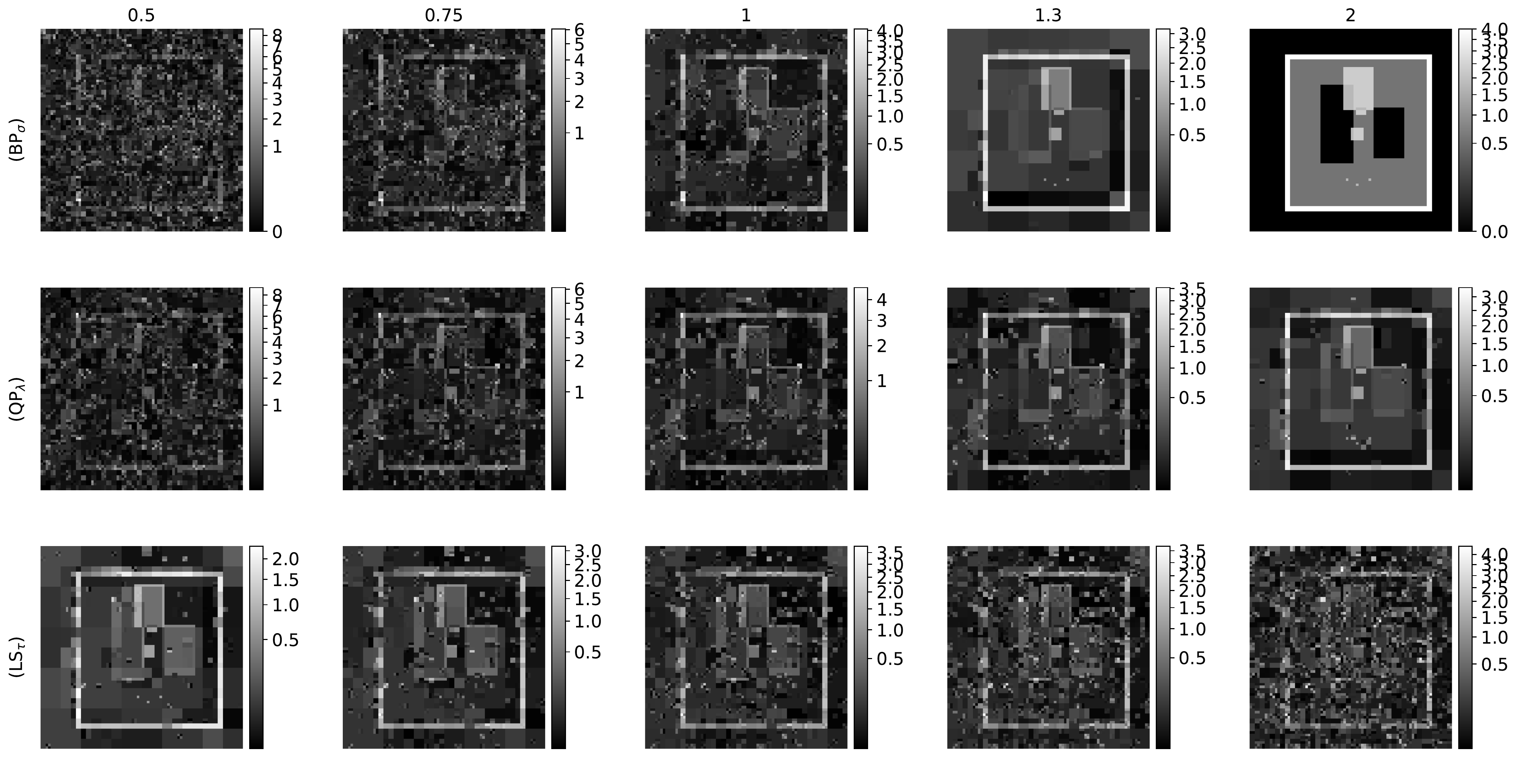}
  
  \caption{Wavelet space compressed sensing problem with the square Shepp-Logan
    phantom for different values of the normalized parameter when
    $(s, N, m, \eta) = (416, 6418, 3110, 0.5)$. \textbf{Top:} The sections of
    the psnr surface for which estimator recovery will be visualized are
    depicted by the dots which lie nearly on the black dotted lines, themselves
    located at $\rho = 0.5, 0.75, 1, 4/3, 2$. \textbf{Middle:} This first
    $3 \times 5$ group of plots shows each program's solution for a particular
    value of the normalized parameter. Pixel values in an image are associated
    to its colour bar. \textbf{Bottom:} This $3\times 5$ group of plots depicts
    pixel-wise nnse for each (program, parameter) pairing. The program is
    denoted along the left-hand side, while the normalized parameter value is
    denoted along top row.}

  \label{fig:lasso-sslp-grid-plots-2}
\end{figure}

\clearpage
\section{Proofs}
\label{sec:proofs}

\subsection{Proof of worst-case risk equivalence}
\label{sec:proof-worst-case-risk-equivalence}

\begin{lemma}[Increasing risk]
  \label{lem:increasing-risk}
  Fix $x_{0} \in \Sigma_{s}^{N}$, $\|x_{0}\|_{1} = 1$.  Then
  $\hat R(\tau; \tau x_{0}, N, \eta)$ is an increasing function of
  $\tau \geq 0$.
\end{lemma}

\begin{proof}[Proof of {Lemma \ref{lem:increasing-risk}}]
  Given $y(\tau) := \tau x_{0} + \eta z$ for $\eta > 0$ and
  $z \in \reals^{N}, z_{i} \iid \mathcal{N}(0,1)$, let
  $\hat x(\tau) := \hat x(\tau; y(\tau))$ solve
  \begin{align*}
    \hat x(\tau; y(\tau)) %
    := \argmin_{x} \{ \| y(\tau) - x\|_{2} : %
      \|x\|_{1} \leq \tau \}
  \end{align*}
  Let $K := B_{1}^{N} - x_{0}$, a convex set containing the origin. Using a
  standard scaling property of orthogonal projections,
  \begin{align*}
    \hat x(\tau) - x_{0}
    & = \argmin_{w} \{ \| \eta z - w \|_{2} : \|w + \tau x_{0}\|_{1} \leq \tau \} \\
    & = \mathrm{P}_{\tau K} (\eta z). %
  \end{align*}
  Hence, it follows by Lemma \ref{lem:projection-lemma} that
  $\|\hat x(\tau) - \tau x_{0}\|_{2}$ is an increasing function of $\tau$.
\end{proof}

\begin{proposition}[Risk equivalence]
  \label{prop:risk-equivalence}
  Let $\eta, \tau > 0$ and fix $N \geq 2$. Then
  \begin{align*}
    \sup_{x \in \Sigma_{s}^{N}} %
    \hat R(\|x\|_{1}; x, N, \eta) %
    = \max_{\substack{x \in \Sigma_{s}^{N}\\\|x\|_{1}=1}} %
    \lim_{\tau \to \infty} \hat R(\tau; \tau x, N, \eta) %
    = \max_{\substack{x \in \Sigma_{s}^{N}\\\|x\|_{1}=1}} %
    \lim_{\eta \to 0} \hat R (1 ; x, N, \eta). 
  \end{align*}
\end{proposition}
\begin{proof}[Proof of {Proposition \ref{prop:risk-equivalence}}]
  The first equality is an immediate consequence of
  Lemma \ref{lem:increasing-risk}:
  \begin{align*}
    \sup_{x \in \Sigma_{s}^{N}} \hat R(\|x\|_{1}; x, N, \eta) %
    = \max_{\substack{x \in \Sigma_{s}^{N}\\\|x\|_{1}=1}} %
    \sup_{\tau > 0} \hat R(\tau; \tau x, N, \eta) %
    = \max_{\substack{x \in \Sigma_{s}^{N}\\\|x\|_{1}=1}} %
    \lim_{\tau\to\infty} \hat R(\tau; \tau x, N, \eta). 
  \end{align*}
  The second equality follows from a standard property of orthogonal
  projections, and the risk expression derived in
  Lemma \ref{lem:increasing-risk}. For $K := B_{1}^{N} - x$,
  \begin{align*}
    \max_{\substack{x \in \Sigma_{s}^{N}\\\|x\|_{1}=1}} %
    \lim_{\tau\to\infty} \hat R(\tau; \tau x, N, \eta) %
    & = \max_{\substack{x \in \Sigma_{s}^{N}\\\|x\|_{1}=1}} %
    \lim_{\tau\to\infty} \eta^{-2}\| \mathrm{P}_{\tau K}(\eta z) \|_{2}^{2}
    = \max_{\substack{x \in \Sigma_{s}^{N}\\\|x\|_{1}=1}} %
    \lim_{\tau\to\infty} \frac{\tau^{2}}{\eta^{2}}\| \mathrm{P}_{K}(\tau^{-1}\eta z) \|_{2}^{2}
    \\
    & = \max_{\substack{x \in \Sigma_{s}^{N}\\\|x\|_{1}=1}} %
    \lim_{\tilde \eta \to 0} \tilde \eta^{-2}\| \mathrm{P}_{K}(\tilde \eta z) \|_{2}^{2}
    = \max_{\substack{x \in \Sigma_{s}^{N}\\\|x\|_{1}=1}} %
    \lim_{\eta \to 0} \hat R(1; x, N, \eta).
  \end{align*}
\end{proof}

\subsection{Proof of {\lspd} optimal risk}
\label{sec:proof-lspd-optimal-risk}

\begin{proof}[Proof of Proposition \ref{prop:lspd-optimal-risk}]
  Directly from Theorem \ref{thm:hassibi-2-1},
  \begin{align*}
    R^{*}(s, N) %
    = \max_{\substack{x_{0} \in \Sigma_{s}^{N}\\\|x_{0}\|_{1}=1}}
    \mathbf{D}(T_{B_{1}^{N}}(x_{0})^{\circ}) 
  \end{align*}
  where $\mathbf{D}(T_{B_{1}^{N}}(x_{0})^{\circ})$ is the mean-squared distance
  to the polar of the $\ell_{1}$ descent cone. The operator $\mathbf{D}$ has the
  following desirable relation to the Gaussian mean width, where $\mathcal{C}$
  is a non-empty convex cone \cite[Prop 10.2]{amelunxen2014living}:
  \begin{align*}
    w^{2}(\mathcal{C} \cap \mathbb{S}^{N-1}) %
    \leq \mathbf{D}(C^{\circ}) %
    \leq w^{2}(\mathcal{C} \cap \mathbb{S}^{N-1}) + 1. 
  \end{align*}
  Thus, it suffices to lower- and upper-bound
  $w^{2}\big( T_{B_{1}^{N}} \cap \sph^{N-1}\big)$. The desired upper bound is an
  elementary but technical exercise using H\"older's inequality, Stirling's
  approximation and a bit of calculus. The lower bound may be computed using
  Sudakov's inequality and \cite[Lemma 10.12]{foucart2013mathematical}. It
  thereby follows that
  \begin{align*}
    c s\log (N/s) %
    \leq \mathbf{D}(T_{B_{1}^{N}}(x_{0})^{\circ})
    \leq C s \log (N/s).
  \end{align*}
  where $c, C > 0$ are universal constants. Accordingly,
  $c s\log (N/s) \leq R^{*}(s,N) \leq C s \log
  (N/s)$.

  From Theorem \ref{thm:qppd-rhs-stability},
  $R^{\sharp}(\lambda^{*}; s, N) \leq C s \log N$ for any $N \geq N_{0}(s)$ with
  $N_{0}(s)$ sufficiently large. Using the above equation gives, for
  $c, C_{1} > 0$, $C s \log N \leq C_{1} c s\log(N/s) \leq C_{1} R^{*}(s,
  N)$. Finally, observe that $R^{\sharp}(\lambda^{*}; s, N)$ is trivially lower
  bounded by $M^{*}(s, N) = \Theta(s\log(N/s))$ \cite{candes2013well}.
\end{proof}

\subsection{Proof of $\ell_{1}$ tangent cone equivalence}
\label{sec:proof-tangent-cone-equivalence}

\begin{proof}[Proof of {Lemma \ref{lem:descent-cone-equivalence}}]
  First observe that the definition of $F_{\mathcal{C}}(x)$ is equivalent to
  \begin{align*}
    F_{\mathcal{C}}(x) = \{ h \in \reals^{N} %
    : h = z-x,\, \|z\|_{1} \leq \|x\|_{1}\}. 
  \end{align*}
  Next, observe that $K(x)$ is a cone. So, for left containment, it suffices to
  show $F_{\mathcal{C}}(x) \subseteq K(x)$ since the cone generated by a set is
  no larger than any cone containing that set. These two expressions:
  \begin{align*}
    \ip{\mathrm{sgn}(x)_{T}, x} %
    &= \|x\|_{1} %
    \geq \|z\|_{1} %
    = \|z_{T}\|_{1} + \|h_{T^{C}}\|_{1} %
    \\
    \|z_{T}\|_{1} %
    &= \ip{\mathrm{sgn}(z), z_{T}} %
      \geq \ip{\mathrm{sgn}(x), z_{T}},
  \end{align*}
  are by definition of $h = z - x \in F_{\mathcal{C}}(x)$. They combine to
  yield left containment:
  \begin{align*}
    \|h_{T^{C}}\|_{1} %
    \leq - \ip{\mathrm{sgn}(x), z_{T} - x} %
    = - \ip{\mathrm{sgn}(x), h_{T}} %
    = - \ip{\mathrm{sgn}(x), h}.
  \end{align*}
  To show right containment, first fix $w \in K(x)$ and select $\alpha \geq 0$
  sufficiently small so that $z := x + \alpha w$ admits $z_{j}x_{j} \geq 0$ for
  all $j \in T$.
  Using $\alpha \| w_{T^{C}}\|_{1} \leq - \alpha \ip{\mathrm{sgn}(x), w_{T}}$,
  we show $\|z\|_{1}\leq \|x\|_{1}$ implying that
  $\alpha w \in F_{\mathcal{C}}(x)$, whence $w \in T_{\mathcal{C}}(x)$. Where
  $h := \alpha w = z-x$,
  \begin{align*}
    \|z\|_{1} %
    & = \|z_{T}\|_{1} + \|z_{T^{C}}\|_{1} %
      = \ip{\mathrm{sgn}(z_{T}), z_{T}} + \|h_{T^{C}}\|_{1} %
    \\
    &\leq \ip{\mathrm{sgn}(z_{T}), z_{T}} - \ip{\mathrm{sgn}(x), h_{T}}
      = \ip{\mathrm{sgn}(z_{T}), z_{T}} - \ip{\mathrm{sgn}(x), h_{T}}
      + \ip{\mathrm{sgn}(x), x} - \ip{\mathrm{sgn}(x), x}
    \\
    & = \ip{\mathrm{sgn}(x), x} + \ip{\mathrm{sgn}(z_{T}), z_{T}}
      - \ip{\mathrm{sgn}(x), z_{T}} %
      = \|x\|_{1} + \ip{\mathrm{sgn}(z_{T}) - \mathrm{sgn}(x), z_{T}} %
      = \|x\|_{1}
  \end{align*}
  where the latter equality follows from the fact that
  $\ip{\mathrm{sgn}(z_{T}) - \mathrm{sgn}(x), z_{T}} \neq 0$ only if
  $x_{j}z_{j} < 0$ for some $j \in T$, which goes against the initial
  assumption defining $\alpha$ and $z$. Thus, $w \in T_{\mathcal{C}}(x)$ and
  $T_{\mathcal{C}}(x) = K(x)$ as desired.
\end{proof}

\subsection{Proof of the projection lemma}
\label{ssec:proof-proj-lemma}


\begin{proof}[Proof of {Lemma \ref{lem:projection-lemma}}]
  Define $z_{\alpha} := \mathrm{P}_{\alpha K}(z)$ for $\alpha = 1, \lambda$ and
  define $f(t) := \|u_{t}\|_{2}^{2}$, where
  $u_{t} := t z_{\lambda} + (1-t) z_{1}$ for $t \in [0, 1]$. Our goal is to show
  $\left.\dee{}{t}\right|_{t = 0} f(t) \geq 0$; this implies
  $\|z_{\lambda}\|_{2} \geq \|z_{1}\|_{2}$, because $f$ is convex. Expanding
  $f(t)$,
\begin{align*}
  f(t) %
  &= t^{2} \big(( \|z_{\lambda}\|_{2}^{2} - 2 \ip{z_{1}, z_{\lambda}} %
    + \|z_{1}\|_{2}^{2} \big) %
    + 2t \big( \ip{z_{1}, z_{\lambda}} - \|z_{1}\|_{2}^{2} \big) %
    + \|z_{1}\|_{2}^{2}.
\end{align*}
So it is required to check the condition $(\star)$:
\begin{align*}
  \left.\dee{}{t} f(t)\right|_{t=0} %
  &= \left[ 2t \|z_{\lambda} - z_{1}\|_{2}^{2}
    + 2 \ip{z_{1}, z_{\lambda} - z_{1}} \right|_{t=0} %
    = 2 \ip{z_{1}, z_{\lambda} - z_{1}} \overset{(\star)}{\geq} 0 
\end{align*}
The projection condition says that if $\mathrm{P}_{C}(x)$ is the projection of
$x$ onto a convex set $C$ then for any $y \in C$,
$\ip{y - \mathrm{P}_{C}(x), x - \mathrm{P}_{C}(x)} \leq 0$.  From the projection
condition \cite{bertsekas2003convex}, we have
\begin{itemize}[topsep=8pt, itemsep=1pt]
\item $\ip{\lambda^{-1} z_{\lambda} - z_{1}, z- z_{1}} \leq 0$
\item $\ip{\lambda z_{1} - z_{\lambda}, z - z_{\lambda}} \leq 0$.
\end{itemize}
Accordingly,
\begin{align*}
  0 &\geq \ip{z_{\lambda} - \lambda z_{1}, z-z_{1}}
      + \ip{\lambda z_{1} - z_{\lambda}, z - z_{\lambda}}
  \\
    & = \ip{\lambda z_{1} - z_{\lambda}, z_{1}-z_{\lambda}}
      = \ip{(\lambda-1)z_{1}, z_{1} - z_{\lambda}} + \|z_{1} - z_{\lambda}\|_{2}^{2} %
  \\
    & \geq (\lambda - 1)\ip{z_{1}, z_{1} - z_{\lambda}}
\end{align*}
which is equivalent to $\ip{z_{1}, z_{\lambda} - z_{1}} \geq 0$. Therefore, $f$
is a convex function increasing on the interval $t \in [0, 1]$, whence
$\|z_{1}\|_{2} \leq \|z_{\lambda}\|_{2}$ as desired.
\end{proof}

\begin{remark}
  There is a simpler way to begin the proof of the projection lemma. To show
  \begin{align*}
    \|z_{1}\|_{2} \leq \|z_{\lambda}\|_{2} \quad \iff \quad
    \|z_{1}\|_{2}^{2} \leq \|z_{1}\|_{2}\|z_{\lambda}\|_{2},
  \end{align*}
  one may instead prove the following chain,
  \begin{align*}
    \|z_{1}\|_{2}^{2} \leq \ip{z_{1}, z_{\lambda}} \leq \|z_{1}\|_{2}\|z_{\lambda}\|_{2}. 
  \end{align*}
  The latter inequality is true by Cauchy-Schwarz, so it remains only to prove the former:
  \begin{align*}
    \ip{z_{1}, z_{\lambda}} - \|z_{1}\|_{2}^{2} \geq 0 \quad \iff \quad %
    \ip{z_{1}, z_{\lambda} - z_{1}} \geq 0. 
  \end{align*}
  Rearranging shows this inequality is equivalent to $(\star)$, and the
  remainder of the proof proceeds as is. This remark is included for intuition,
  but this approach is less generalizable. For example, it does not yield the
  rate of growth observed in the remark at the end of
  \ref{sssec:projection-lemma}.

\end{remark}

\subsection{Elementary results from probability}
\label{ssec:elem-results-from}



We briefly recall two aspects of how normal random vectors concentrate in high
dimensions.

\begin{proposition}
  \label{prop:bppd-oc-const-prob}
  Let $z \in \reals^{N}$ with $z_{i} \iid \mathcal{N}(0,1)$, fix constants
  $0 < C_{2} < C_{1} < \infty$ and define the event $\mathcal{Z}_{\pm}$ by
  $\mathcal{Z}_{\pm} := \{C_{2} \sqrt{2N} \leq \|z\|_{2}^{2} - N \leq C_{1}
  \sqrt{2N}\}$. There exists a constant $p = p(C_{1}, C_{2}) > 0$ and integer
  $N_{0} \geq 1$ such that for all $N \geq N_{0}$,
  \begin{align*}
    \mathbb{P}\big( \mathcal{Z}_{\pm}\big) \geq p
  \end{align*}

\end{proposition}

\begin{proof}[Proof of {Proposition \ref{prop:bppd-oc-const-prob}}]
  Define the $\chi^{2}_{N}$-distributed random variable
  \begin{align*}
    X_{N} := \frac{\|z\|_{2}^{2} - N}{\sqrt{2 N}}. 
  \end{align*}
  Since $X_{N} \xrightarrow{N\to\infty} \mathcal{N}(0,1)$ by the central limit
  theorem, for any $\varepsilon > 0$ there is $N_{0} \in \nats$ such that
  \begin{align*}
    \big| \mathbb{P}\big( X_{N} \leq t\big) - \Phi(t)\big| \leq \varepsilon
  \end{align*}
  for all $N \geq N_{0}$, where $\Phi$ is the standard normal cdf. One need
  merely choose $\varepsilon > 0$ so that
  \begin{align*}
    \mathbb{P}\big( \mathcal{Z}_{\pm}\big) %
    \geq \Phi(C_{2}) - \Phi(C_{1}) - 2\varepsilon
    =: p(C_{1},C_{2}) > 0
  \end{align*}
  and choose the first $N_{0}$ for which the chain of inequalities is valid for
  all $N \geq N_{0}$.
\end{proof}


\begin{corollary}
  \label{coro:large-deviation}
  Fix $N, N_{0} \in \nats$ with $N \geq N_{0} \geq 2$. Let $z \in \reals^{N}$
  with $z_{i} \iid \mathcal{N}(0,1)$ and define the event
  \begin{align*}
    A_N := \big\{ \|z\|_2^2 \leq N - 2 \sqrt N %
    \quad\And\quad %
    \|z\|_\infty \leq \sqrt{3 \log N}\big\}
  \end{align*}
  There exists a real constant $C = C(N_{0}) > 0$ such that $\mathbb{P}(A_{N}) \geq C$. 
\end{corollary}

\begin{proof}[Proof of {Corollary \ref{coro:large-deviation}}]
  Given $N$, define the events $E_{N} := \{ \|z\|_{2}^{2} \leq N - 2 \sqrt N \}$
  and $F_{N} := \{ \|z\|_{\infty} \leq \sqrt{3 \log N} \}$. Using the standard
  identity $\Phi(-x) \leq \phi(x) / x$, we note that
  \begin{align*}
    \mathbb{P}(F_{N}) %
    \geq 1 - 2 N \mathbb{P}( |Z| > \sqrt{3\log N}) %
    \geq 1 - \frac{2}{\sqrt{\frac32 \pi N_{0} \log N_{0} }} %
    > 0
  \end{align*}
  With this, and a standard argument similar to that of
  Proposition \ref{prop:bppd-oc-const-prob}, one may show
  \begin{align*}
    \mathbb{P}(A_{N}) %
    = \mathbb{P}(E_{N}F_{N}) %
    = \mathbb{P}(E_{N}\mid F_{N}) \mathbb{P}(F_{N}) %
    \geq \mathbb{P}(E_{N}) \mathbb{P}(F_{N})
    \geq C %
    \geq 0. 
  \end{align*}
\end{proof}

We also recall that an event holding with high probability, intersected with an
event occurring with constant probability, still occurs with constant
probability.


\begin{proposition}
  \label{prop:bppd-oc-const-whp-prob}
  Let $N \geq 1$ be an integer and suppose that $\mathcal{E} = \mathcal{E}(N)$
  is an event that holds with high probability in the sense that
  \begin{align*}
    \mathbb{P}(\mathcal{E}_{N}) \geq 1 - p(N)
  \end{align*}
  for some function $p(N) > 0$ with $\lim_{N\to\infty} p(N) = 0$. Suppose also
  that for an event $\mathcal{F} = \mathcal{F}(N)$ there exists $q > 0$ such
  that $\inf_{N\geq 1} \mathbb{P}(\mathcal{F}(N)) \geq q$. Then there exists a
  constant $q' > 0$ and integer $N_{0} \geq 1$ such that
  $\mathbb{P}\big(\mathcal{E}(N) \cap \mathcal{F}\big) \geq q'$ for all
  $N \geq N_{0}$.
\end{proposition}

\begin{proof}[Proof of {Proposition \ref{prop:bppd-oc-const-whp-prob}}]
  The proof is very similar to that of Proposition \ref{prop:bppd-oc-const-prob}. Simply
  choose a threshold $\varepsilon > 0$ and select the first $N_{0} \geq 1$ for
  which
  \begin{align*}
    \mathbb{P}\big(\mathcal{E}(N) \cap \mathcal{F}\big) %
    \geq q - p(N) %
    \geq q - \varepsilon %
    =: q' %
    > 0
  \end{align*}
  for all $N \geq N_{0}$. 
\end{proof}

\begin{remark}
  An example of such a $p(N)$ as in Proposition \ref{prop:bppd-oc-const-whp-prob} is
  $p(N) \sim O(e^{-N})$ when $\mathcal{E}_{N} := \{ |X-\mu| \leq t\}$ for $X$ a
  subgaussian random variable, $\E X = \mu$ and $t > 0$.
\end{remark}

\subsection{Proof of {\lspd} parameter instability}
\label{ssec:proof-lspd-instability}



\begin{proof}[Proof of {Theorem \ref{thm:constr-pd}}]

  Let $x_{0} \in \Sigma_{s}^{N}$ with non-empty support and let $\tau > 0$ be
  the governing parameter of {\lspd}. First suppose the parameter is chosen
  smaller than the optimal value, \ie $\tau < \|x_{0}\|_{1}$. The discrepancy
  of the guess, $\rho := |\|x_{0}\|_{1} - \tau| = \|x_{0}\|_{1} - \tau > 0$,
  induces the instability.

  The solution $\hat x(\tau)$ to {\lspd} satisfies
  $0 \leq \|\hat x(\tau) \|_{1} \leq \tau$ by
  construction. Therefore, by the Cauchy-Schwarz inequality and an application
  of the triangle inequality,
  \begin{align*}
    \|\hat x(\tau) - x_{0} \|_{2}^{2} %
    \geq N^{-1} \|\hat x(\tau) - x_{0}\|_{1}^{2} %
    \geq  \frac{\rho^{2}}{N} > 0. 
  \end{align*}
  Accordingly, 
  \begin{align*}
    \lim_{\eta \to 0} \frac{1}{\eta^{2}} \|\hat x(\tau) - x_{0}\|_{2}^{2} %
    \geq \lim_{\eta \to 0} \frac{\rho^{2}}{N \eta^{2}} %
    = \infty. 
  \end{align*}
  Next assume $\tau$ is chosen too large, with discrepancy between the correct
  and actual guesses for the parameter again being denoted
  $ \rho = \tau - \|x_{0}\|_{1} > 0$. Two key pieces of intuition guide this
  result. The first is that the error of approximation should be controlled by
  the effective dimension of the constraint set. The second suggests that $y$
  continues to lie within the constraint set for sufficiently small noise level,
  meaning recovery behaves as though it were unconstrained. Hence, the effective
  dimension of the problem is that of the ambient dimension, and so one should
  expect the error to be proportional to $N$.

  First, we show that for $\eta$ sufficiently small, $y \in \tau B_{1}^{N}$ with
  high probability. Fix a sequence $\eta_{j} \xrightarrow{j\to\infty} 0$ and
  define $y_{j} := x_{0} + \eta_{j} z$. Since $\|z\|_{1}$ is subgaussian,
  Theorem \ref{thm:hoeffding-subgaussian} implies there is a constant $C > 0$ such
  that
  \begin{align*}
    \mathbb{P}\big( \|z\|_{1} \geq t + N\sqrt{\frac{2}{\pi}}\big) %
    \leq \mathbb{P}\big( \big| \|z\|_{1} - N\sqrt{\frac{2}{\pi}} \big |
    \geq t \big) %
    \leq e \cdot \exp\big( \frac{-t^{2}}{C N} \big). 
  \end{align*}
  In order to satisfy $x_{0} + \eta z \in \tau B_{1}^{N}$, we need
  $\|x_{0} + \eta z \|_{1} < \tau$, for which $\eta \|z\|_{1} < \rho$ is
  sufficient. The probability that this event does not occur is upper
  bounded by
  \begin{align*}
    \mathbb{P}\big( \|z\|_{1} \geq \frac{\rho}{\eta}\big) %
    \leq \mathbb{P}\big( \|z\|_{1} \geq t + N \sqrt{\frac{2}{\pi}}\big)  %
    \leq e \cdot \exp\big( -\frac{t^{2}}{CN}\big)
  \end{align*}
  For $t = \rho/\eta - N\sqrt{\frac{2}{\pi}}$, and $\tilde C > 0$ a new
  constant,
  \begin{align*}
    \mathbb{P}\left(\|z\|_{1} %
    \geq \frac{\rho}{\eta}\right) %
    \leq e\cdot \exp\left( - \frac{ %
    \big(\rho/\eta - N \sqrt{2/\pi}\big)^{2}}{CN}\right) %
    \lesssim \tilde C \exp\left( - \frac{ \rho^{2}}{N\eta^{2}}\right) %
    \xrightarrow{\eta \to 0} 0. 
  \end{align*}
  Let $E_{j} := \{ \|z\|_{1} < \displaystyle\frac{\rho}{\eta_{j}}\}$ for
  $j \geq 1$; their respective probabilities lower-bounded by
  $p_{j} := 1 - \tilde C \exp(-\rho^{2} / N\eta_{j}^{2})$. Given
  $0 < \varepsilon \ll 1$, denote by $j_{0}$ the first integer such that
  $p_{j} \geq 1 - \varepsilon$ for all $j \geq j_{0}$. On $E_{j}$ with
  $j \geq j_{0}$, $y_{j} \in \tau B_{1}^{N}$ so $y_{j}$ is the unique minimizer
  of {\lspd}, meaning:
  \begin{align*}
    \frac{1}{\eta^{2}} \|\hat x(\tau) - x_{0}\|_{2}^{2} = \|z\|_{2}^{2}. 
  \end{align*}
  The result follows by bounding the following expectations:
  \begin{align}
    \lim_{\eta \to 0} \frac{1}{\eta^{2}} \E \|\hat x(\tau) - x_{0}\|_{2}^{2} %
    = \lim_{j \to \infty } \E \big[ \eta_{j}^{-2} %
      \|\hat x(\tau) - x_{0}\|_{2}^{2} \1(E_{j}) \big] %
      + \E \big[ \eta_{j}^{-2} \|\hat x(\tau) - x_{0}\|_{2}^{2} %
      \1(E_{j}^{C}) \big]. \tag{$\star$} %
  \end{align}
  The first term converges by dominated convergence theorem:
  \begin{align*}
    \lim_{j \to \infty } \E \big[ \eta_{j}^{-2} %
      \|\hat x(\tau) - x_{0}\|_{2}^{2} \1(E_{j}) \big] %
    = \lim_{j\to \infty} \E \big[ \1(E_{j}) \|z\|_{2}^{2} \big]
      = \E \big[ \|z\|_{2}^{2} \big]
      = N.
  \end{align*}
  On $E_{j}^{C}$,
  $\|\hat x(\tau) - x_{0}\|_{2}^{2} \leq \|\hat x(\tau) - x_{0}\|_{1}^{2} \leq
  \rho^{2}\eta^{2}$, so by dominated convergence theorem,
  \begin{align*}
    \lim_{j \to \infty} \E\big[ \eta_{j}^{-2} \|\hat x(\tau) - x_{0}\|_{2}^{2}
    \1(E_{j}^{C})\big] %
    \leq \lim_{j \to \infty} \E\big[ \|z\|_{1}^{2} \1(E_{j}^{C})\big] = 0.
  \end{align*}
  This immediately yields the desired result,
  \begin{align*}
    \lim_{\eta \to 0} \frac{1}{\eta^{2}} \E \|\hat x(\tau) - x_{0}\|_{2}^{2} = N. 
  \end{align*}
  To prove the final case where $\tau = \|x_{0}\|_{1}$, set
  $\mathcal{C} = B_{1}^{N}$ in \eqref{eq:oymakhassibi} of Theorem
  Theorem \ref{thm:hassibi-2-1}. Then,
  \begin{align*}
    \lim_{\eta\to 0 } \eta^{-2} \E \|\hat x (\tau) - x_{0}\|_{2}^{2} %
    = \mathbf{D}(T_{\mathbb{B}_{1}^{N}}(x_{0})^{\circ}) = \Theta(s \log (N/s)) %
    \ll N.
  \end{align*}
\end{proof}

\subsection{Proofs of {\qppd} results}
\label{ssec:proofs-qppd-results}


\begin{proof}[Proof of {Proposition \ref{prop:qppd-smooth}}]
  Because $z$ is isotropic and iid, one can split the signal
  $x_{0} = x_{0}^+ - x_{0}^-$ into ``positive'' and ``negative'' components, and
  so it suffices to consider the case where $x_{0,j} \geq 0$ for all
  $j \in [N]$. The heart of this proposition again relies on the fact that the
  noise limits to $0$. In general, $\lambda > 0$ is finite and typically small
  ($\sim \mathcal{O}(\eta\sqrt{\log N})$, so we require only that
  $|x_{0,j}| = \mathcal{O}(1)$ for $j \in T = \mathrm{supp}(x_{0})$. This
  requirement can be written $x_{0,j} \geq a > 0$ for all $j \in T$ and some
  real number $a > 0$. Recall that the minimizer of {\qppd} is given by the
  soft-thresholding operator which we denote by
  \begin{align*}
    x^\sharp(\eta\lambda) = S_{\eta\lambda}(x_{0} + \eta z). 
  \end{align*}
  Where $k \in T, \ell \in T^{C}$ so that $x_{0,k} \geq a$, $x_{0,\ell} = 0$,
  one has
  \begin{align*}
    S_{\eta\lambda}(x_{0,k} + \eta z_k) - x_{0,k} &=
    \begin{cases}
      \eta(z_{k}-\lambda) & x_{0,k} > \eta(\lambda-z_k)\\
      -x_{0,k} & |x_{0,k} + \eta z_k | \leq \eta \lambda\\
      \eta(z_{k} + \lambda) & x_{0,k} < - \eta (\lambda+z_k)
    \end{cases}
                        &
    S_{\eta\lambda}(\eta z_\ell) &= \begin{cases}
      \eta(z_\ell-\lambda) & z_\ell > \lambda\\
      0 & |z_\ell| \leq \lambda\\
      \eta(z_\ell + \lambda) & z_\ell < -\lambda
    \end{cases}
  \end{align*}
  and so independence of $z_{j}$ yields
  \begin{align*}
    \lim_{\eta\to 0 } \frac{1}{\eta^2} \E \|x^\sharp(\eta\lambda) - x_{0}\|_2^2 %
    = \lim_{\eta\to 0 } \frac{s}{\eta^2}
    \E \big[ (S_{\eta\lambda}(x_{0,k} + \eta z_{k}) - x_{0,k})^{2}\big]
    + \lim_{\eta\to 0 } \frac{N-s}{\eta^2}
    \E \big[ S_{\eta\lambda}(z_{\ell})^{2} \big].
  \end{align*}
  Passing to a sequence $\eta_{j} \to 0$, there exists
  $J \in \nats$ such that for all $j \geq J$,
  \begin{align}
    S_{\eta_{j}\lambda}(x_{0,k}+\eta_{j}z_{k}) -x_{0,k} %
    = \eta_{j}(z_{k} - \lambda) %
    \qquad\text{with high probability.}
    \tag{$\star$}
  \end{align}
  If this equality were true almost surely, it would follow that for $k \in T$
  \begin{align*}
    \eta^{-2}\E \| (x^{\sharp}(\eta\lambda)- x_{0})_{T}\|_{2}^{2} = \E\big[ (z_{T} - \lambda)^{2} \big] = \E[z_{T}^{2} \big] + s\lambda^{2} = s(1 + \lambda^{2}). 
  \end{align*}
  Indeed this is still true in the case of $(\star)$ with $\eta\to 0$. In
  particular, using independence of $z_{k}$ for $k \in T$ and denoting
  by $E_{j}$ the high probability event $(\star)$, we obtain by
  similar means as in the proof of Theorem \ref{thm:constr-pd},
  \begin{align*}
    \lim_{\eta\to 0} \eta^{-2} \E \|(x^{\sharp}(\eta\lambda) - x_{0})_{T}\|_{2}^{2} %
    &= \lim_{j\to \infty} \frac{s}{\eta_{j}^{2}}
    \E\big[ \eta_{j}^{2}(z_{k}-\lambda)^{2} \1(E_{j}) \big] 
    + \frac{s}{\eta_{j}^{2}} \E\big[ (x^{\sharp}_{k}(\eta_{j}\lambda) - x_{0,k})^{2}
      \1(E_{j}^{C}) \big] 
    = s(1+\lambda^{2})
  \end{align*}
  Next, define $G(\lambda) := (1 + \lambda^{2}) \Phi(-\lambda) - \lambda \phi(\lambda)$. By independence of the entries of $z_{T^{C}}$, with any $\ell \in T^{C}$, the second quantity is exactly computable as
  \begin{align*}
    \lim_{\eta\to 0 } \eta^{-2}\E \|(x^{\sharp}(\eta\lambda) - x)_{T^{C}}\|_{2}^{2} %
    = (N-s) \E\big[ S_{\lambda}(z_{\ell})^{2}\big] = 2 (N-s) G(\lambda),
  \end{align*}
  where the final equality is by definition of $S_{\lambda}$ and elementary
  calculations (\emph{cf.} remark remark \ref{rmk:slickness}). Therefore, as
  desired,
  \begin{align*}
    \lim_{\eta\to 0} \eta^{-2} \E \|x^{\sharp}(\eta\lambda) - x\|_{2}^{2} %
    = s(1+\lambda^{2}) + 2(N-s)G(\lambda).
  \end{align*}
\end{proof}


\begin{proof}[Proof of {Corollary \ref{coro:qppd-max-formulation}}]
  For $0 \leq t \leq s$, where we define for simplicity of notation
  $\Sigma_{-1}^{N} := \emptyset$, observe that
  \begin{align*}
    \sup_{x_{0} \in \Sigma_{t}^{N}\setminus \Sigma_{t-1}^{N}}%
    R^{\sharp}(\lambda; x_{0}, N, \eta) %
    = R^{\sharp}(\lambda; t, N)
  \end{align*}
  because the regime $\eta \to 0$ is equivalent, by a rescaling argument, to the
  regime in which $\eta > 0$ and $|x_{0, j}| \to \infty$ for
  $j \in \mathrm{supp}(x_{0})$ (as shown explicitly in the proof of
  Proposition \ref{prop:risk-equivalence}). Therefore,
  \begin{align*}
    \sup_{x_{0} \in \Sigma_{s}^{N}} R^{\sharp}(\lambda; x_{0}, N, \eta) %
    &= \max_{0 \leq t \leq s} \sup_{x_{0} \in \Sigma_{t}^{N}\setminus \Sigma_{t-1}^{N}}%
      R^{\sharp}(\lambda; x_{0}, N, \eta) %
    \\
    & = \max \{ R^{\sharp}(\lambda; 0, N), R^{\sharp}(\lambda; s, N)\} %
    \\
    & = R^{\sharp}(\lambda; s, N)
  \end{align*}
  by linearity of the $\max$ argument and the fact that
  $1 + \lambda^{2} \geq G(\lambda)$ for $\lambda > 0$. 
\end{proof}

\subsubsection{Proof of {\qppd} parameter instability}
\label{sssec:proof-qppd-instability}


  

We now prove Lemma Lemma \ref{lem:qppd-instability}. 

\begin{proof}[Proof of Lemma Lemma \ref{lem:qppd-instability}]
  By Proposition Proposition \ref{prop:qppd-smooth},
  $R^{\sharp}(\lambda; s, N) = s ( 1+ \lambda^{2}) + 2(N-s) G(\lambda)$. We
  prove the result by controlling $G'(\lambda)$ using integration by
  parts. Thus,
  \begin{align*}
    \frac{\mathrm{d}}{\mathrm{d}\lambda}G(\lambda) %
      = 2\lambda \Phi(-\lambda) - 2\phi(\lambda) %
      \leq 2\lambda(\frac{1}{\lambda} - \frac{1}{\lambda^3} %
      + \frac{3}{\lambda^5})\phi(\lambda) - 2\phi(\lambda) %
      = -2\frac{\lambda^2 - 3}{\lambda^4}\phi(\lambda)
  \end{align*}
  A simple substitution yields, for all
  $N > \exp\Big(\frac{3}{2}(1-\varepsilon)^{-2}\Big)$,
  \begin{align*}
    \left.\frac{\mathrm{d}}{\mathrm{d}u}\right|_{u = 1-\varepsilon}\hskip-22pt G(u\bar \lambda)
    &\leq \left[-2\frac{(u\bar \lambda)^2 - 3}{u^{4}\bar \lambda^{3}}%
    \phi(u\bar \lambda) \right|_{u = 1-\varepsilon} %
    = -\frac{2(1-\varepsilon)^2 \log (N) - 3}{%
      (1-\varepsilon)^4 \sqrt{\pi \log^{3}(N)}} 
    N^{-(1-\varepsilon)^2} =: -\frac12 \gamma(N, \varepsilon) N^{-(1-\varepsilon)^{2}}.
  \end{align*}
  Multiplying $G((1-\varepsilon)\bar \lambda)$ by $N-s$ yields
  \begin{align*}
    \left|\dee{}{u} R^{\sharp}(u\bar \lambda ; s, N)\right|_{u = 1-\varepsilon} %
    &\geq (N-s)\gamma(N,\varepsilon) %
      N^{-(1-\varepsilon)^2} -2s(1-\varepsilon)\sqrt{2 \log N}
    \\
    & = \gamma(N, \varepsilon) N^{2\varepsilon - \varepsilon^{2}} - s\gamma(N, \varepsilon)N^{-(1-\varepsilon)^{2}} - 2s(1-\varepsilon)\sqrt{2 \log N}
    \\
    &\geq C N^{\varepsilon}
  \end{align*}
  for some constant $C > 0$ under the condition that $N\geq N_{0}$, where
  $N_{0} > \exp\Big(\frac{3}{2}(1-\varepsilon)^{-2}\Big)$ is chosen so that for
  all $N \geq N_{0}$ the following two conditions are satisfied:
  \begin{align*}
    \begin{cases}
      (N-s)\gamma(N,\varepsilon) N^{-(1-\varepsilon)^2} %
      \geq 2s(1-\varepsilon)\sqrt{2 \log N} %
      \\[.25cm]
      \gamma(N,\varepsilon) \big(1 - \frac{s}{N}\big) %
      \geq 2s(1-\varepsilon)N^{-2\varepsilon+\varepsilon^{2}}\sqrt{2 \log N} + C
      N^{-\varepsilon + \varepsilon^{2}}
    \end{cases}
  \end{align*}
  In this regime, one achieves unbounded growth of the risk as a
  power law of the ambient dimension.
\end{proof}

\begin{remark}
  Using integration by parts, one has for $x > 0$,
  \begin{align*}
    \Phi(-x) %
    &= \int_{x}^{\infty} \phi(t) \d t %
    = \Big(\frac{1}{x} - \frac{1}{x^{3}}\Big)\phi(x) %
    + 3 \int_{x}^{\infty} \frac{t \phi(t)}{t^{5}} \d t
    \\
    &\leq \Big(\frac{1}{x} - \frac{1}{x^{3}}\Big)\phi(x) %
    + 3 x^{-5} \int_{x}^{\infty} t \phi(t) \d t %
    = \Big(\frac{1}{x} - \frac{1}{x^{3}} + \frac{3}{x^{5}}\Big)\phi(x) %
  \end{align*}

\end{remark}



\begin{proof}[Proof of Theorem \ref{thm:qppd-instability}]
  Define $f(u) := \dee{}{u} R^{\sharp}(u\bar\lambda; s, N)$ and
  $F(u) := R^{\sharp}(u\bar \lambda; s, N)$ its anti-derivative. The proof is an
  application of the fundamental theorem of calculus:
  \begin{align*}
    F(1) - F(1-\varepsilon) %
    = \int_{0}^{\varepsilon} f(1-t)\d t %
    \leq - C \int_{0}^{\varepsilon} N^{t} \d t %
    = C\frac{1 - N^{\varepsilon}}{\log N}.
  \end{align*}
  The result follows by substituting:
  \begin{align*}
    R^{\sharp}((1-\varepsilon)\bar \lambda; s, N) %
    \geq C \frac{N^{\varepsilon} - 1}{\log N} %
    + R^{\sharp}(\bar \lambda; s, N) %
    \geq C \frac{N^{\varepsilon}}{\log N}
  \end{align*}
  where the latter inequality holds after taking $N$ sufficiently large, and
  $C > 0$ is a universal constant that has changed values in the final
  expression. 
\end{proof}


\begin{proof}[Proof of {Proposition \ref{prop:asymptotic-equivalence}}]
  By Proposition \ref{prop:qppd-smooth},
  $R^{\sharp}(\lambda; s, N) = s ( 1+ \lambda^{2}) + 2(N-s) G(\lambda)$. We
  prove the result by controlling $G'(\lambda)$. One may lower bound
  $G'(\lambda)$ as
  \begin{align*}
    \frac{\mathrm{d}}{\mathrm{d}\lambda}G(\lambda) %
    = 2\lambda \Phi(-\lambda) - 2\phi(\lambda) %
    \geq 2\lambda\Big(\frac{\lambda}{\lambda^{2} + 1}\Big)\phi(\lambda)
    - 2\phi(\lambda) %
    = - 2 \frac{\phi(\lambda)}{\lambda^{2} + 1}.
  \end{align*}
  This gives the following lower bound for
  $\dee{}{\lambda}R^{\sharp}(\lambda; s, N)$:
  \begin{align*}
    \dee{R^{\sharp}}{\lambda}(\lambda; s, N) %
    \geq 2s\lambda - 4(N-s) \frac{\phi(\lambda)}{\lambda^{2} + 1}
    \geq 2\lambda - 4N\frac{\phi(\lambda)}{\lambda^{2} + 1} %
    = \frac{2}{\lambda^{2} + 1} \big( \lambda(\lambda^{2} + 1) %
    - 2N\phi(\lambda)\big).
  \end{align*}
  Substituting $\bar \lambda$ gives a positive quantity, since $N \geq 2$:
  \begin{align*}
    \frac{2}{2\log N + 1} \big(\sqrt{2\log N} (2\log N + 1) %
    - \frac{2}{\sqrt{2\pi}} \big) %
    > 0.
  \end{align*}
  Consequently, $\lambda^{*} < \bar \lambda$ because $\lambda^{*}$ is the value
  giving optimal risk and $\dee{}{\lambda} R^{\sharp}(\lambda; s, N)$ is
  increasing for all $\lambda \geq \bar \lambda$. Then it must be that
  $|\lambda^{*} - \bar \lambda| < \varepsilon$ for any $\varepsilon > 0$ when
  $N$ is sufficiently large. Indeed, fix $\varepsilon > 0$. By
  Lemma \ref{lem:qppd-instability} there exists $N_{0} \geq 1$ so that for all
  $N \geq N_{0}$ {\qppd} is parameter unstable for $\lambda < \bar \lambda$,
  yielding $R^{\sharp}(\lambda; s, N) \gtrsim N^{\varepsilon}$. But
  $R^{\sharp}(\lambda^{*}; s, N) \leq C R^{*}(s, N)$ for $N \geq N_{0}$ by
  Proposition \ref{prop:lspd-optimal-risk}, where we re-choose $N_{0} = N_{0}(s)$ if
  necessary. Thus, it must be that $|\lambda^{*} - \bar \lambda| < \varepsilon$
  for all $N \geq N_{0}$. In particular,
  $\lim_{N\to\infty} \bar \lambda / \lambda^{*} = 1$.
\end{proof}

\begin{remark}
  \label{rmk:slickness}
  One may derive the following lower bound using integration by parts.
  \begin{align*}
    \Phi(-\lambda) \geq \frac{\lambda}{\lambda^{2} + 1}\phi(\lambda)
  \end{align*}
  Let $Z \sim \mathcal{N}(0,1)$ be a standard normal random variable and let
  $S_{\lambda}(\cdot)$ denote soft-thresholding by $\lambda > 0$. Then,
  \begin{align*}
    0 \leq \E\big[ S_{\lambda}(Z)^{2}\big] %
    = 2 \int_{\lambda}^{\infty} (z - \lambda)^{2} \phi(\lambda) \d z %
    = 2(1+ \lambda^{2}) \Phi(-\lambda) - 2\lambda\phi(\lambda).
  \end{align*}
  Thus, $(1+ \lambda^{2}) \Phi(-\lambda) \geq \lambda\phi(\lambda)$, giving the
  desired lower bound.
\end{remark}

\subsubsection{Proof of {\qppd} right-sided stability}
\label{ssec:proof-qppd-stability}


We next prove right-sided stability of {\qppd}. 

\begin{proof}[Proof of {Theorem \ref{thm:qppd-rhs-stability}}]

  Given $L = \lambda / \lambda^{*} > 1$, define $\bar L = \bar L(s, N) > 0$ by
  $\lambda = \bar L \bar \lambda = \bar L \sqrt{2 \log N}$. Note
  $\lim_{N\to \infty} \bar{L}(s, N) = L$, because $\bar \lambda$ is
  asymptotically equivalent to $\lambda^{*}$ up to constants. A direct
  substitution of $\lambda = \bar L \bar \lambda = \bar L\sqrt{2\log N}$ in the
  analytic formula for $R^{\sharp}(\lambda; s, N)$ yields the desired bound,
  noting that $R^{\sharp}(\lambda^{*}; s, N)$ equals $R^{*}(s, N)$ up to
  constants. Thus, there is $C > 0$ and $N_{0} = N_{0}(s) \geq 2$ so
  that for all $N \geq N_{0}$
  \begin{align*}
    R^{\sharp}(\lambda; s, N) %
    \leq s ( 1 + 2\bar L^{2} \log N) + %
    \frac{N-s}{\bar L N^{\bar L^{2}}\sqrt{\pi\log N}} %
    \leq C L^{2} R^{*}(s, N). 
  \end{align*}
\end{proof}

\subsection{Proofs of {\bppd} results}
\label{ssec:proofs-bppd-results}

\subsubsection{Proof of underconstrained {\bppd} parameter instability}
\label{sec:proof-bppd-uc}



We prove parameter instability of {\bppd} in the underconstrained regime.

\begin{proof}[Proof of {Lemma \ref{lem:bppd-uc}}]

  By scaling, it suffices to consider the case where $\eta = 1$. Define the event
  \begin{align*}
    A_{N} %
    := \big\{ \|z\|_{2}^{2} \leq N - 2 \sqrt N \quad\And\quad %
    \|z\|_{\infty} \leq \sqrt{3 \log N} \big\}. 
  \end{align*}
  On $A_{N}$, it follows from the KKT conditions, where
  $h = \tilde x(\sigma) - x_{0}$, that
  \begin{align*}
    N %
    \leq \sigma^{2} = \|h\|_{2}^{2} - 2\ip{h, z} + \|z\|_{2}^{2} %
    \leq \|h\|_{2}^{2} - 2\ip{h,z} + N - 2\sqrt N 
  \end{align*}
  By Cauchy-Schwartz and definition of $A_{N}$,
  \begin{align*}
    \frac12 \|h\|_{2}^{2} %
    \geq \sqrt N +\ip{h, z} %
    \geq \sqrt N - \|h\|_{1} \|z\|_{\infty} %
    \geq \sqrt N - \|h\|_{1} \sqrt{3 \log N}. 
  \end{align*}
  Applying Proposition \ref{prop:descent-cone} and the binomial inequality
  $2ab \leq a^{2} + b^{2}$ gives
  \begin{align*}
    \sqrt N - \|h\|_{1} \sqrt{3 \log N} %
    \geq \sqrt N - 2\sqrt{s} \|h\|_{2} \sqrt{3 \log N} %
    \geq \sqrt N - \frac12 \|h\|_{2}^{2} %
    - 6s \log N
  \end{align*}
  Combining these two groups of inequalities gives
  $\|h\|_{2}^{2} \geq \sqrt N - 6s \log N$. Hence, by Bayes' rule and
  Corollary \ref{coro:large-deviation} there exist dimension independent constants
  $C, C' > 0$ such that
  \begin{align*}
    \E \|\tilde x (\sigma) - x_{0} \|_{2}^{2} %
    \geq \mathbb{P}(A_{N}) \cdot \E \big[ \|\tilde x(\sigma) - x_{0} \|_{2}^{2} %
    \mid A_{N} \big] %
    \geq C'  \big( \sqrt N - 6 s \log N \big) %
    \geq C\sqrt N. 
  \end{align*}
  The final inequality follows by the assumption that $N \geq N_{0}(s)$.
\end{proof}

\subsubsection{Supporting propositions for the geometric lemma}
\label{sssec:supporting-propositions-for-bppd-oc-3a}

This section is dedicated to several results necessary for the proof of
Lemma \ref{lem:bppd-oc-3a}, a main lemma in the proof of
Theorem \ref{thm:bppd-minimax} and Theorem \ref{thm:bppd-maximin}. We state and prove
these propositions in line.

\begin{proposition}
  \label{prop:bppd-oc-1a}
  Fix $C_{1} > 0$. Let $\alpha_{1} = a_{1}N^{1/4}$ and
  $\lambda = L\sqrt{\frac{N}{\log N}}$. Where
  $K_{1} := \lambda B_{1}^{N} \cap \alpha_{1}B_{2}^{N}$, there exists a choice
  of universal constants $a_{1} > 0$, $L \gg 1$ and
  $N_{0} = N_{0}^{\eqref{prop:bppd-oc-1a}}(a_{1}, C_{1}, L) \geq 1$ satisfying
  \begin{align*}
    N_{0}^{\eqref{prop:bppd-oc-1a}}(a_{1}, C_{1}, L) := D_{1}^{2 / (2 D_{2} - 1)}, \quad %
    D_{1} := \frac{a_{1}^{2}}{5L^{2}} < 1, %
    D_{2} := 2 \Big( \frac{C_{1} + a_{1}^{2}}{L^{2}}\Big)^{2} < \frac{1}{2}
  \end{align*}
  so that for all $N > N_{0}$
  \begin{align*}
    w (K_{1}) \geq (\frac{a_{1}^{2} + C_{1}}{2})\sqrt N.
  \end{align*}
\end{proposition}


\begin{proof}[Proof of {Proposition \ref{prop:bppd-oc-1a}}]
  Since $w(K_{1}) = \E_{z} \sup_{q \in K_{1}}\ip{q, z}$ is the Gaussian mean
  width of $K_{1}$, we may invoke Proposition \ref{prop:bellec2} to obtain a sufficient
  chain of inequalities:
  \begin{align*}
    \E \sup_{K_{1}}\ip{q, z} %
    = w(K_{1}) %
    \overset{\eqref{prop:bellec2}}{\geq} \frac{\sqrt 2}{4} \kappa \lambda \sqrt{
    \log \big( \frac{N \alpha_{1}^{2}}{5\lambda^{2}}\big)} %
    \overset{(*)}{\geq} \Big(\frac{\alpha_{1}^{2} + C_{1}}{2}\Big) \sqrt N.
  \end{align*}
  In particular, Proposition \ref{prop:bellec2} holds with $\kappa = 1$, since $\kappa$
  is the lower-RIP constant of the sensing matrix for {\bppd}, which is the
  identity. We thus turn our attention to $(*)$, which is equivalent to
  \begin{align*}
    \log\big( D_{1} \sqrt N \log N\big) \geq D_{2} \log N, \qquad %
    D_{1} := \frac{a_{1}^{2}}{5L^{2}}, %
    D_{2} := 2 \Big( \frac{C_{1} + a_{1}^{2}}{L^{2}}\Big)^{2}
  \end{align*}
  Rearranging gives
  \begin{align*}
    \frac12 + \frac{\log D_{1} + \log\log N}{\log N} \geq D_{2}
  \end{align*}
  and for $D_{1}, 2D_{2} \leq 1$, this is certainly satisfied for
    $N \geq D_{1}^{2 / (2 D_{2} - 1)}$ (\eg $L=11$ imposes $N\gtrsim 10^{5}$ when
    $a_{1} = 1, C_{1} = 2$). Accordingly, it suffices to choose
    $N_{0} = N_{0}(a_{1}, C_{1}, L)$ as in the proposition statement so that for
    all $N \geq N_{0}$, as desired,
  \begin{align*}
    w(K_{1}) \geq \Big(\frac{a_{1}^{2} + C_{1}}{2}\Big) \sqrt N.
  \end{align*}
\end{proof}

\begin{proposition}
  \label{prop:gmw-lb-ip}
  Fix $\delta > 0$, $c \in (0, 1)$. Let
  $K_{1} = \lambda B_{1}^{N} \cap \alpha_{1} B_{2}^{N}$ be as defined
  above. There are universal constants $\tilde{D}_{1} > 0, N_{0} \geq 2$ such
  that for
  \begin{align*}
    N > N_{0} := N_{0}^{\eqref{prop:gmw-lb-ip}}(c, \tilde D_{1}, \delta, L) %
    := \Big( \frac{1}{\tilde{D}_{1} L^{2} (1 - c)^{2}} \log\big(
    \frac{1}{\delta}\big)\Big)^{2}
  \end{align*}
  there exists $q \in K_{1}$ such that $\ip{q, z} \geq c w (K_{1})$ with
  probability at least $1 - \delta$, where $z \in \reals^{N}$ with
  $z_{i} \iid \mathcal{N}(0,1)$.
\end{proposition}


\begin{proof}[Proof of {Proposition \ref{prop:gmw-lb-ip}}]
  Note that $K_{1} \subseteq \reals^{N}$ is a topological space and define the
  centered Gaussian process $f_{x} := \ip{x, g}$ for
  $g_{i} \iid \mathcal{N}(0, 1)$. Observe that
  $\|f\|_{K_{1}} := \sup_{x \in K_{1}} |f_{x}|$ is almost surely finite. For any
  $u > 0$,
  \begin{align*}
    \mathbb{P} \big( %
        \sup_{x \in K_{1}}|\ip{x, g}| < w(K_{1}) - u \big) %
    \leq \exp\big(-\frac{u^{2}}{2 \sigma_{K_{1}}^{2}} \big). 
  \end{align*}
  by Theorem \ref{thm:borell-tis}. Therefore, for $c \in (0, 1)$,
  \begin{align*}
    \mathbb{P}\big( \sup_{x \in K_{1}}|\ip{x, g}| < c w (K_{1}) \big) %
    \leq \exp\big( %
    - \frac{(1-c)^{2} w^{2} (K_{1})}{2 \sigma_{K_{1}}^{2}} \big) %
    \leq \exp\big( %
    - \frac{(1-c)^{2} L^{2} \sqrt N \log( D_{1} \sqrt N \log N )}{16 \log N} \big) %
    \leq \delta
  \end{align*}
  because
  \begin{align*}
    \sigma_{K_{1}}^{2} %
    &= \sup_{x \in K_{1}} \E |\ip{x,g}|^{2} %
      = \sup_{x \in K_{1}} \sum_{i=1}^{N} x_{i}^{2}\E |g_{i}|^{2} %
      = \sup_{x \in K_{1}} \|x\|_{2}^{2} = \alpha_{1}^{2} = \sqrt N. 
  \end{align*}
  A specific choice of $q \in K_{1}$ follows by choosing the $q \in K_{1}$ that
  realizes the supremum, since $K_{1}$ is closed.
  
\end{proof}

\begin{proposition}
  \label{prop:bppd-oc-1b}
  Fix $C_{1}, \delta > 0$ and define the event
  $\mathcal{Z}_{-} := \{ \|z\|_{2}^{2} \leq N + C_{1}\sqrt N\}$ for
  $z \in \reals^{N}$ with $z_{i} \iid \mathcal{N}(0,1)$. There is a universal
  constant $N_{0} = N_{0}^{\eqref{prop:bppd-oc-1b}} \geq 1$ satisfying
  \begin{align*}
    N_{0}^{\eqref{prop:bppd-oc-1b}} %
    \geq \max \{N_{0}^{\eqref{prop:bppd-oc-1a}}(a_{1}, C_{1}, L), %
    N_{0}^{\eqref{prop:gmw-lb-ip}}(c, \tilde D_{1}, \delta, L)\big\},
  \end{align*}
  and a universal constant
  $k_{1} = k_{1}(N_{0}^{\eqref{prop:bppd-oc-1b}}, \delta) > 0$ so that for all
  $N \geq N_{0}$ there is an event $\mathcal{E} \subseteq \mathcal{Z}_{-}$
  satisfying
  \begin{align*}
    K_{1} \cap F \neq \emptyset \text{ on } \mathcal{E} %
    \qquad \text{and} \qquad 
    \mathbb{P}(\mathcal{E}) \geq \mathbb{P}(\mathcal{Z}_{-}) - \delta. 
  \end{align*}
\end{proposition}


\begin{proof}[Proof of {Proposition \ref{prop:bppd-oc-1b}}]

  By Proposition \ref{prop:gmw-lb-ip}, for any $c_{1} \in (0, 1) $ there is an event
  $\mathcal{E}_{1}$ that holds with high probability such that
  ${\sup_{q \in K_{1}} \ip{q, z} \geq c_{1} w(K_{1})}$ on
  $\mathcal{E}_{1}$. Subsequent statements are made on the restriction to
  $\mathcal{E}_{1}$.

  As $K_{1}$ is closed, there is $q \in K_{1}$ realizing the supremum, whence
  $\ip{q, z} \geq c_{1} w(K_{1})$. Now, choose $C_{1}' > 0$ such that
  $C_{1} \geq c_{1}^{-1}\big(a_{1}^{2} + C_{1}\big) - a_{1}^{2}$. Then
  $q \in K_{1}$ satisfies
  \begin{align*}
    \ip{q, z} \geq c_{1} w(K_{1}) \geq c_{1} \Big( \frac{a_{1}^{2} + C_{1}'}{2}\Big) \sqrt N \geq \Big( \frac{a_{1}^{2} + C_{1}}{2}\Big) \sqrt N. 
  \end{align*}
  Now, because $\|q\|_{2} \leq \alpha_{1}$ and $q \in K_{1}$, it holds on the
  event $\mathcal{E}_{1} \cap \mathcal{Z}_{-}$ that
  \begin{align*}
    \Big( \frac{a_{1}^{2} + C_{1}}{2}\Big) \sqrt N %
    \geq \frac12 \|q\|_{2}^{2} + \frac12 \big( \|z\|_{2}^{2} - N\big)
  \end{align*}
  Combining the two previous chains of inequalities implies that
  \begin{align*}
    \|q - z\|_{2}^{2} \leq N
  \end{align*}
  Namely, there exists an event $\mathcal{Z}_{-} \cap \mathcal{E}_{1}$, such
  that $q \in K_{1} \cap F$, so long as $N \geq
  N_{0}^{\eqref{prop:bppd-oc-1b}}$. Because $\mathcal{E}_{1}$ holds with high
  probability and the probability of $\mathcal{Z}_{-}$ is lower-bounded by a
  universal constant, Proposition \ref{prop:bppd-oc-const-whp-prob} implies
  $\mathbb{P}(\mathcal{Z}_{-} \cap \mathcal{E}_{1}) \geq
  k_{1}(N_{0}^{\eqref{prop:bppd-oc-1b}}, \delta)$ for
  $N \geq N_{0}^{\eqref{prop:bppd-oc-1b}}$, where
  \begin{align*}
    N_{0}^{\eqref{prop:bppd-oc-1b}} %
    \geq \max \{N_{0}^{\eqref{prop:bppd-oc-1a}}(a_{1}, C_{1}, L), %
    N_{0}^{\eqref{prop:gmw-lb-ip}}(c, \tilde D_{1}, \delta, L)\big\}. 
  \end{align*}
  %
  %
  %

\end{proof}

\begin{proposition}
  \label{prop:bppd-oc-2a}
  Fix $C_{2} > 0$ and let $L \geq 1$. Set
  $K_{2} := \lambda B_{1}^{N} \cap \alpha_{2}B_{2}^{N}$, where
  $\lambda = L\sqrt{\frac{N}{\log N}}$. There is a maximal choice of
  $\alpha_{2} = \alpha_{2}(N) > 0$ so that for all $N \geq 1$,
  \begin{align*}
    w(K_{2}) \leq \frac{C_{2}}{2} \sqrt N
  \end{align*}

\end{proposition}



\begin{proof}[Proof of {Proposition \ref{prop:bppd-oc-2a}}]

  Since $w(K_{2}) = \E_{z} \sup_{q \in K_{2}} \ip{q, z}$ is the Gaussian mean
  width of $K_{2}$, we may invoke Proposition \ref{prop:bellec1} to obtain a sufficient
  chain of inequalities:
  \begin{align*}
    w(K_{2}) %
    \overset{\eqref{prop:bellec1}}{\leq} 4\lambda \sqrt{ \log
    (\frac{4eN\alpha_{2}^{2}}{\lambda^{2}})} %
    \overset{(**)}{\leq} \frac{C_{2}}{2} \sqrt N. 
  \end{align*}
  The first inequality follows by \eqref{prop:bellec1} immediately. Rearranging
  and substituting for $\lambda$, $(**)$ is equivalent to
  \begin{align*}
    D_{3}\log N \geq \log\big( D_{4}\alpha_{2}^{2} \log N\big), %
    \qquad D_{3} := \Big(\frac{C_{2}}{8L}\Big)^{2}, D_{4} := \frac{4e}{L^{2}}. 
  \end{align*}
  This inequality is satisfied for any $\alpha_{2}$ with
  \begin{align*}
    \alpha_{2}^{2} \leq \frac{N^{D_{3}}}{D_{4} \log N} =: A^{2}(C_{2}, N)
  \end{align*}
  For example, one may choose
  \begin{align*}
    \alpha_{2} = \frac{L N^{D_{5}}}{2 \sqrt{e \log N}}, \qquad %
    D_{5} := \frac{C_{2}^{2}}{32 L^{2}}. 
  \end{align*}
  For such $0 < \alpha_{2} \leq A(N; C_{2}, L)$, it holds as desired that
  $w(K_{2}) \leq \frac{C_{2}}{2} \sqrt N $.
  
\end{proof}

\begin{remark}
  Notice that we want to choose $N_{0}$ so that $A(C_{2}, N)$ is increasing for
  all $N \geq N_{0}$. A quick calculation reveals that
  $N_{0} = N_{0}^{\eqref{prop:bppd-oc-2a}}(C_{2}, L) := \exp\big( (2 D_{5})^{-1}
  \big)$ is sufficient.
\end{remark}

\begin{proposition}
  \label{prop:gmw-ub-ip}
  Fix $\delta > 0$, and $C > 1$. Let
  $K_{2} = \lambda B_{1}^{N} \cap \alpha_{2} B_{2}^{N}$ as above. There are
  universal constants $\tilde{D}_{2} > 0, N_{0} \geq 1$ such that for
  \begin{align*}
    N %
    > N_{0} := N_{0}^{\eqref{prop:gmw-ub-ip}}(C, \tilde D_{2}, \delta, L) %
    := \Big( \frac{1}{\tilde{D}_{2} L^{2} (C - 1)^{2}} %
    \log \big( \frac{1}{\delta}\big)\Big)^{2}
  \end{align*}
  one has $\sup_{q \in K_{2}} \ip{q,z} \leq C w(K_{2})$ with probability at
  least $ 1 - \delta$, where $z \in \reals^{N}$ with
  $z_{i} \iid \mathcal{N}(0,1)$.
\end{proposition}


\begin{proof}[Proof of {Proposition \ref{prop:gmw-ub-ip}}]
  Define the centered Gaussian process $f_{x} := \ip{x,g}$ for
  $x \in K_{2}\subseteq \reals^{N}$, a topological space, and where
  $g_{i} \iid \mathcal{N}(0,1)$. Observe
  $\|f\|_{K_{2}} = \sup_{x \in K_{2}}|f_{x}| < \infty$ almost surely. For any $u > 0$,
  \begin{align*}
    \mathbb{P}\big( %
    \sup_{x \in K_{2}} |\ip{x, g}| > w(K_{2}) + u \big) %
    \leq \exp\Big( - \frac{u^{2}}{2 \sigma_{K_{2}}^{2}} \Big)
  \end{align*}
  by Theorem \ref{thm:borell-tis}. Hence, for $C > 1$,
  \begin{align*}
    \mathbb{P}\big( \sup_{x \in K_{2}} | \ip{x, g}| > C w(K_{2}) \big) %
    \leq \exp\Big( - \frac{(C-1)^{2} w^{2}(K_{2})}{2 \sigma_{K_{2}}^{2}} \Big) %
    \leq \exp\Big( - \frac{(C-1)^{2} L^{2} N \log(D_{1} \sqrt N \log N)}{%
    16 \alpha_{2}^{2} \log N }\Big) %
    \leq \delta
  \end{align*}
  because
  \begin{align*}
    \sigma_{K_{2}}^{2} %
    = \sup_{x\in K_{2}} \E |\ip{x, g}|^{2} %
    = \sup_{x \in K_{2}}\sum_{i=1}^{N} x_{i}\E|g_{i}|^{2} %
    = \sup_{x\in K_{2}}\|x\|_{2}^{2} %
    = \alpha_{2}^{2} %
    \leq \alpha_{1}^{2} %
    = \sqrt N.
  \end{align*}
  Finally, for $\delta > 0$ and $C > 1$,
  $\sup_{x \in K_{2}} | \ip{x, g}| \leq C w(K_{2})$ with probability at least
  $1 - \delta$ for any $N \geq N_{0}^{\eqref{prop:gmw-ub-ip}}$.

\end{proof}

\begin{proposition}
  \label{prop:bppd-oc-2b}
  Fix $C_{2}, \delta > 0$ and define the event
  $\mathcal{Z}_{+} := \{ \|z\|_{2}^{2} \geq N + C_{2} \sqrt N \} $ where
  $z \in \reals^{N}$ with $z_{i} \iid \mathcal{N}(0,1)$. There is a universal
  constant $N_{0} := N_{0}^{\eqref{prop:bppd-oc-2b}} \geq 1$ satisfying
  \begin{align*}
    N_{0}^{\eqref{prop:bppd-oc-2b}} \geq \max \big\{ %
    N_{0}^{\eqref{prop:bppd-oc-2a}}, N_{0}^{\eqref{prop:gmw-ub-ip}} \big\}. 
  \end{align*}
  and a universal constant $k_{2} = k_{2}(N_{0}, \delta) > 0$ so that for all
  $N \geq N_{0}$ there is an event $\mathcal{E}\subseteq \mathcal{Z}_{+}$
  satisfying
  \begin{align*}
    K_{2} \cap F = \emptyset \text{ on } \mathcal{E} %
    \qquad \text{and} \qquad %
    \mathbb{P}(\mathcal{E}) \geq k_{2} := \mathbb{P} (\mathcal{Z}_{+}) - \delta. 
  \end{align*}
\end{proposition}


\begin{proof}[Proof of {Proposition \ref{prop:bppd-oc-2b}}]
  By Proposition \ref{prop:gmw-ub-ip}, for any $0 < c_{2} < 1$ there is an event
  $\mathcal{E}_{2}$ that holds with high probability such that
  $\sup_{q \in K_{2}}\ip{q, z} \leq c_{2} w(K_{2})$ on
  $\mathcal{E}_{2}$. Because $K_{2}$ is closed, there is $q \in K_{2}$ realizing
  the supremum when restricted to $\mathcal{E}_{2}$, whence
  \begin{align*}
    {\ip{q, z} \leq \sup_{q' \in K_{2}} \ip{q', z} \leq c_{2} w(K_{2})}.
  \end{align*}
  Now, choose $C_{2}' > 0$ such that $ 0 \leq C_{2} \leq c_{2}C_{2}'$. Then
  $q \in K_{2}$ satisfies
  \begin{align*}
    \ip{q, z} %
    \leq c_{2} w(K_{2}) %
    \leq c_{2}\frac{C_{2}'}{2} \sqrt N %
    \leq \frac{C_{2}}{2} \sqrt N
  \end{align*}
  On the other hand, for any $q' \in F$ on the event $\mathcal{Z}_{+}$,
  \begin{align*}
    C_{2}\sqrt N \leq \|q'\|_{2}^{2} + \|z\|_{2}^{2} - N \leq 2\ip{q', z}
  \end{align*}
  whence $K_{2} \cap F = \emptyset$ on the event
  $\mathcal{Z}_{+} \cap \mathcal{E}_{2}$. Because $\mathcal{E}_{2}$ holds with
  high probability and the probability of $\mathcal{Z}_{+}$ is lower-bounded by
  a universal constant, Proposition \ref{prop:bppd-oc-const-whp-prob} implies
  $\mathbb{P}\big(\mathcal{Z}_{+} \cap \mathcal{E}_{2}\big) >
  k_{2}(N_{0}^{\eqref{prop:bppd-oc-2b}}, \delta)$ for
  $N \geq N_{0}^{\eqref{prop:bppd-oc-1b}}$ where
  \begin{align*}
    N_{0}^{\eqref{prop:bppd-oc-2b}} \geq \max \big\{ %
    N_{0}^{\eqref{prop:bppd-oc-2a}}, N_{0}^{\eqref{prop:gmw-ub-ip}} \big\}. 
  \end{align*}
\end{proof}

\subsubsection{Proof of the geometric lemma}
\label{sec:proof-bppd-oc-3a}

We now have the tools required for Lemma \ref{lem:bppd-oc-3a}. For intuiton of the
result, we refer the reader to Figure \ref{fig:bppd-oc-3a} in
\ref{ssec:overconstrained-bppd}.

  

\begin{proof}[Proof of {Lemma \ref{lem:bppd-oc-3a}}]
  The proof of the first two items follows trivially from
  Proposition \ref{prop:bppd-oc-1b} and Proposition \ref{prop:bppd-oc-2b}. Define the event
  \begin{align*}
  \mathcal{E} := \mathcal{Z}_{-} \cap \mathcal{E}_{1} \cap \mathcal{Z}_{+} \cap
  \mathcal{E}_{2}
  \end{align*}
  To prove the final item, observe that
  $\mathbb{P}\big(\mathcal{E}\big) \geq \mathbb{P}\big(\mathcal{Z}_{-} \cap
  \mathcal{Z}_{+}\big) - 2 \delta \geq k_{3}$ for all sufficiently large
  $N$. This is a direct consequence of Proposition \ref{prop:bppd-oc-const-prob} and
  Proposition \ref{prop:bppd-oc-const-whp-prob}.

  The proof of the third item follows from a note in
  Proposition \ref{prop:bppd-oc-2a}. Specifically, the result holds for any choice of
  $\alpha_{2}$ satisfying
  \begin{align*}
    0 < \alpha_{2} \leq A(N; C_{2}; L) = \frac{LN^{D_{5}}}{2 \sqrt{e \log N }}, %
    \quad D_{5} := \frac{C_{2}^{2}}{32 L^{2}}
  \end{align*}
  Hence, choose $C_{3}, q > 0$ so that $\alpha_{2} > C_{3}N^{q}$ for all
  $N \geq N_{0}^{\eqref{lem:bppd-oc-3a}} \geq N_{0}^{\eqref{prop:bppd-oc-2a}}$.
\end{proof}

\subsubsection{Proofs for overconstrained suboptimality}
\label{sec:proof-bppd-oc-subopt-1}

First we prove a key ingredient in the main results for
$\tilde R(\sigma; x_{0}, N, \eta)$ parameter instability. Then, we prove the
lemma that extends {\bppd} parameter instability from $\sigma = \sqrt N$ and
$x_{0} \equiv 0$ to $\sigma \leq \sqrt N$ and $x_{0} \equiv 0$. Finally, we
prove the restricted maximin result, yielding parameter instability for
overconstrained {\bppd}.
 

\begin{proof}[Proof of {Corollary \ref{coro:bppd-oc-3a}}]
  Restrict to the event $\mathcal{E}$ as given in the lemma and assume that
  $N \geq N_{0}^{\eqref{lem:bppd-oc-3a}}$.  $K_{1} \cap F$ is non-empty, so
  $\tilde x(\sigma) \in K_{1} \cap F$ by definition.  $K_{2} \cap F = \emptyset$
  thereby implies
  \begin{align*}
    \tilde x(\sigma) %
    \in \lambda B_{1}^{N} \cap %
    \big(\alpha_{1}B_{2}^{N}\setminus \alpha_{2}B_{2}^{N}\big) %
    \cap F %
    = (K_{1} \setminus K_{2}) \cap F. 
  \end{align*}
  Whence follows $\|\tilde x(\sigma)\|_{1} \leq \lambda$ and
  $\alpha_{2} \leq \|\tilde x(\sigma)\|_{2} \leq \alpha_{1}$. Applying Bayes'
  rule to the noise-normalized risk yields:
  \begin{align*}
    \tilde R(\sigma; 0, N, \eta) %
    \geq \frac{\mathbb{P}(\mathcal{E})}{\eta^{2}} %
    \E \big[ \|\tilde x(\sigma)\|_{2}^{2} \mid \mathcal{E} \big] %
    \geq k_{3} C_{3} N^{q} %
    =: C N^{q}. 
  \end{align*}
\end{proof}


\begin{proof}[Proof of {Lemma \ref{lem:bppd-oc-sigma-lt-sqrtN}}]
  This result is an immediate consequence of Corollary \ref{coro:projection-lemma}. 
  

\end{proof}


\begin{proof}[Proof of {Theorem \ref{thm:bppd-oc-maximin}}]
  Without loss of generality, assume $\eta = 1$. We may trivially lower-bound
  the minimax expression by considering only the case where $x_{0} \equiv 0$,
  \begin{align*}
    \sup_{x_{0}\in \Sigma_{s}^{N}} \inf_{\sigma \leq \sqrt N} %
    \tilde R(\sigma; x_{0}, N, 1)
    \geq
    \inf_{\sigma \leq \sqrt N} \tilde R(\sigma; 0, N, 1)
  \end{align*}
  Lemma \ref{lem:bppd-oc-sigma-lt-sqrtN} and Corollary \ref{coro:bppd-oc-3a} imply in
  turn,
  \begin{align*}
    \inf_{\sigma \leq \sqrt N } \tilde R(\sigma; 0, N, \eta) %
    \geq \tilde R(\sqrt N; 0, N, \eta) %
    \geq C N^{q}
  \end{align*}
  for all $N \geq N_{0}$, where $N_{0} \geq N_{0}^{\eqref{lem:bppd-oc-3a}}$ and
  $C, q > 0$ are chosen according to Lemma \ref{lem:bppd-oc-3a}.
\end{proof}

\subsubsection{Proof of minimax suboptimality}
\label{sssec:proof-minimax-subopt}



We prove that {\bppd} is minimax suboptimal. 

\begin{proof}[Proof of {Theorem \ref{thm:bppd-minimax}}]
  Without loss of generality, take $\eta = 1$. Observe that
  \begin{align*}
    \inf_{\sigma > 0} \sup_{x_{0} \in \Sigma_{s}^{N}}
    \tilde R(\sigma; x_{0}, N, 1) %
    = \min \Big\{ \inf_{\sigma \leq \sqrt N} S(\sigma) ,
    \inf_{\sigma > \sqrt N} S(\sigma) \Big\}
  \end{align*}
  where
  $S(\sigma) := \sup_{x_{0} \in \Sigma_{s}^{N}} \tilde R(\sigma; x_{0}, N,
  1)$. Next, assume $N \geq N_{0}^{\eqref{lem:bppd-oc-3a}}$. Then one has
  $\inf_{\sigma > \sqrt N} S(\sigma) \geq C_{1} \sqrt N$ by
  Lemma \ref{lem:bppd-uc}. Moreover, a trivial lower bound,
  Lemma \ref{lem:bppd-oc-sigma-lt-sqrtN} and Corollary \ref{coro:bppd-oc-3a}
  successively imply
  \begin{align*}
    \inf_{\sigma \leq \sqrt N} S(\sigma) %
    &\geq \inf_{\sigma \leq \sqrt N} \tilde R(\sigma; 0, N, 1) %
      \geq \tilde R(\sqrt N; 0, N, 1) \geq C_{2} N^{q}.
  \end{align*}
  In particular, there is a universal constant $C > 0$ so that
  \begin{align*}
    \inf_{\sigma > 0} \sup_{x_{0} \in \Sigma_{s}^{N}}
    \tilde R(\sigma; x_{0}, N, 1) %
    \geq \min \{ C_{2} N^{q}, C_{1} \sqrt N \} \geq C N^{q}.
  \end{align*}
  
\end{proof}

\subsubsection{Proof of maximin suboptimality}
\label{sssec:proof-maximin-subopt}


We prove that {\bppd} is maximin suboptimal. 

\begin{proof}[Proof of {Lemma \ref{lem:bppd-oc-sge1}}]
  The proof is completed by the following chain of inequalities. The first and
  last equalities are by definition of the {\bppd} estimator. The first
  inequality follows by relaxing the objective; the second inequality follows by
  relaxing the constraint condition.
  \begin{align*}
    \|\tilde x_{T^{C}}\|_{2} %
    & = \big\| \argmin\{ \|x\|_{1} : \|y - x\|_{2}^{2} \leq \sigma^{2} \}_{T^{C}} \big\|_{2}
    \\
    & \geq \big\| \argmin \{ \|x_{T^{C}}\|_{1} : %
      \|y - x \|_{2}^{2} \leq \sigma^{2}\}_{T^{C}} \big\|_{2} %
    \\
    &\geq \big\| \argmin \{ \|x_{T^{C}}\|_{1} : %
      \|(y - x)_{T^{C}}\|_{2}^{2} \leq \sigma^{2} \}_{T^{C}} \big\|_{2}
    \\
    &\equiv \|\tilde x' \|_{2}
  \end{align*}
\end{proof}


\begin{proof}[Proof of {Theorem \ref{thm:bppd-maximin}}]
  We may trivially lower-bound the maximin expression by considering the case
  where $x_{0} := N e_{1}$ where $e_{1}$ is the first standard basis
  vector. Without loss of generality, we may assume that this entry is in the
  first coordinate, and is at least $N$. Again without loss of generality, it
  suffices to consider the case where $\eta = 1$. We write the lower bound as
  \begin{align*}
    \sup_{x \in \Sigma_{s}^{N}} \inf_{\sigma > 0} %
    \tilde R(\sigma; x, N, 1) %
    \geq \inf_{\sigma > 0} \tilde R(\sigma; x_{0}, N, 1). 
  \end{align*}
  If $\sigma \geq \sqrt N$, then the result follows by
  Lemma \ref{lem:bppd-uc}. Otherwise, it must be that $\sigma \leq \sqrt N$, in
  which case the result follows immediately by Lemma \ref{lem:bppd-oc-sge1}. In
  this latter case, we have implicitly assumed that if
  $\sigma \in (\sqrt{N-1}, \sqrt N)$, then the omitted technical exercise of
  adjusting constants in Corollary \ref{coro:bppd-oc-3a} and its constituents has been
  carried out. For further detail on this caveat, see the remark immediately
  succeeding Corollary \ref{coro:bppd-oc-sge1}.
\end{proof}


\section{Conclusions}
\label{sec:conclusions}


We have illustrated regimes in which each program is unstable. The theory of
section \ref{sec:param-inst-lspd}, section \ref{sec:param-inst-qppd} and
section \ref{sec:param-inst-bppd} proves asymptotic results for each program, while
the numerics of section \ref{sec:numerical-results} supports using the asymptotic
behaviour as a basis for practical intuition. Thus, we hope these results inform
practitioners about which program to use.

In section \ref{sec:param-inst-lspd} and \ref{ssec:lspd-numerics} we observe
that {\lspd} exhibits parameter instability in the low-noise regime. The risk
$\hat R(\tau; x_{0}, N,\eta)$ develops an asymptotic singularity as
$\eta \to 0$, blowing up for any $\tau \neq \|x_{0}\|_{1}$, where
$\hat R(\|x_{0}\|_{1}; x_{0}, N,\eta)$ attains minimax order-optimal
error. Numerical simulations support that \newline$\hat R(\tau; x_{0}, N, \eta)$
develops cusp-like behaviour in the low-noise regime, which agrees with the
asymptotic singularity of Theorem \ref{thm:constr-pd}. Notably, {\lspd} parameter
instability manifests in very low dimensions relative to practical problem
sizes. Outside of the low-noise regime, {\lspd} appears to exhibit better
parameter stability, as exemplified in Figure \ref{fig:parameter-stability}.

In section \ref{sec:param-inst-qppd} and section \ref{sec:qppd-numerics} we observe that
{\qppd} exhibits left-sided parameter instability in the low-noise regime.  When
$\lambda < \bar \lambda$ we prove that $R^{\sharp}(\lambda; s, N)$ scales
asymptotically as a power law of $N$. The suboptimal scaling of the risk
manifests in relatively higher dimensional problems, as suggested by
Figure \ref{fig:qppd-instability-a}. Minimax order-optimal scaling of the risk when
$\lambda \geq \bar \lambda$ is clear from Figure \ref{fig:qppd-instability-b}. The
numerics of section \ref{sec:numerical-results} support that {\qppd} is generally
the most stable of the three programs considered.

In section \ref{sec:param-inst-bppd} and \ref{ssec:bppd-numerics} we observe
that {\bppd} exhibits parameter instability in the very sparse regime. Notably,
$\tilde R(\sigma; x_{0}, N, \eta)$ is maximin suboptimal for \emph{any} choice
of $\sigma > 0$ for $s / N$ sufficiently small. This behaviour is supported by
Figure \ref{fig:bppd-numerics-a}, in which the best average loss of {\bppd} is
a $82.2$ times worse than that for {\lspd} and {\qppd}.  Further, the average
loss for {\bppd} exhibits a clear cusp-like behaviour in Figure
\ref{fig:bppd-numerics-a}, like for that of {\lspd}, which would be an
interesting object of further study. Outside of the very sparse regime, {\bppd}
appears to exhibit parameter stability, as exemplified in Figure
\ref{fig:parameter-stability}.

In section \ref{sec:realistic-denoising} we portray how estimators behave as a
function of the normalized parameter for each program. We show the kinds of
pathologies from which these estimators suffer in unstable regimes, and
demonstrate that estimators for compressed sensing problems can exhibit similar
pathologies (section \ref{sec:lasso-example}). These simulations support the
intuition that our theory may be extended to the compressed sensing setting.

Finally, we demonstrated the usefulness of Lemma \ref{lem:projection-lemma}. By
this result, the size of $\tilde x(\eta \sqrt N)$ controls the size of
$\tilde x (\sigma)$ for $\sigma \leq \eta\sqrt N$ when $x_{0} \equiv 0$. This
was key to demonstrating risk suboptimality for underconstrained
{\bppd}. Moreover, Lemma \ref{lem:projection-lemma} was used to prove
$\hat R(\tau; \tau x_{0}, N, \eta)$ is an increasing function of $\tau$ when
$\|x_{0}\|_{1} = 1$. Thus, the projection lemma was particularly effective for
proving minimax order-optimality of $R^{*}(s, N)$.

Future works include extending the main results to the CS set-up and to more
general atomic norms. These results may also extend to ones under more general
noise models. Some of these extensions are in preparation by the
authors. Lastly, it would be interesting to see what role parameter instability
might play in proximal point algorithms and those algorithms relying on proximal
operators. Conversely, it would be useful to understand rigorously when a PD
program exhibits parameter instability, and to determine systematically the
regime in which that instability arises.

\section*{Funding}
\label{sec:funding}

This work was supported by the Natural Sciences and Engineering Research Council of Canada (NSERC) [CGSD3-489677 to A.B., 22R23068 to Y.P., 22R82411 to O.Y., 22R68054 to O.Y.]; and the Pacific Institute for the Mathematical Sciences (PIMS) [CRG 33: HDDA to Y.P., CRG 33: HDDA to O.Y.]. 


\section*{Acknowledgements}
\label{sec:acknowledgements}

We would like to thank Dr. Navid Ghadermarzy for a careful reading of the manuscript.

\bibliographystyle{plain}
\bibliography{references}

\end{document}